\documentclass[10pt]{article}

\usepackage{pdfpages}
\usepackage[english]{babel}
\usepackage{hyperref}
\usepackage{bookmark} % fixes "document outline"
\usepackage[utf8]{inputenc}

\usepackage{mathptmx}
\usepackage{nimbusmono}
\usepackage[T1]{fontenc}

\usepackage[
  a4paper,
  top=105pt,
  left= 95pt,
  right=95pt,
]{geometry}

\usepackage{titling}
\pretitle{
\begin{center}\fontsize{12pt}{14pt}\bfseries}
  \posttitle{\par
\end{center}}

\usepackage{titlesec}
\titleformat{\section}{\normalfont\normalsize\bfseries}{\thesection.}{0.5em}{}
\titleformat{\subsection}{\normalfont\normalsize}{\thesubsection.}{0.5em}{\itshape}
\titleformat{\subsubsection}{\normalfont\normalsize}{\thesubsubsection.}{0.5em}{\itshape}

\titlespacing{\section}{0pt}{2em}{1em}
\titlespacing{\subsection}{0pt}{1em}{0.5em}
\titlespacing{\subsubsection}{0pt}{1em}{0.5em}

\usepackage{fancyhdr}

\usepackage[font=footnotesize]{caption}
\usepackage{parskip}
\renewenvironment{abstract}
{\par\noindent\ignorespaces
  \begin{center}
  \begin{minipage}{0.90\textwidth}\fontsize{8pt}{10pt}}
    {
    \end{minipage}
\end{center}\par\medskip}

\usepackage{float}
\usepackage{url}
\usepackage{nicefrac}       % compact symbols for 1/2, etc.
\usepackage{microtype}      % microtypography
\usepackage{graphicx}
\usepackage{doi}

\usepackage[
  backend=biber,
  style=authoryear,
  natbib=true,
  maxcitenames=2, % Abbreviate author list if more than 2 authors
  maxbibnames=99, % Show all authors in the bibliography
  giveninits=true, % Use initials for given names
  terseinits=true % Add a period after initials
  uniquename=false, % Disable disambiguation of author names
]{biblatex}
\bibliography{sources}

\defbibheading{myheadingmain}[\refname]{%
  \centering\scshape #1%
}

\defbibheading{myheadingsupplementary}[\refname]{%
  \centering\scshape Supplementary #1%
}

\title{Degrees-of-freedom penalized piecewise regression}

\date{} % We manually add the date later on

\author{} % We manually add the authors later on
\newcommand{\authorinfo}[4]{%
  \begin{center}
    {
      \textsc{#1} \\
      \emph{#2, #3} \\
      \emph{#4}
    }
  \end{center}
}

\definecolor{linkclr}{cmyk}{0.973,0.957,0,0.04}

\hypersetup{
  colorlinks,
  breaklinks=true,
  plainpages=false,
  citecolor=linkclr,
  linkcolor=linkclr,
  urlcolor=linkclr,
  bookmarksopen=true,
  bookmarksnumbered=true,
  bookmarksdepth=8,
  unicode=true,
  pdfpagelayout={SinglePage},
  pdffitwindow=true,
  pdftitle={Degrees-of-freedom penalized piecewise regression},
  pdfauthor={Stefan Volz and Martin Storath and Andreas Weinmann},
  pdfsubject={math.NA},
  pdfkeywords={piecewise regression, fast algorithm, penalized least squares, nonconvex optimization, model selection, dynamic programming, signal estimation, signal processing}
}

\usepackage{changes}
\definechangesauthor[name=MS, color=orange]{MS}
\definechangesauthor[name=SV, color=mpl-blue]{SV}
\definechangesauthor[name=SVp, color=mpl-green]{SVp}
\definechangesauthor[name=SVm, color=mpl-red]{SVm}

\def\svadd#1{{\color{black}#1}}

\usepackage{xcolor}
\definecolor{mpl-blue}{HTML}{1f77b4}
\definecolor{mpl-orange}{HTML}{fba402}
\definecolor{mpl-green}{HTML}{2ca02c}
\definecolor{mpl-red}{HTML}{d62728}
\definecolor{mpl-purple}{HTML}{9467bd}
\definecolor{mpl-brown}{HTML}{8c564b}
\definecolor{mpl-pink}{HTML}{e377c2}
\definecolor{mpl-grey}{HTML}{7f7f7f}
\definecolor{mpl-gray}{HTML}{7f7f7f}
\definecolor{mpl-lime}{HTML}{bcbd22}
\definecolor{mpl-cyan}{HTML}{17becf}

\usepackage[breakable]{tcolorbox}
\usepackage{csquotes}

\usepackage{subcaption}

\usepackage{float}

\usepackage{footmisc} % footnote references

\usepackage{amsmath}
\usepackage{amssymb}
\usepackage{amsfonts}
\usepackage{amsthm}
\usepackage{mathtools}

\usepackage{mathrsfs}
\usepackage{tikz}
\usetikzlibrary{cd}
\usetikzlibrary{babel}
\usetikzlibrary{decorations.pathreplacing}

\usepackage{enumerate}
\usepackage[loadonly]{enumitem}

\usepackage{booktabs}
\usepackage{multirow}
\usepackage[flushleft]{threeparttable}

\usepackage{listings}
\definecolor{backcolour}{rgb}{1,1,1}

\usepackage{algorithm}% http://ctan.org/pkg/algorithms
\usepackage{algpseudocode}% http://ctan.org/pkg/algorithmicx

\algnewcommand\algorithmicmatch{\textbf{match}}
\algnewcommand\algorithmiccase{\textbf{case}}
\algdef{SE}[MATCH]{Match}{EndMatch}[1]{\algorithmicmatch\ #1\ \algorithmicdo}{\algorithmicend\ \algorithmicmatch}%
\algdef{SE}[CASE]{Case}{EndCase}[1]{\algorithmiccase\ #1}{\algorithmicend\ \algorithmiccase}%
\algtext*{EndMatch}%
\algtext*{EndCase}%

\algdef{SE}[DOWHILE]{Do}{EndDoWhile}{\algorithmicdo}[1]{\algorithmicwhile\ #1}%

\algblock{Input}{EndInput}
\algnotext{EndInput}
\algblock{Output}{EndOutput}
\algnotext{EndOutput}

\algblock{Description}{EndDescription}
\algnotext{EndDescription}

\usepackage{varwidth} % http://ctan.org/pkg/varwidth

\lstdefinestyle{mystyle}{
  backgroundcolor=\color{backcolour},
  commentstyle=\color{mpl-green},
  keywordstyle=\color{mpl-cyan},
  numberstyle=\tiny\color{mpl-gray},
  stringstyle=\color{mpl-green},
  basicstyle=\ttfamily\scriptsize,
  breakatwhitespace=false,
  breaklines=true,
  captionpos=b,
  keepspaces=true,
  numbers=left,
  numbersep=5pt,
  showspaces=false,
  showstringspaces=false,
  showtabs=false,
  tabsize=2
}

\usepackage{wrapfig}

\newtheorem{theorem}{Theorem}
\newtheorem{lemma}[theorem]{Lemma}

\newtheorem{example}[theorem]{Example}
\newtheorem{corollary}[theorem]{Corollary}
\theoremstyle{definition}
\newtheorem{assumption}{Assumption}
\theoremstyle{remark}
\newtheorem{remark}{Remark}

\def\R{\ensuremath{\mathbb{R}}}

\def\N{\ensuremath{\mathbb{N}}}

\def\OPart{\operatorname{OPart}}
\def\Seg{\operatorname{Seg}}

\def\Norm#1{\lVert#1\rVert}

\undef{\det}
\def\det{\operatorname{det}\,}

\usepackage{xspace}

\def\st{s.\,t.\xspace}
\def\resp{resp.\xspace}

\def\ie{i.\,e.\xspace}

\def\GHz{GHz\xspace}

\def\argmin{\operatornamewithlimits{arg\, min}}

\def\CV{\operatorname{CV}}

\def\O{\mathcal{O}}

\def\Ortho{\mathbf{O}}
\def\GL{\mathbf{GL}}

\def\sgn{\operatorname{sgn}}

\def\dof{\operatorname{dof}}
\def\ElemCount#1{\# #1}

\def\sUnique{RMGT-unique\xspace}

\newlist{inlinelist}{enumerate*}{1}
\setlist[inlinelist,1]{label=\emph{(\roman*)}, before=\unskip{~}\quad, after=\unskip{~}\quad, itemjoin={\quad}}

\counterwithin*{section}{part} % Reset equation numbering by section

\begin{document}

\maketitle
% clear footer from first page
\thispagestyle{empty}

\vspace{-1.7cm}

\authorinfo
{
  \href{https://orcid.org/0000-0001-7007-7773}{\includegraphics[scale=0.06]{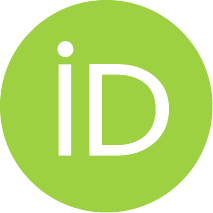}\hspace{1mm}\textcolor{black}{Stefan Volz}},
  \href{https://orcid.org/0000-0003-1427-0776}{\includegraphics[scale=0.06]{orcid.pdf}\hspace{1mm}\textcolor{black}{Martin Storath}}
}
{Lab for Mathematical Methods in Computer Vision and Machine Learning}
{Technische Hochschule Würzburg-Schweinfurt}
{Ignaz-Schön-Str. 11, 97421 Schweinfurt, Germany}

\begin{center}\textsc{and}
\end{center}

\authorinfo
{\href{https://orcid.org/0000-0002-4969-7609}{\includegraphics[scale=0.06]{orcid.pdf}\hspace{1mm}\textcolor{black}{Andreas Weinmann}}}
{Department of Mathematics and Natural Sciences}
{Hochschule Darmstadt}
{Schöfferstraße 3, 64295 Darmstadt, Germany}

\begin{center}\emph{Preprint}
\end{center}
\begin{center}\textsc{\today}
\end{center}

\begin{abstract}
  Many popular piecewise regression models rely on minimizing a cost function on the model fit with a linear penalty on the number of segments. However, this penalty does not take into account varying complexities of the model functions on the segments potentially leading to overfitting when models with varying complexities, such as polynomials of different degrees, are used. In this work, we enhance on this approach by instead using a penalty on the sum of the degrees of freedom over all segments, called degrees-of-freedom penalized piecewise regression (DofPPR).
  We show that the solutions of the resulting minimization problem  are unique for almost all input data in a least squares setting.
  We develop a fast algorithm which does not only compute a minimizer but also  determines an optimal hyperparameter -- in the sense
  of rolling cross validation with the one standard error rule -- exactly.
  This eliminates manual hyperparameter selection.
  Our method supports optional user parameters for incorporating domain knowledge.
  We provide an  open-source Python/Rust code for the piecewise polynomial least squares case which can be extended to further models. We demonstrate the practical utility through a simulation study and by applications to real data.
  A constrained variant of the proposed method gives state-of-the-art results in the Turing benchmark for unsupervised changepoint detection.
\end{abstract}

% keywords can be removed
% \keywords{Piecewise Regression \and Fast Algorithm \and Penalized Least Squares \and Nonconvex Optimization \and Model Selection \and Dynamic Programming \and Signal Estimation \and Signal Processing}

%%%%%%%%%%%%%%%%%%%%%%%%%%%%%%%%%%%%%%%%%%%%%%%%%%%%%%%%%%%%%%%%%%%%%%%%%%%%%%%
%%% start of main text
%%%%%%%%%%%%%%%%%%%%%%%%%%%%%%%%%%%%%%%%%%%%%%%%%%%%%%%%%%%%%%%%%%%%%%%%%%%%%%%

\part{Main Text}

\section{Introduction}

Assume we are given the the noisy samples $y_i = g(t_i)$
of a function $g \colon T \to X$ at time points $t_1 < \ldots< t_N$, where $(X, d)$ is a metric space and $T$ a real interval.
In many practical applications, the  signal $g$ can be well described by a piecewise function of some sort.
For example, piecewise constant signals appear in
the reconstruction of brain stimuli \citep{winkler2005don},
single-molecule analysis \citep{joo2008advances,loeff2021autostepfinder},
cellular ion channel functionalities \citep{hotz2013idealizing},
the rotations of the bacterial flagellar motor
\citep{sowa2005direct, nord2017catch} or medication use research \citep{wagner2002segmented}.
Similarly we find higher order piecewise polynomial functions in
fuel consumption estimation in automotive engineering \citep{paper:energy-management-for-the-electric-powernet}, the
modeling of human learning in quantitative psychology \citep{paper:piecewise-power-laws-in-individual-learning-curves},
and
studies of animal ecology based on biotelemetry \citep{paper:using-piecewise-regression-to-identify-biological-phenomena}.

The goal is to find a functional description of the underlying piecewise signal. This leads to a piecewise regression problem (also known as segmented regression).
Popular approaches for this task are \textit{partition-penalized models};
they are based on a cost functional of the competing objectives of data fidelity cost and parsimony of the segments, weighted by a parameter $\gamma > 0$ that represents their relative importance:
\begin{align}\label{main:eq:partition_penalized_model}
  \min_{P \text{ partition on } \{1,\ldots,n\}} \  \gamma \cdot \# P + \sum_{I \in P} \min_{\omega \in \Omega}d_I( \omega, y).
\end{align}
Here $d_I(\omega, y)$ measures the goodness of fit the model function $\omega$ to the data on the interval $I$;
for example $d_I(\omega, y) =  \sum_{i \in I} (\omega(t_i)-  y_i)^2,$
$P$ is a partition into discrete intervals $I$ of the time indices $\{1, \ldots, n\}$
and $\# P $ denotes the number of segments in the partition.
Typical instances of \eqref{main:eq:partition_penalized_model} are piecewise polynomial models and piecewise smooth models with $\ell^p$ data fidelity \citep{winkler2002smoothers, friedrich2008complexity, little2011generalized,jackson2005algorithm,killick2012optimal,weinmann2014l1potts,  blake1987visual, mumford1989optimal, storath2019smoothing, storath2023smoothing, yu2022localising}. The approach has been generalized to manifold-valued data
\citep{weinmann2016mumford, storath2017jumppenalized} and indirectly measured data \citep{storath2014jump, weinmann2015iterative}.
An important condition on $d_I$ and $\Omega$ is that $\min_{\omega \in \Omega}d_I( \omega, y)$ can be computed within a reasonable time frame.
A typical application of the piecewise regression model \eqref{main:eq:partition_penalized_model} is changepoint detection
as the boundaries of an optimal partition can be interpreted as changepoints, see e.g. \cite{killick2012optimal, paper:ruptures}.

In applications, the  complexities of
the used models on each segements may vary from segment to segment.
A simple yet instructive kind of a toy example is the timeseries of the distance traveled by a vehicle which may be parked, drive at a constant speed, or accelerate.
A parsimonious model would use a constant model for the parked phase, a linear model for the driving phase, and a quadratic model for the acceleration phase.
Another class of examples are signals that
can be well described by a simple model until reaching a change point, after which a more complex model provides a better fit;
specific examples are accelerometer readings from a motorcycle before and during an accident \citep{silverman1984spline},
heart rates of brown bears after spotting a drone \citep{paper:using-piecewise-regression-to-identify-biological-phenomena},
and bolus arrival times in dynamic contrast-enhanced MRI \citep{cheong2003automatic,bendinger2019bolus}.
The Turing changepoint dataset \citep{paper:an-evaluation-of-change-point-detection-algorithms}
contains further examples where models with mixed complexities may provide a satisfactory fit,
for example  the total private construction spending in the U.S. and the concentration of carbon dioxide in the atmosphere.
Generally speaking, allowing mixed complexities is reasonable whenever we cannot assume a priori that the models on different segments have the same complexity.

A limitation of the penalty based on the number of segments $\# P$ as in \eqref{main:eq:partition_penalized_model} is that it does not take into account the complexity of the models on the segments $I$. Each new segment has the same
cost, independently of the degrees of freedom of the regression function on the segment.
As an example, consider a piecewise polynomial regression: If the regression function on each piece is polynomial up to degree 3, then a regression function typically has degree 3, even if a polynomial of lower degree gives an almost as good but \enquote{sparser} representation.
To address this, one may introduce an extra penalty on the degrees of freedom on the segments, but this comes at the cost of introducing a new hyperparameter which complicates  model selection.

\subsection{Proposed method and contributions}\label{main:sec:proposed_method}
In this work, we study a model that penalizes not the number of segments but instead the degrees of freedom (dof) of the regression function on those segments:
\begin{align}\label{main:eq:proposed_model_intro}
  \min_{\substack{P \text{ partition of } \{1,\ldots,n\} \\ \lambda \colon P \to \N}} \  \sum_{I \in P} \min_{\omega \in \Omega^{\leq \lambda_I}}d_I( \omega, y) + \gamma \lambda_I.
\end{align}
Here $\Omega^{\leq \nu}$ is a space of regression functions with at most $\nu \in \N$ degrees of freedom (for example a $\nu$ dimensional linear space), $\lambda$ a function assigning each segment $I$ of $P$ a number of degrees of freedom $\nu \leq \nu_{\mathrm{max}}$, and $\nu_{\mathrm{max}}$ is an upper bound on the degrees of freedom per discrete interval.
To fix the ideas, one may think of the $\Omega^{\leq \nu}$ as the space of polynomials of maximum degree $\nu -1.$
We refer to \eqref{main:eq:proposed_model_intro} as \textit{degrees-of-freedom penalized piecewise regression model (DofPPR).}
A particularly important instance of \eqref{main:eq:proposed_model_intro} is the case of (weighted) least square fitting:
\begin{align}\label{main:eq:proposed_model_intro_poly}
  \min_{\substack{P \text{ partition of } \{1,\ldots,n\} \\ \lambda \colon P \to \N}} \  \sum_{I \in P} \left[\min_{\omega \in \Omega^{\leq {\lambda_I}}} \sum_{i \in I} w_i | \omega(t_i) - y_i|_2^2 + \gamma {\lambda_I} \right],
\end{align}
where $w\in \R^n$ is a weight vector with positive entries.

A crucial point of \eqref{main:eq:proposed_model_intro} (and of \eqref{main:eq:proposed_model_intro_poly}) is that it involves a model selection procedure on each segment, and the problem can be regarded as minimizing the sum of the AIC scores over all segments.
In contrast to the partition penalized model \eqref{main:eq:partition_penalized_model},
the DofPPR model \eqref{main:eq:proposed_model_intro} takes into account the complexity
of an estimator on the intervals, and this allows more flexible adaption to signals with mixed complexities.
Figure \ref{main:fig:comparison_pcw_DofPPR} illustrates the difference of
the partition penalized model \eqref{main:eq:partition_penalized_model}
and the DofPPR model \eqref{main:eq:proposed_model_intro} for piecewise polynomial least squares regression.
(The example is designed to illustrate the difference between the models
clearly and isolatedly from other factors, and results on more diverse data will be shown later in the paper.)

\begin{figure}
  \centering
  \includegraphics[width=1\textwidth]{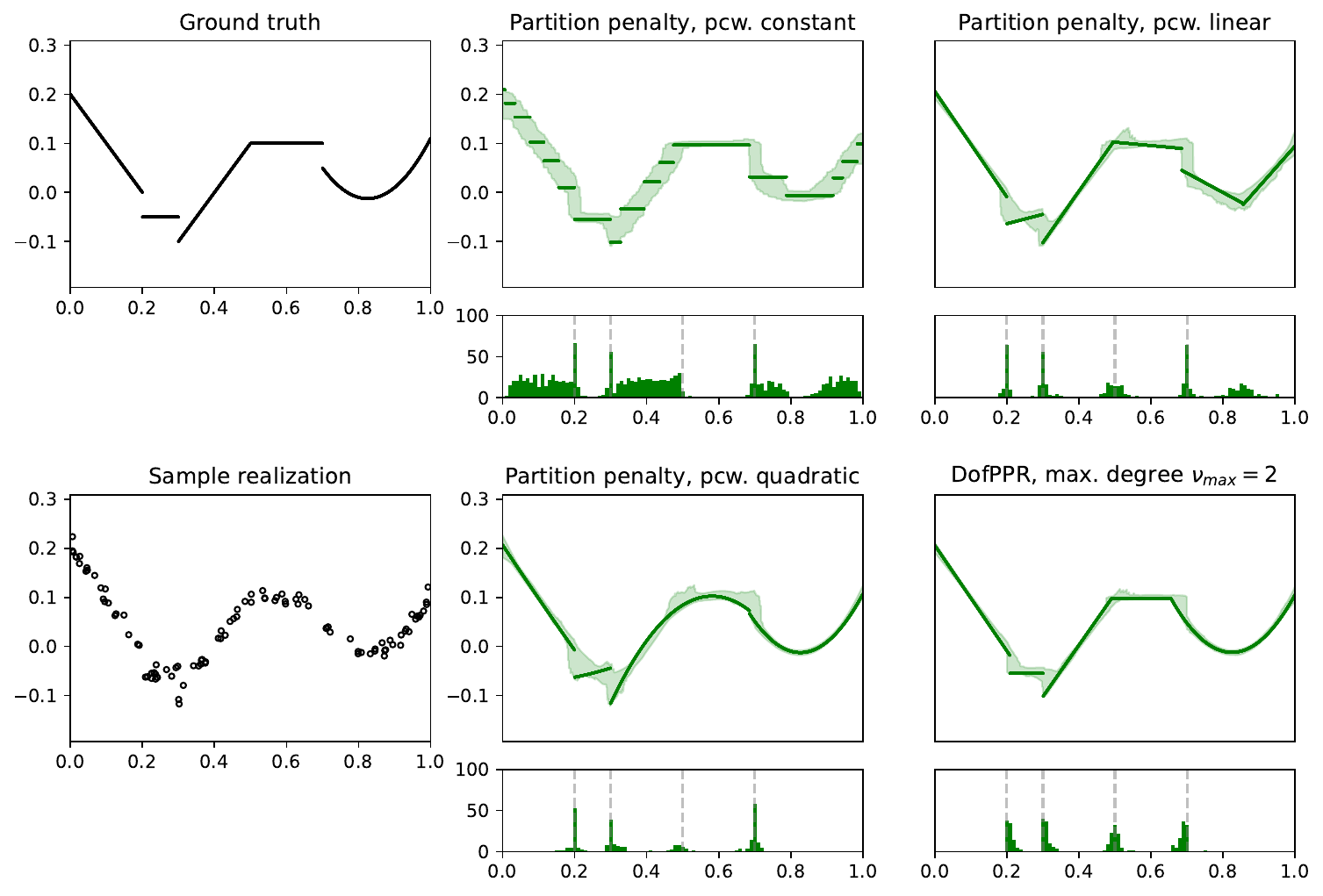}
  \caption{Comparison of the partition penalized model
    and the degrees-of-freedom penalized model for piecewise polynomial least squares regression: A realization (bottom left) consists of the piecewise polynomial ground truth signal (top left) sampled at 100 randomly (uniformly in $[0,1]$) chosen data sites, corrupted by i.i.d. Gaussian noise with standard deviation $\sigma = 0.01.$
    The top central, top right and bottom central tiles show the results of the partition penalized
    model \eqref{main:eq:partition_penalized_model} with  constant, linear and quadratic least squares polynomial fit, respectively.
    The solid lines represent the estimate of the depicted sample realization,
    and the shadings depict the $0.025$ to $0.975$  quantiles over 100 realizations.
    The last tile shows the result of the degrees of freedom penalized model (Equations \eqref{main:eq:proposed_model_intro}, \eqref{main:eq:proposed_model_intro_poly}),
    with least squares polynomial fit up to degree two.
    The taken breakpoints are depicted as histogram below along with the true breakpoints (dashed).
    The penalty $\gamma = 0.001$ was chosen for the DofPPR manually by visual inspection.
    The parameters for the other model were chosen such that the penalties on the segments
    match the respective penalty of the DofPPR model for constant, linear and quadratic polynomials, so $\gamma = 0.001, 0.002, 0.003$ for top central, top right, and bottom central, respectively.
  }
  \label{main:fig:comparison_pcw_DofPPR}
\end{figure}

Just as for the partition penalized models \eqref{main:eq:partition_penalized_model}, minimizers of the DofPPR model \eqref{main:eq:proposed_model_intro} may not be unique for general data fidelity terms.
For the least squares case \eqref{main:eq:proposed_model_intro_poly},
we show under certain mild assumptions on the basis function systems and up to excluding interpolating parts that the minimizing partition $P^*_\gamma$ is unique for almost all (w.r.t. Lebesgue measure) input data $y \in \R^n$ (Theorem~\ref{main:thm:UniquenessAEPart}).
The corresponding discrete function estimate $\hat y$ (which is the evaluation of the corresponding piecewise regression function at the discrete sample points $t_1, \ldots, t_n$) is shown to be unique even without these assumptions (Theorem~\ref{main:thm:uniquessFunction}).

For a fixed $\gamma$-parameter, problem \eqref{main:eq:proposed_model_intro} can be solved by adapting
standard dynamic programming methods \citep{winkler2002smoothers, friedrich2008complexity,killick2012optimal}.
A main contribution of this work is a new fast algorithm that provides the full regularization path for \eqref{main:eq:proposed_model_intro}; i.e. it provides the mapping $\gamma \mapsto (P^*_\gamma, \lambda^*_\gamma)$ for all $\gamma \geq 0.$
(In case of non-unique minimizers,  $(P^*_\gamma, \lambda^*_\gamma)$ denotes a distinguished minimizer.) The result is exact up to usual numerical errors.
In Theorem \ref{main:thm:complexity} we show that, under the assumption that a model of complexity $m$ can be fit to data of length $n$ in $\O(mn)$ time, the complete algorithm's time complexity is no worse than $\O(n^3m^2),$
where $n$ is the length of the input timeseries and $m$ is an upper bound on the (local) model complexity -- \ie only considering $\lambda$ such that $\lambda_I \leq m$ for all $I$. In particular for the case of least squares polynomial models of degrees no larger than $m$ for real data, we show that the presented algorithm is in $\O(n^3m)$.

Selecting the hyperparameter $\gamma$ is a common practical issue,
as the value of the parameter is difficult to interpret.
For the partition-penalized model
a series of criteria have been proposed, for example information based criteria \citep{zhang2007modified, yao1988estimating},  an interval criterion \citep{winkler2005don}, Bayesian methods \citep{frecon2017bayesian}, and cross validation
\citep{segmentation-of-the-mean-of-heteroscedastic-data-via-cv}.
We adopt the latter approach, and -- as we deal with time series -- use \textit{rolling cross validation.}
To obtain more parsimonious results it is common to apply the \emph{one standard error rule (OSE)} to cross-validated models rather than directly choosing the model with the strictly lowest score \citep{book:elements-of-statistical-learning}.
Remarkably,  we are able to compute a globally optimal hyperparameter with respect to rolling cross validation and the one standard error rule exactly without noteworthy additional extra computational effort.
The key observation is that $\gamma \mapsto \CV(\gamma)$ is a piecewise constant function, and that the proposed algorithm
admits its exact computation for all $\gamma \geq 0.$

We provide a Python/Rust implementation on Github at \url{https://github.com/SV-97/pcw-regrs}, featuring a fully implemented piecewise polynomial weighted least squares method. The code can be easily modified to accommodate other model spaces and cost functions.
One key benefit of our method is that it does not require additional tuning, such as adjusting model parameters or hyperparameters for an optimizer. However, it is easy to add optional parameters so experts can use their knowledge or impose restrictions on the models. This makes our method useful for both experts and non-experts. Here, we implemented the maximum number of degrees of freedom $\nu_{\mathrm{total}} \in \{ 1, \ldots, N\}$ as additional optional hyperparameter.

Potential applications include dimensionality reduction of time series data and serving as a foundation for a change point detector in exploratory data analysis.
A simulation study using piecewise polynomials of mixed degrees shows the advantage of the proposed DofPPR with OSE parameter selection over partition penalized models. Additionally, we highlight the method's effectiveness in exploratory data analysis for real-world time series data.
A simple variant of our method that uses the same constraint on  the number of changepoints as the current state-of-the-art method gives  state-of-the-art results on the Turing benchmark for unsupervised changepoint detection \citep{paper:an-evaluation-of-change-point-detection-algorithms}.

\subsection{Prior and related work}\label{main:ssec:relatedWork}
To the best of our knowledge, the piecewise regression model, as defined in Equation \eqref{main:eq:proposed_model_intro}, has not been previously explored, and as a result, its properties and associated algorithms have not been examined in prior research. However, there are various approaches that are closely related.

Piecewise regression models have a long history.
Early works employ a hard constraint on the maximum number of segments  \citep{bruce1965optimum,bellman1969curve, auger1989algorithms}.
Partition-penalized (or jump-penalized) piecewise regression
with a linear penalty on the number of segments as in \eqref{main:eq:partition_penalized_model} were studied
in various forms with different approximation spaces
and different cost functions.
The piecewise constant least squares regression case seems to be best understood:  \cite{wittich2008complexity} has proven uniqueness results and
\cite{boysen2009consistencies} obtained several consistency results.
Piecewise constant regression with  robust $\ell^1$ cost function have been
studied by \cite{friedrich2008complexity}, \cite{weinmann2014l1potts}, and \cite{storath2017jumppenalized}.
To account for variations on the segments, several approaches have been proposed.
\cite{mumford1985boundary} and  \cite{blake1987visual} studied first and second order splines, respectively, where  \cite{blake1987visual} obtained results on the jump localization and on robustness to noise.
Piecewise polynomial model have been studied by \cite{storath2019smoothing} and \cite{baby2020adaptive}.
\cite{zheng2022consistency} have shown consistency for
models with linear segment penalty
as presented in \eqref{main:eq:partition_penalized_model}.
For piecewise  polynomial regression,
\cite{yu2022localising} provide upper bounds on the localisation error
and derive global information-theoretic lower bounds.
\cite{romano2022detecting} use random walks
and autocorrelated noise via an AR(1) process
which leads to an energy functional combining Mumford-Shah/Blake-Zisserman type  penalties
and autoregressive terms.
\cite{hallgren2022changepoint} model the variations within a segment using an
$m$-dependent model where the $m$ is unknown.

Regarding parameter selection, some references have been given earlier in Section \ref{main:sec:proposed_method}.
Among the common approaches are information criteria, such as AIC, BIC, mBIC, MDL \citep{shi2022comparison}, and different kinds of cross validation \citep{arlot2011segmentation, storath2023smoothing}.
\cite{pein2021cross} discuss pitfalls and solutions when using an interleaved cross validation procedure.

Fast algorithms for piecewise regression models have been studied in a series of works. \cite{bellman1969curve} proposed a dynamic programming algorithm
for piecewise linear estimation. \cite{auger1989algorithms} proposed a dynamic programming algorithm for the fixed number of segments in $O(n^2 T(n))$ where $T(n)$ is the time needed for fitting a model to a segment.
An $O(n^2)$  algorithm for complexity penalized  least squares regression with piecewise constant model functions was proposed in \cite{winkler2002smoothers}.
Further algorithms for the piecewise constant regresssion problems were proposed by \cite{jackson2005algorithm}, \cite{little2011generalized2}, \cite{little2011generalized}, \cite{johnson2013dynamic}.
The PELT method \citep{killick2012optimal} uses a pruning strategy to improved the time complexity to $O(n)$ if the number of discontinuities grows sufficiently fast with the length of the signal.
A related pruning strategy with similar observations has been proposed by \cite{storath2014fast}.
For the first order spline model, an algorithm of cubic worst-case complexity  has been proposed by \cite{blake1989comparison}, and been improved to quadratic complexity by \cite{hohm2015algorithmic}. For the second order spline model, algorithms of quadratic complexity have been proposed by \cite{storath2019smoothing} and \cite{storath2023smoothing}, in a discrete and the continuous setting, respectively.
\cite{lindelov2020mcp} proposed a Bayesian piecewise regression approach.
Several papers consider piecewise regression with continuity contraints
such as change in slope problems \citep{fearnhead2019detecting,runge2020change,fearnhead2024cpop}.
Further variants include $\ell^1$ data terms \citep{friedrich2008complexity, weinmann2014l1potts}, manifold-valued data terms \citep{weinmann2016mumford, storath2017jumppenalized}
and an inverse problem setup
\citep{storath2014jump, weinmann2015iterative}.

Regularization paths of partition penalized  models have been investigated  by \cite{winkler2002smoothers} for the piecewise constant least squares model. They have shown that the  parameter space
for $\gamma$ is partitioned into finitely many intervals which give identical solutions.
Regularization paths of further piecewise  models have been given by \cite{friedrich2008complexity},  for the piecewise constant least absolute deviation model by \cite{storath2017jumppenalized}, and in a changepoint detection context by \cite{computationally-efficient-changepoint-detection}.
Further related are works on the regularization paths of the lasso  \citep{friedman2010regularization}.

Piecewise regression is closely related to changepoint detection.
Detected changepoints can define the pieces for piecewise regression,
and conversely, the segment boundaries of piecewise regression can be considered as changepoints.
Besides the discussed piecewise regression methods, there are a series of other changepoint estimation methods.  CUSUM-based methods \citep{page1954continuous, romano2023fast}, Bayesian changepoint inference \cite{fearnhead2006exact}, methods based on binary segmentation \citep{fryzlewicz2014wild, kovacs2023seeded}, a narrowest over the threshold method,  \citep{baranowski2019narrowest}, or a
Bayesian ensemble approach  \citep{zhao2019detecting}, nonparametric maximum likelihood approaches \citep{zou2014nonparametric, haynes2017computationally}, multiscale testing \citep{frick2014multiscale}, random forests \citep{londschien2023random},  to mention only a few.
We refer to \cite{paper:an-evaluation-of-change-point-detection-algorithms} and the references therein for a more detailed overview and a comparison over selected changepoint detection methods.
We note that while both concepts deal with segmenting data, changepoint detection emphasizes locating points of change, while piecewise regression focuses on fitting separate models to each segment. In particular, changepoint detection methods that are not of piecewise regression type do not necessarily come with an associated regression function on the segments.

\subsection{Preliminaries and notation}
We call an (ordered) set of the form $I = \{ l, l+1 , \ldots, r\}$ a discrete interval,
and we abbreviate it by $l:r .$
Throughout the paper we assume that
the elements of a partition on $\{1, \ldots, n\}$ are discrete intervals.
We occasionally use the notation $y_I = (y_i)_{i \in I}$ for extracting a subvector with indices $I.$
We define the set of ordered partitions of a set $A$ to be the set of all partitions of $A$, such that for any two elements $I,I'$ of a partition we have either $a < a'$ or $a' < a$ uniformly for all $a \in I, a' \in I'$. We denote this set by $\OPart(A)$ and frequently identify its elements and the elements of its elements with the correspondingly ordered tuples. For example $((1), (2), (3,4), (5,6)) = (1:1, 2:2, 3:4, 5:6)$ is an ordered partition of $1:6$.

For any ordered partition $P$ of a discrete interval let $\Lambda(P) := \{ (\lambda_I)_{I \in P} \colon \lambda_I \in 1: \ElemCount{I} \}$ be the space of valid sequences of degrees of freedom for $P$. We usually denote elements of this set by $\lambda$. For any $I \in P$ the value $\lambda_I$ is a \emph{degree of freedom}.

$\bigsqcup_{i \in I} A_i$ denotes the disjoint union of a family of sets $\{A_i\}_{i \in I}$ indexed by a set $I$. Any elements of $\bigsqcup_{i \in I} A_i$ is naturally in bijection to a pair $(i, a)$ such that $i \in I, a \in A_i$.

Throughout this paper, we assume that the data sites satisfy
$t_1 < t_2 < \ldots < t_n.$
If the data does not satisfy this constraint,
we may merge data sites into a single data point by weighted averaging over the $y$-values
of coinciding $t$-values; see e.g. \cite{hutchinson1986algorithm}.

An illustration of the notation is given in Figure~\ref{main:fig:overview}.
\begin{figure}
  \centering
  \begin{tikzpicture}
    \node[inner sep=0, anchor=south west] (image1) {
      \includegraphics[width=\textwidth]{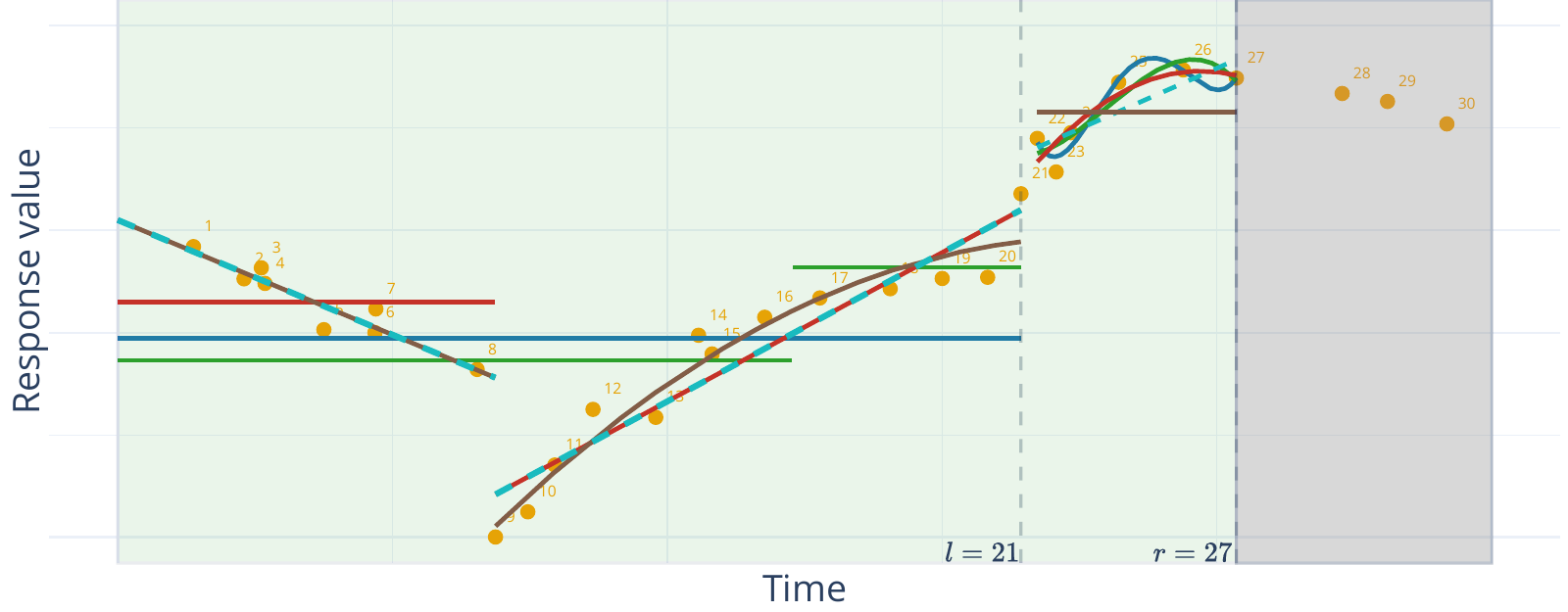}
    };

    \pgfmathsetmacro{\xStart}{0.07}
    \pgfmathsetmacro{\xL}{-0.34}
    \pgfmathsetmacro{\xR}{-0.212}

    \draw[decorate, decoration={brace,amplitude=10pt}, thick]
    ([xshift=\xStart\textwidth]image1.north west)
    -- ([xshift=\xR\textwidth]image1.north east)
    node[midway, above=10pt] {Active segment $I=1:27$};
    \draw[decorate, decoration={brace, mirror, amplitude=10pt}, thick]
    ([xshift=\xL\textwidth]image1.south east)
    -- ([xshift=\xR\textwidth]image1.south east)
    node[midway, below=10pt] {Last segment of partition};

    % \draw[decorate, decoration={brace, mirror, amplitude=10pt}, thick]
    % ([xshift=\xStart\textwidth]image1.south west)
    % -- ([xshift=\xL\textwidth]image1.south east)
    % node[midway, below=10pt] {};
    \draw[decorate, decoration={brace, mirror, amplitude=10pt}, thick]
    ([xshift=\xL\textwidth]image1.south east)
    -- ([xshift=\xR\textwidth]image1.south east)
    node[midway, below=10pt] {Last segment of partition};
  \end{tikzpicture}
  \caption{
    Illustration of notation used in this paper: A timeseries $(t_i, y_i),$ with $i=1,...,30,$ is represented by orange dots.
    $I = 1:27$ represents the first $27$ indices and $y_I$ the corresponding part of the signal.
    The plot shows five different piecewise polynomial functions for the data on $I$
    which all have exactly six degrees of freedom (dofs).
    We first look at the brown line. Its ordered partition is given by $P = (1:8,~9:21,~22:27),$ and the associated dof sequence $\lambda_{\text{brown}}$ is given by $(2, 3, 1);$
    its entries sum up to $\nu = 2+ 3 + 1 = 6.$
    The brown line represents the least squares fit to the data on the three segments
    using polynomial functions with the indicated dofs, so linear, quadratic, and constant.
    The cyan and the red candidate have the same partition, the same degrees of freedom ($\nu = 6$),
    but the different dof sequences $\lambda_{\text{cyan}} = (2, 2, 2)$
    and $\lambda_{\text{red}} = (1, 2, 3),$ respectively.
    (In Section~\ref{main:sec:fast_algorithm} we use the notation $\Lambda^6(P)$ for all valid dof sequences for the
    partition $P$ that sum up to $6$.)
    The green and blue candidates have different partitions, namely $(1:16,17:21,22:27)$ and
    $(1:21,22:27),$ and the dof sequences $\lambda_{\text{green}} = (1, 1, 4),$ and $\lambda_{\text{blue}} = (1, 5),$ respectively.
    The shown candidates share the last segment but
    \enquote{spend} between one and five dofs  on that segment,
    so that  only five to one dofs remain for the remaining data with indices 1 to 21.
    A crucial part of the proposed method is that it computes the solution of the model \eqref{main:eq:proposed_model_intro_poly} for all regularization parameters $\gamma \geq 0$ and for all \enquote{active} segments $I = 1:r,$ with $r =1, \ldots, n,$ efficiently.
  }
  \label{main:fig:overview}
\end{figure}

\subsection{Organization of the paper}
In Section \ref{main:ssec:uniqueness}, we
prove uniqueness of the minimizer in a least squares setting.
In Section \ref{main:sec:fast_algorithm}, we develop and discuss an algorithm that solves the DofPPR problem
and computes the regularization paths, and provide parameter selection strategies.
In Section \ref{main:sec:experiments},
we conduct numerical experiments with simulated and real data.
Section \ref{main:sec:conclusion} is devoted to discussion and conclusion.

Most of the proofs, auxiliary results and details on the implementation are provided in the supplementary material.

\section{Uniqueness results for
  degrees-of-freedom penalized piecewise least squares regression
}\label{main:ssec:uniqueness}

Let us first discuss what we generally may expect regarding the uniqueness of an estimate.
The ultimate goal is to obtain a piecewise regression function which is defined on the real line.
The models \eqref{main:eq:partition_penalized_model}, \eqref{main:eq:proposed_model_intro}
deal  with partitions on the discrete data sites $(t_i)_{i=1,\ldots, n}$ represented by their indices $1:n$.
A natural  real counterpart of a
discrete partition is a partition on the real line into intervals
such that the real segments contain exactly the data sites of the corresponding discrete setting.
Clearly, there are infinitely many
real partitions that fullfil this requirement.
As their segments contain the same data sites, all instances give the same functional value in \eqref{main:eq:partition_penalized_model} and \eqref{main:eq:proposed_model_intro}, and so they are indistinguishable w.r.t. to these models.
Thus, we consider two real partitions as equivalent,
if all their respective segments contain the same data sites.
A natural representative of the equivalence class
is that partition which has its breakpoints at the midpoints between the data sites.
In the following, we use this correspondence,
and the following uniqueness results for partitions always refer to uniqueness of the discrete partitions.

Let us also briefly discuss the relation
of an optimal solution $(P^*_\gamma, \lambda^*_\gamma)$ of \eqref{main:eq:proposed_model_intro}, and a corresponding regression function.
If the
estimated functions $\phi^{\lambda_I} = \argmin_{\omega \in \Omega^{\leq {\lambda_I}}}d_I( \omega, y)$ are unique for all elements of $(P^*_\gamma, \lambda^*_\gamma),$
then that solution and the estimated functions define a piecewise function, denoted by $\omega_{\gamma}^*,$
which is unique up to shifts of the break locations between the sampling points corresponding to the borders of two adjacent segments of the partition.
The corresponding evaluations at the discrete data sites, denoted by $\hat y$ and given by $\hat y_i = \omega_\gamma^*(t_i)$ for all $i \in \{1, \ldots, n\},$
are well-defined and independent of these possible shifts in the breakpoints.

As for the partition-penalized model \eqref{main:eq:partition_penalized_model},
the minimizing (discrete) partitions of the DofPPR model \eqref{main:eq:proposed_model_intro}
are not unique for all combinations of input signals and model parameters.
This can be seen in the following example:

\begin{example}\label{ex:non-uniqueness}
  Consider the time series with sample times $x = [0,1,2]$ and values $y = [0,1,0]$ and least squares polynomial fitting functions. All of the four possible models with three degrees of freedom have vanishing energy: $P=(0:2), \lambda=(3); P=(0:1, 2), \lambda=(2,1); P=(0, 1:2), \lambda=(1,2); P=(0,1,2), \lambda=(1,1,1),$ where $P$ denotes the partition, and $\lambda$ the corresponding sequence of degrees of freedom.
\end{example}
We describe how to resolve ambiguous situations in the estimator in general in the algorithmic part in Section~\ref{main:sec:fast_algorithm}.
For the particularly important  least squares setting \eqref{main:eq:proposed_model_intro_poly},
we are going to show uniqueness results for almost all input data.
To this end, for fixed partitioning $P$ and fixed degree sequence $\lambda = (\lambda_I)_{I \in P},$
we denote the target function in  \eqref{main:eq:proposed_model_intro_poly} by
$G_{P, \lambda}(y)$ where
\begin{align*}
  G_{P, \lambda}(y)
  & = \sum_{I \in P} \left[ \min_{\omega \in \Omega^{\leq \lambda_I}} \sum_{i \in I}  | \omega(t_i) - y_i|_2^2 + \gamma \lambda_I \right]                        \\
  & = \sum_{I \in P} \left[ \min_{\beta \in \R^{{\svadd{\lambda_{I}}}}}\| A_{I, {\svadd{\lambda_{I}}}} \beta - y_I\|^2_2 + \gamma {\svadd{\lambda_{I}}} \right],
\end{align*}
and $A_{I, {\svadd{\lambda_{I}}}}$ denotes the design matrix on the segment $I.$
Using the symbol $\mathrm{id}$ to denote the identity matrix,
a standard computation yields
\begin{equation}\label{main:eq:QuadWithProjection}
  G_{P, \lambda}(y) = \sum_{I \in P} \left[ \| (\pi_{I, {\svadd{\lambda_{I}}}} - \mathrm{id}) y_I\|^2_{2} + \gamma {\svadd{\lambda_{I}}} \right]
\end{equation}
with the projection matrix (also called hat matrix)
\begin{equation}\label{main:eq:ProjectionOnInterval}
  \pi_{I, {\svadd{\lambda_{I}}}} = A_{I, {\svadd{\lambda_{I}}}}(A_{I, {\svadd{\lambda_{I}}}}^TA_{I, {\svadd{\lambda_{I}}}})^{-1} A_{I, {\svadd{\lambda_{I}}}}^T
\end{equation}
in case $A$ has full column rank.
In general, $\pi_{I, {\svadd{\lambda_{I}}}}  = A_{I, {\svadd{\lambda_{I}}}} M,$ where $M$ denotes the Moore-Penrose pseudoinverse of $A_{I, {\svadd{\lambda_{I}}}}.$
Given data $y,$ we  frequently use the hat notation
\begin{equation}\label{main:eq:DefHatAgain}
  \hat y = \hat y_{I, {\svadd{\lambda_{I}}}}  = \pi_{I, {\svadd{\lambda_{I}}}} y.
\end{equation}
(Note that although the corresponding least squares solutions $\beta^\ast$ may be non unique,  $A_{I, {\svadd{\lambda_{I}}}} \beta^\ast$ is unique
and thus $\hat y = \pi_{I, {\svadd{\lambda_{I}}}} y$ is well-defined.)
A particular case appears whenever $\pi_{I, {\svadd{\lambda_{I}}}}$ equals the identity: then the data on the subinterval $I$ remains unchanged,
and the corresponding fit is interpolatory.

Next, we represent, for fixed partitioning $P$ and fixed degree vector $\lambda = ({\svadd{\lambda_{I}}})_{I \in P},$ the solution operator $y \mapsto \argmin G_{P, \lambda}(y)$
for data $y.$
The representing matrix is given by
\begin{equation}\label{main:eq:MatrixRep}
  \bar\pi_{P,\lambda} =
  \begin{bmatrix}
    \pi_{I_1, {\svadd{\lambda_{I_1}}}} & 0                                  & \cdots & 0                                  \\
    0                                  & \pi_{I_2, {\svadd{\lambda_{I_2}}}} & \cdots & 0                                  \\
    \vdots                             & \ddots                             & \ddots & \vdots                             \\
    0                                  & \cdots                             & 0      & \pi_{I_k, {\svadd{\lambda_{I_k}}}}
  \end{bmatrix},
\end{equation}
and the solution operator of $y \mapsto \argmin G_{P, \lambda}(y)$ is given by
$y \mapsto \hat y = \hat y_{P,\lambda} =  \bar\pi_{P,\lambda}y. $

As a first step, we show the following lemma.
\begin{lemma}\label{main:lem:projections}
  Let $P, Q$ with $P \neq Q$ be two partitions, $\lambda = ({\svadd{\lambda_{I}}})_{I\in P},$
  $\mu = (\mu_J)_{J \in Q}.$
  Then either (i) the set
  \begin{equation}
    \{ y \in \R^n : G_{P, \lambda}(y) = G_{Q, \mu}(y)   \} \quad \text{has Lebesgue measure zero,}
  \end{equation}
  or (ii)
  \begin{equation}
    G_{P, \lambda} - G_{Q, \mu} = \alpha   \quad \text{for some constant $\alpha.$}
  \end{equation}
\end{lemma}
The proof is given in the supplementary material.

Using this lemma, we may derive uniqueness in the a.e. sense for the estimate of the regression function. Interestingly, this does not imply uniqueness a.e. w.r.t. the partitioning which we will discuss afterwards.

\begin{theorem}\label{main:thm:uniquessFunction}
  The minimizing (discrete) function $\hat y$ of Equation \eqref{main:eq:proposed_model_intro_poly}
  given by
  $\hat y_i = \omega_\gamma^*(t_i)$ for all $i \in \{1, \ldots, n\}$
  equals
  $\argmin_{P,\lambda} \hat y_{P,\lambda} = \argmin_{P,\lambda} \bar\pi_{P,\lambda}y$
  (given via \eqref{main:eq:MatrixRep})
  and is unique for almost all input data.
\end{theorem}

The proof is given in the supplementary material. (Please note that we have overloaded the hat notation several times; the precise implementation should be clear from the context.)

Unfortunately, we cannot get an as general uniqueness a.e. statement concerning the partitions as the following simple example shows.
\begin{example}\label{ex:NonUniquePartsAlthougAE}
  As a system of basis functions, we take the monomials $b_j = t^{j-1}.$
  Consider data sites $t=0,1,2.$
  As data $y_{a,b,c},$ we consider the polynomial   $t \to a  t^2 + b t + c$ $(a,b,c \in \mathbb R)$ sampled at the data sites.
  We choose the regularization parameter $\gamma = 1.$
  Given $b,c,$ we find $a_0(b,c)$ such that the optimal solution $\hat y$
  equals the data  $y_{a,b,c}$ for all $a >a_0(b,c).$
  The set of all $a,b,c$ such that  $a >a_0(b,c)$ is a non-zero set.
  Unfortunately, also partitioning the data sites into $(0)$ and $(1,2)$ and choosing interpolating constant and linear polynomials
  yields another optimal partitioning on a nonzero set (with however the same function values.)
\end{example}

Despite Example~\ref{ex:NonUniquePartsAlthougAE}, under certain mild assumption on the basis function systems (specified below) and up to excluding interpolating parts (as detailed below)
we are going to show uniqueness results for the partitionings for almost all input data in the following.

\begin{assumption}\label{main:assum:Basis}
  Concerning the basis functions $(b_j)_{j \in \mathbb N}$ (cf. \eqref{main:eq:proposed_model_intro_poly}) we assume that
  \begin{enumerate}[(i) ]
    \item  the design matrices $A$ given by
      $A_{ij} = (b_j(t_{i}))$ with $b_j, 1\leq j\leq n,$  and $t_i, 1\leq i\leq m,$ have maximal rank for any $n,m$  with $n\leq m,$
    \item the corresponding projections  $\pi_{I, {\svadd{\lambda_{I}}}} = A_{I, {\svadd{\lambda_{I}}}}(A_{I, {\svadd{\lambda_{I}}}}^TA_{I, {\svadd{\lambda_{I}}}})^{-1} A_{I, {\svadd{\lambda_{I}}}}^T$ given in \eqref{main:eq:ProjectionOnInterval}  corresponding to $A$ are atomic for any $n\leq m,$ in the sense that $\pi_{I, {\svadd{\lambda_{I}}}}$ cannot be further decomposed into two or more block diagonal matrices.
  \end{enumerate}
\end{assumption}

We see that this assumption is fulfilled by the class of polynomials.
\begin{lemma}\label{main:lem:PolynomialsFulfillAss}
  Consider polynomial design matrices  $A \in \R^{n \times m}$  with $m< n$,
  where each polynomial $b_j,$ ${j \in \mathbb N}$, is of degree $j-1,$
  e.g. $A_{ij} = t_i^{j-1},$
  then $A$ fulfills Assumption~\ref{main:assum:Basis}.
\end{lemma}
The proof is given in the supplementary material.

\begin{remark}
  Note that basis function systems not fulfilling (ii) of  Assumption~\ref{main:assum:Basis} are in a sense redundant in our setup.
  In fact, the corresponding projection $\pi_{I, {\svadd{\lambda_{I}}}}$ acts independently on the subintervals of $I$ corresponding to the blocks,
  and thus introduces a second notion of ``change point'' within the interval. This behavior may often be unwanted or undesirable when modeling.
\end{remark}

We need the specific notion of {\em standard block decompositions} of the representing matrices $\bar\pi_{P,\lambda}.$
We call a block decomposed matrix $B=\mathrm{diag} (B_i)$  (not necessarily one of the form \eqref{main:eq:MatrixRep}) standard block decomposed, if each $B_i$ is an atomic projection in the sense above.
(Note that if $B_i$ is one-dimensional the corresponding entry of the matrix equals $1$.)
Obviously, each matrix $\bar\pi_{P,\lambda}$ possesses a standard block decomposition
(which however need not equal the form \eqref{main:eq:MatrixRep}).
Further, each standard block decomposition induces a partitioning $P$ with intervals $I$ corresponding to the blocks of the decomposition.

We next relate $\bar\pi_{P,\lambda}$ with a standard decomposition and a corresponding partitioning $\tilde P$ (plus degree vector $\tilde \lambda$): (i) If $\pi_{I, {\svadd{\lambda_{I}}}}$ is non-interpolatory,  i.e., $\pi_{I, {\svadd{\lambda_{I}}}}$ does not equal the identity, we leave the block $\pi_{I, {\svadd{\lambda_{I}}}}$ unchanged, and include $I$ into the partitioning $\tilde P;$
we let the corresponding degree $\tilde {\svadd{\lambda_{I}}} = {\svadd{\lambda_{I}}}.$  (ii) If $\pi_{I, {\svadd{\lambda_{I}}}}$ is interpolatory,  i.e., $\pi_{I, {\svadd{\lambda_{I}}}}$ equals the identity, we replace the block by $1\times 1$ blocks (with entry $1$.)
Then, the standard block decomposition $\bar\pi_{P,\lambda}$ is of the form
\begin{equation}\label{main:eq:MatrixStandardDec}
  \bar\pi_{P,\lambda} =
  \begin{bmatrix}
    p_1    & 0      & \cdots & 0      \\
    0      & p_2    & \cdots & 0      \\
    \vdots & \ddots & \ddots & \vdots \\
    0      & \cdots & 0      & p_m
  \end{bmatrix}.
\end{equation}
where $p_i$ is either a $1 \times 1$-matrix with entry $1$ or equals one of the diagonal blocks $\pi_{I_i, {\svadd{\lambda_{I_i}}}}$ of \eqref{main:eq:MatrixRep}
(at the same position in the matrix $\bar\pi_{P,\lambda}.$)
This decomposition induces a partitioning $\tilde P$ consisting of the intervals $I$ of the following form: (i) either $I$ is an interval of the partitioning $J$ on which $\pi_{I, {\svadd{\lambda_{I}}}}$ is non-interpolatory, i.e., $\pi_{I, {\svadd{\lambda_{I}}}}$ does not equal the identity or (ii) $I$ is of length $1.$
A corresponding degree vector $\tilde \lambda$ is obtained by letting
$\tilde \lambda_I = \lambda_I$ in case of (i), and $\tilde \lambda_I = 1$ in case of (ii) respectively.

As illustration, we consider as a toy example, the representing matrix $\bar\pi_{P,\lambda}$
given by
\begin{equation}\label{main:eq:MatrixStandardDecExample}
  \bar\pi_{P,\lambda} =
  \begin{bmatrix}
    q & 0 & 0 & 0 \\
    0 & 1 & 0 & 0 \\
    0 & 0 & 1 & 0 \\
    0 & 0 & 0 & r
  \end{bmatrix}.
\end{equation}
where $q,r$ are the hat matrices of, say, polynomial interpolation on a equidistant grid with polynomial degree $2$ and four sample points for $q$ and five sample points for $r,$ respectively.
(The zeros correspond to zero block matrices of the respective sizes.)
Then there are four intervals of the corresponding partitioning of lengths $5,1,1,4,$
and the corresponding blocks of the standard decomposition are $p_1=q,$ $p_2=p_3=1$ (the $1 \times 1$ matrix with entry $1$) as well as $p_4=r.$

We formulate the following immediate observations as a lemma.
Loosely speaking, it asserts that we may restrict to consider partitions corresponding to standard block decompositions.

\begin{lemma}\label{main:lem:standardDecompositionExUn}
  Assume that Assumption~\ref{main:assum:Basis} is fulfilled.
  Consider a partitioning $P$ and degree vector $\lambda = ({\svadd{\lambda_{I}}})_{I \in P},$ with corresponding solution operator
  $\bar\pi_{P,\lambda}$ given via   \eqref{main:eq:MatrixRep}.
  Then, \emph{(i)} the matrices
  \begin{equation}\label{main:eq:equalMatrixf}
    \bar\pi_{P,\lambda}  = \bar\pi_{\tilde P,\tilde \lambda},
  \end{equation}
  i.e., the solution operators for partitioning $P$ and degree vector $\lambda = ({\svadd{\lambda_{I}}})_{I \in P},$ and for the standard decomposition  $\tilde P$ and degree vector $\tilde \lambda$  are equal.
  \emph{(ii)} The corresponding standard partitioning $\tilde P$ together with the degree vector $\tilde \lambda$  is unique.
\end{lemma}
The proof is given in the supplementary material.

We get the following uniqueness result in particular concerning (discrete) partitions.
\begin{theorem}  \label{main:thm:UniquenessAEPart}
  We consider the minimization problem \eqref{main:eq:proposed_model_intro_poly}.
  We assume that Assumption~\ref{main:assum:Basis} is fulfilled
  for the system of basis functions $(b_j)_{j \in \mathbb N}.$
  Then,
  the minimizing (discrete) function $\hat y$ of Equation \eqref{main:eq:proposed_model_intro_poly}
  given by
  $\hat y_i = \omega_\gamma^*(t_i)$ for all $i \in \{1, \ldots, n\}$
  is unique
  for almost all input data $y$ (w.r.t. Lebesgue measure.)
  Moreover, segments $I^\ast$ of a minimizing partitioning corresponding to non-interpolating estimation are unique
  for almost all input data $y$ (w.r.t. Lebesgue measure.).
\end{theorem}

The proof is given in the supplementary material.

Specifing the basis function system to the class of polynomials, we may formulate the following corollary.
\begin{corollary}\label{main:cor:Uniqueness4Pol}
  We consider the minimization problem \eqref{main:eq:proposed_model_intro_poly}
  for a system of polynomials $(b_j)_{j \in \mathbb N},$
  where each polynomial $b_j,$ $j \in \mathbb N,$  has precisely degree $j.$
  Then,
  the minimizing (discrete) function $\hat y$ of Equation \eqref{main:eq:proposed_model_intro_poly}
  given by
  $\hat y_i = \omega_\gamma^*(t_i)$ for all $i \in \{1, \ldots, n\}$
  is unique
  for almost all input data $y$ (w.r.t. Lebesgue measure.)
  Moreover, segments $I^\ast$ of a minimizing partitioning corresponding to non-interpolating estimation are unique
  for almost all input data $y$ (w.r.t. Lebesgue measure.)
\end{corollary}
The proof is given in the supplementary material.

We further obtain the following corollary which is important for implementation purposes as well.
\begin{corollary}\label{main:cor:UniquenessNuMax}
  Consider the same situation as in Corollary \ref{main:cor:Uniqueness4Pol}.
  To compute a minimizer of  the problem in  \eqref{main:eq:proposed_model_intro_poly} we may restrict, for a given partition $P,$
  the search space to the maximal degree
  \begin{equation}\label{main:eq:numax_corr}
    \nu_{\mathrm{max}}(I) = \max(1, \ElemCount{I} - 1)
  \end{equation}
  for each interval $I$ of the partitioning $P.$
\end{corollary}
The proof is given in the supplementary material.

\section{Fast algorithm for computing the regularization paths and model selection} \label{main:sec:fast_algorithm}

We next develop a fast  algorithm for computing the full regularization paths for the DofPPR problem  \eqref{main:eq:proposed_model_intro}, meaning a solver for all hyperparameters $\gamma \geq 0$ simultaneously.
This enables us to perform model selection based on rolling cross validation (with and without OSE-rule) requiring only little additional computational effort.

\subsection{Computing the regularization paths}

The general strategy is as follows:
We reduce the problem to solving a collection of constrained problems which in turn are solved using dynamic programming.
Extending a result from \cite{friedrich2008complexity}, we find that the optimal solution as a function of $\gamma$ is piecewise constant with only finitely many jumps, and we show that the  partition of $\R_{\geq 0}$ corresponding to this piecewise constant function may be obtained by computing the pointwise minima of finitely many affine-linear functions.

As preparation, we  formulate \eqref{main:eq:proposed_model_intro_poly}
for partial data $(t_1, \ldots, t_r),$
$(y_1, \ldots, y_r)$ and more compactly as
\begin{align}\label{main:eq:proposed_model_reformulated}
  \min_{\substack{P \in \OPart(1:r) \\ \lambda \in \Lambda(P)}} \sum_{\substack{I \in P }} d_I^{\lambda_I} + \gamma \lambda_I
\end{align}
where $d_I^{\lambda_I} = \min_{\omega \in \Omega^{\leq \lambda_I}}d_I( \omega, y).$
The core of the method is solving the following constrained problem, referred to as \emph{$\nu$ degree of freedom partition problem}:
\begin{align}\label{main:eq:degree_of_freedom_partition_problem}
  B_r^\nu  := \min_{\substack{P \in \OPart(1:r) \\ \lambda \in \Lambda^{\nu}(P)}} \sum_{\substack{I \in P  }} d_I^{\lambda_I},
\end{align}
where $\Lambda^{\nu}(P) := \{ \lambda \in \Lambda(P) \colon \nu = \sum_{\substack{I \in P }} \lambda_I \}$
are the dof sequences with exactly $\nu$ degrees of freedom in total.
We call the values $B_r^\nu$ the \emph{Bellman values} and refer to the associated table of numbers $(B_r^\nu)_{r, \nu}$ as \emph{Bellman table} $B$. Moreover we call a pair $(r, \nu)$ \emph{feasible} if the feasible set of the $\nu$ degree of freedom partition problem~\eqref{main:eq:degree_of_freedom_partition_problem} on data $1:r$ is nonempty so that $B_r^\nu$ is well-defined.

The solutions of \eqref{main:eq:proposed_model_reformulated}
and those of \eqref{main:eq:degree_of_freedom_partition_problem}
are related as follows:

\begin{lemma}[Reduction to constrained problems]\label{main:lem:reduction_to_unpenalized}
  For fixed $\gamma \geq 0, r \in 1:N$ the $\gamma$-penalized  problem
  \eqref{main:eq:proposed_model_reformulated}
  is equivalent to the problem
  \begin{align}\label{main:eq:tabulation_gamma}
    \min_{\nu \in 1:r} B_r^\nu + \gamma \nu.
  \end{align}
\end{lemma}
The proof is given in the supplementary material.

Computing the regularization paths for all possible values of $\gamma$ amounts to finding the pointwise minimum of a collection of affine linear functions and in particular their critical points:
By the previous lemma, computing the regularization paths is equivalent to solving $\min_{\nu \in 1:r} B_r^\nu + \gamma \nu$. The map $\gamma \mapsto \min_{\nu \in 1:r} B_r^\nu + \gamma \nu$ is the pointwise minimum of the set $\mathcal{F} := \{ \gamma \mapsto B_r^\nu + \gamma \nu \}_{\nu \in 1:r}$ of affine linear maps. This pointwise minimum is piecewise affine linear with the pieces between any two adjacent critical points being an element of $\mathcal{F}$. Any such piece then determines a solution for a whole range of penalties, such that the solution of the full optimization problem is a piecewise constant function of $\gamma$.
This correspondence is illustrated in Figure~\ref{main:fig:gamma_nu_mapping}.

\begin{figure}
  \centering
  \includegraphics[width=0.5\textwidth]{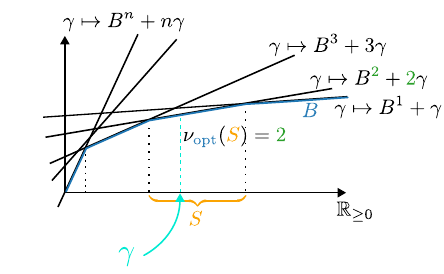}
  \caption{
  Computing the regularization paths of \eqref{main:eq:proposed_model_intro} for all possible values of $\gamma$ amounts to finding the pointwise minimum of a collection of affine linear functions and in particular their critical values. The $\gamma$-parameters between two critical values ($S$ in the graphic) provide the same minimizers of \eqref{main:eq:proposed_model_intro}. Hence, the mapping from the set of hyperparameters to the optimal solution $\gamma \mapsto (P^*_\gamma, \lambda^*_\gamma)$ is a piecewise constant function.}
  \label{main:fig:gamma_nu_mapping}
\end{figure}

In Section \ref{sup:sec:pointwise} of the supplementary material, we provide an algorithm that can be used to efficiently compute the pointwise minimum. This algorithm operates in linear time when applied to sorted inputs, and sorted inputs can be readily ensured in implementations without any overhead. At the points of intersection of two affine linear functions, our algorithm returns the model with fewer degrees of freedom in accordance with the principle of parsimony.

A crucial point for an efficient algorithm is that the Bellman
values $B_r^\nu$ can be computed efficiently by dynamic programming:

\begin{theorem} [Recursion for Bellman values]\label{main:thm:bellman-recursion}
  The values $B_r^\nu$ for $r \in 1:n, \nu \in 1:r$ admit the recursion
  \begin{align}\label{main:eq:bellman_recursion}
    B_r^\nu = \min_{\substack{ p_R \in 1 : \nu \\ l \in 0:r}} B_l^{\nu - p_R} + d_{l+1:r}^{p_R}
  \end{align}
  with initial conditions $B_r^1 = d_{1:r}^1, B_r^0 = B_0^\nu = 0$.
\end{theorem}
The proof is given in the supplementary material.

Note that the Bellman table thus computed generates a graph (an acyclic directed graph where every
node has at most one outgoing edge -- \ie a polyforest with edges oriented towards the roots):
for each feasible pair $(r, \nu)$ there is a node in the graph.
If the minimizing arguments of $\eqref{main:eq:bellman_recursion}$ indicate that
there is only one segment in the partition (so the segment $0:r$ is best approximated using a single model
with $\nu$ degrees of freedom) then the node $(r,\nu)$ in the graph has no outgoing edge.
If it indicates that there is a prior segment with node $(r',\nu')$ then there is an edge from $(r,\nu)$ to
$(r',\nu')$.
We call this the \emph{cut graph}. Algorithm~\ref{alg:pseudocode-main} details
how this graph can be used to compute optimal piecewise models.

\begin{figure}
  \centering
  \includegraphics[width=0.5\textwidth]{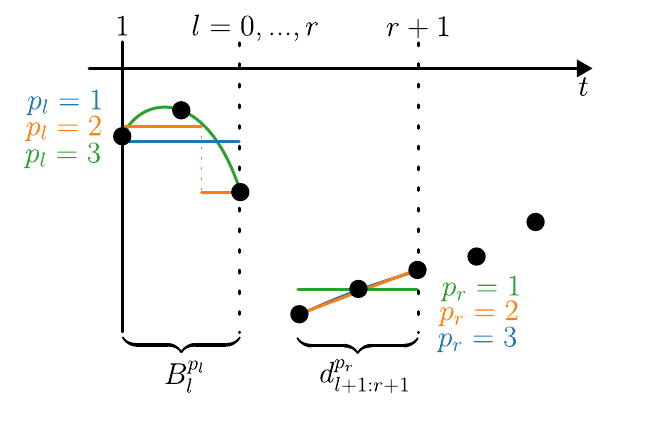}
  \caption{A visualization of the dynamic programm for solving the $\nu$-degree of freedom partition problem \eqref{main:eq:degree_of_freedom_partition_problem}.
    Candidates for the optimal solution on data $1:r+1$ with $\nu$ degrees of freedom
    are composed of the best (continuous) approximation on data $l+1:r+1$ with $p_R$ degrees of freedom
  and the best (piecewise) solution on data $1:l$ with $\nu - p_R$ degrees of freedom.}
\end{figure}

\subsection{Obtaining minimizers by backtracking and resolving ambiguities}

Knowing the regularization paths,
we are also interested in finding the corresponding solutions, i.e. the partition of the data itself $P^*_\gamma,$ the corresponding sequence of degrees of freedom $\lambda^*_\gamma,$ and  the piecewise regression function $\omega^*_\gamma$.
This can be achieved in two steps: first, backtracking
to obtain the results of \eqref{main:eq:degree_of_freedom_partition_problem}
for each $\nu = 1, \ldots, n,$
and second performing a lookup using the above correspondence between $\gamma$ and $\nu.$

The backtracking works as follows: When solving \eqref{main:eq:degree_of_freedom_partition_problem},
we store  for each feasible $\nu =1, \ldots, n,$ and $r = 1, \ldots,n$ the minimizing arguments
of \eqref{main:eq:bellman_recursion} (that is, the leftmost boundary index of the segment ending with index $r$ and the degrees of freedom on this segment) to construct the cut graph.
With these informations, an optimal partition, the corresponding sequence of degrees of freedom, and a regression function with $\nu$ degrees of freedom is computed
by following the edges in the cut graph until a terminal node is reached.
This gives a minimizer of the constrained problem \eqref{main:eq:degree_of_freedom_partition_problem}
for each $\nu =1, \ldots, n.$ A more detailed description is included in Algorithm~\ref{alg:pseudocode-main}.

As for the lookup,  a solution for \eqref{main:eq:proposed_model_intro}  for a specific $\gamma \geq 0$ is obtained
by computing the minimizing argument  of \eqref{main:eq:tabulation_gamma}
and returning the corresponding solution of \eqref{main:eq:degree_of_freedom_partition_problem}.

Recall from Section \ref{main:ssec:uniqueness} that the minimizers of the studied model are unique (in an a.e.-sense) in
certain least squares settings,
and that the minimizers  are not unique in general.
We now describe how to deal with potential ambiguities.
First, to rule out general ambiguity issues with function estimation which do not stem from the piecewise estimation, we assume that, on each segment $I$,
the fitted functions $\phi^{\lambda_I} = \argmin_{\omega \in \Omega^{\leq {\lambda_I}}}d_I( \omega, y)$ are either unique, or admit a canonical choice in case of non-uniqueness. Then, as mentioned in Section~\ref{main:ssec:uniqueness}, an optimal solution $(P^*_\gamma, \lambda^*_\gamma)$ of \eqref{main:eq:proposed_model_intro} defines a piecewise function $\omega_{\gamma}^*$
uniquely up to shift of the break locations between the borders of two adjacent segments of the partition.
Now inspecting again Example~\ref{ex:non-uniqueness},
it seems that  none of the models is clearly better than the others, and the model consisting of only one segment seems attractive for parsimony reasons. However merely trying to select the partition with the lowest number of segments does not resolve ambiguities in general as we can see by considering a similar example with values $0,1,0,1$ and $\nu=4$.
Here, we employ the following selection strategy
between equally good models during backtracking in the dynamic program, which we call \emph{right-maximal graph-tracing} and call uniqueness results that rely on this strategy \emph{\sUnique}:
To determine a partition of $1:r$ with a corresponding sequence $\lambda \in \Lambda^{\nu}(P)$ we start with a {budget} of $\nu$ degrees of freedom. We then select the longest rightmost segment of our partition, which corresponds to choosing the lowest $l$ among the minimizers in the recursion from Theorem \ref{main:thm:bellman-recursion}. This in turn uniquely determines that we spend $p_R$ dofs on $l+1:r$. We proceed recursively downwards to reconstruct the remainder of $P$ and $\lambda$.
An implementation may choose another way to achieve uniqueness at this step. One obvious option to improve results would be to use an additional forward cross-validation at this step (for $r < n$) to cut down the number of options before picking a single representative. However due to the additional incurred cost of such methods our implementation uses the basic scheme just proposed.

Given this algorithm we may now associate to each $B_r^\nu$ a unique partition and dof-sequence. This allows us to make the following statement:

\begin{theorem}\label{main:thm:solution-pcw-const}
  There is a finite partition of $\R_{\geq 0}$ into intervals such that the solution to the $\gamma$-penalized partition problem is a function of $\gamma$ that is constant on each element of the partition. This solution is \sUnique except for finitely many penalties $\gamma$ -- so it is in particular unique almost everywhere.
\end{theorem}
The proof is given in the supplementary material.

\subsection{Parameter selection by rolling cross-validation with OSE-rule}\label{main:sec:parameter_choice}

The hyperparameter $\gamma$ is chosen via rolling cross validation,
a type of cross validation that respects the temporal order of the time-series and avoids look-ahead bias.
It is given as follows:
\begin{align}\label{main:eq:cv_score_intro}
  \CV(\gamma) = \frac{1}{n-1}\sum_{r=1}^{n-1} (\omega^{*}_{\gamma, r}(t_{r+1}) - y_{r+1})^2.
\end{align}
Here $\omega^{*}_{\gamma, r}$ is the piecewise regression function associated to  $(P^{*}_{\gamma, r}, \lambda^*_{\gamma, r})$ which denotes a solution of \eqref{main:eq:proposed_model_intro} on the partial data $(t_1, y_1), \ldots, (t_r, y_r).$
(One may use other distance functions than the mean squared error.)
An optimal parameter in the sense of cross validation is a value that minimizes $\gamma \mapsto \CV(\gamma).$

To obtain such an optimal parameter,  we use the following key observation:
As the mapping $\gamma \mapsto (P^{*}_{\gamma, r}, \lambda^*_{\gamma, r})$ is piecewise constant,
so is the $r$-th summand, $\CV_r(\gamma) = (\omega^{*}_{\gamma, r}(t_{r+1}) - y_{r+1})^2,$ in \eqref{main:eq:cv_score_intro}.
Hence $\CV$ is piecewise constant as the mean of finitely many piecewise constant functions.
Likewise, on full data, the mapping $\gamma \mapsto (P^{*}_{\gamma, n}, \lambda^*_{\gamma, n})$ is piecewise constant.
Thus, the joint mapping $\gamma \mapsto [ CV(\gamma), P^{*}_{\gamma, n}, \lambda^*_{\gamma, n}]$ is piecewise constant as well.
By the piecewise constancy, all hyperparameters
of a piece of the mapping can be considered as equivalent,
and we may represent them by a single value, say the midpoints of the pieces, denoted by $\{ \gamma_1, \ldots, \gamma_L\}.$
To obtain a unique result, we invoke the principle of parsimony and select
the largest minimizing argument, we denote this choice by $\gamma_{\CV}$:
\begin{equation}\label{main:eq:CV_min}
  \gamma_{\CV} = \max{(\argmin_{\gamma \in \{ \gamma_1, \ldots, \gamma_L\}} \CV(\gamma)}).
\end{equation}

To obtain even more parsimonious results it is common to apply the \emph{one standard error rule (OSE)} to cross-validated models rather than directly choosing the model with the strictly lowest score
\cite[p. 244]{book:elements-of-statistical-learning}. For the considered model this means taking the largest parameter such that the CV-score is within a one-standard error window of the minimum value, so
\begin{equation}\label{main:eq:CV_OSE}
  \gamma_{\text{OSE}} = \max_{\gamma \in  \{ \gamma_1, \ldots, \gamma_L\}}\{ \CV(\gamma) \leq \CV(\gamma_{\CV}) + \mathrm{SE(\gamma_{\CV})}\},
\end{equation}
where $\mathrm{SE}(\gamma) = \sqrt{\mathrm{Var}(\CV_1(\gamma), \ldots, \CV_{n-1}(\gamma))}/\sqrt{n-1}$ is the standard error of the mean in \eqref{main:eq:cv_score_intro}, see \cite{paper:the-one-standard-error-rule-for-model-selection}.
(As $\gamma_{\text{CV}},$ the parameter $ \gamma_{\text{OSE}}$ is a representative
for an interval of hyperparameters which are all associated to the same model and to the same score.)
The CV scoring function and
the solutions of the two hyperparameter choices
for the example of Figure~\ref{main:fig:comparison_pcw_DofPPR} are given in Figure~\ref{main:fig:cv_score_toy}.

A few important points about the above selection methods:
\emph{(i)} Computing $\gamma_{\CV}$ and $\gamma_{\text{OSE}}$ comes with litte additional effort.
This is because  the regularization paths have been computed by the algorithm proposed in further above.
Furthermore, we do not need to compute $\omega^{*}_{\gamma, r}$ on its entire domain; it is sufficient to compute the $\omega^{*}_{\gamma, r}$ on  the rightmost segment of the defining partition.
\textit{(ii)} The resulting values are computed exactly up to the usual numerical errors of floating point arithmetic if the computation of the tabulations and the model fits admit these precisions.
\textit{(iii)}  The selection strategies are invariant to (global) scaling of the signal.

\begin{remark}
  A natural question arising in the above procedure
  is why  crossvalidation is done with respect to
  $\gamma$ based on \eqref{main:eq:proposed_model_intro} and not with respect to the total number of degrees of freedom $\nu_{\mathrm{total}}$ based on \eqref{main:eq:proposed_model_reformulated}.
  The reason is that rolling cross-validation works with signals of different lengths ($r =1, ..., n-1$) and
  the  total number of degrees of freedom has a different meaning for each of those signals.
\end{remark}

\begin{figure}
  \centering
  \includegraphics[width=1\textwidth]{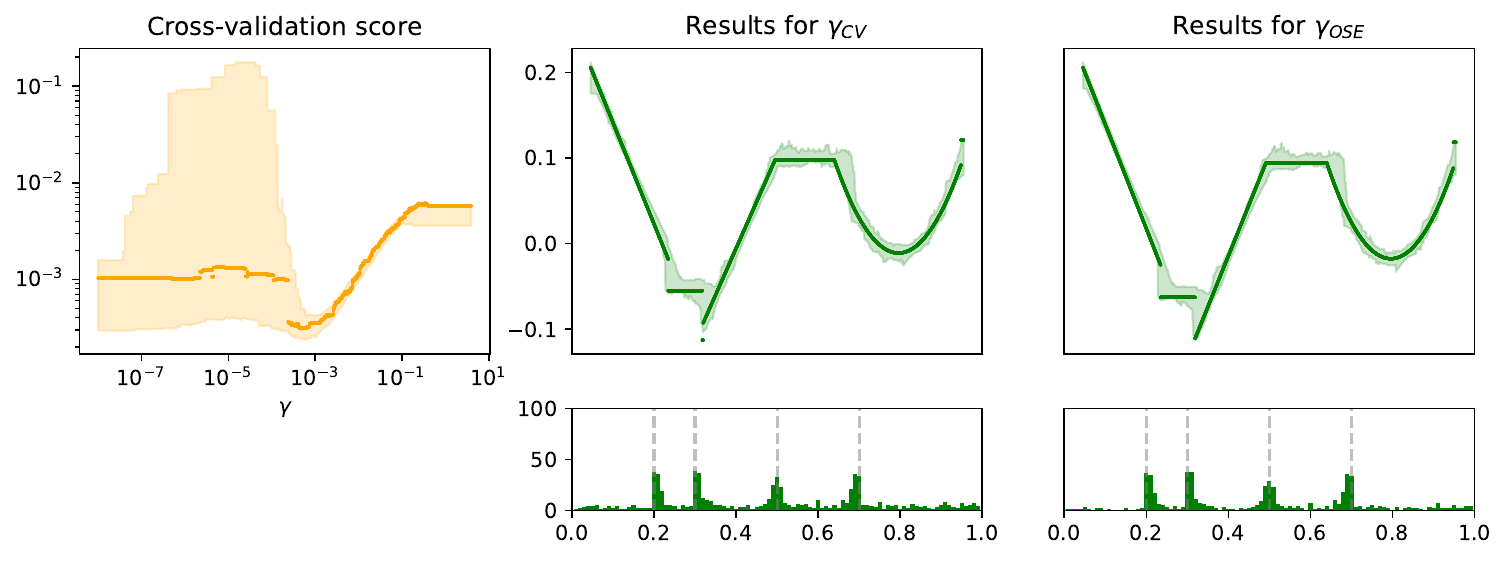}
  \caption{
    Visualization of the proposed parameter selection.
    \emph{Left:} The piecewise constant CV scoring function for the sample realization of Figure~\ref{main:fig:comparison_pcw_DofPPR}
    (solid line)
    and the $2.5 \%$ to $97.5 \%$ quantiles over all realizations (shaded).
    \emph{Middle and right:}
    The corresponding results
    for the parameter choices $\gamma_{\text{CV}}$
    and $\gamma_{\text{OSE}},$ respectively.
    The legend is identical to that of Figure~\ref{main:fig:comparison_pcw_DofPPR}.
    The histograms of the breakpoints indicate a slightly
    better localization of the breakpoints near the true breaks.
  }
  \label{main:fig:cv_score_toy}
\end{figure}

To further investigate the  parameter choice strategy,
we use two more piecewise polynomial test signals with two jumps which are designed such that
one jump is considerably larger that the other (Figure~\ref{main:fig:accuracy}, top row).
The count of dofs and jumps that are estimated accurately versus the sample size
is reported for 500 realizations using uniformly sampled locations and Gaussian noise of standard deviation 0.05.
We observe that the number of accurately estimated dofs
increases with the number of samples for both signals.
This trend is also observed for the number of
accurately estimated jumps in the first signal.
For the second signal, the number of jumps
has a tendency to be overestimated.
This can be attributed to the fact
that the underlying signal  can
be well approximated by a piecewise function
of a linear and a constant polynomial on the last segment
(compare also the estimates in the top row of Figure~\ref{main:fig:cv_vs_time}).
The results using $\gamma_{\text{CV}}$
and $\gamma_{\text{OSE}}$ are relatively close
together, and $\gamma_{\text{OSE}}$ shows a slight advantage over $\gamma_{\text{CV}}$
except for the dofs of the second signal.

\begin{figure}
  \centering
  \begin{subfigure}{\textwidth}
    \includegraphics[width=0.4\textwidth]{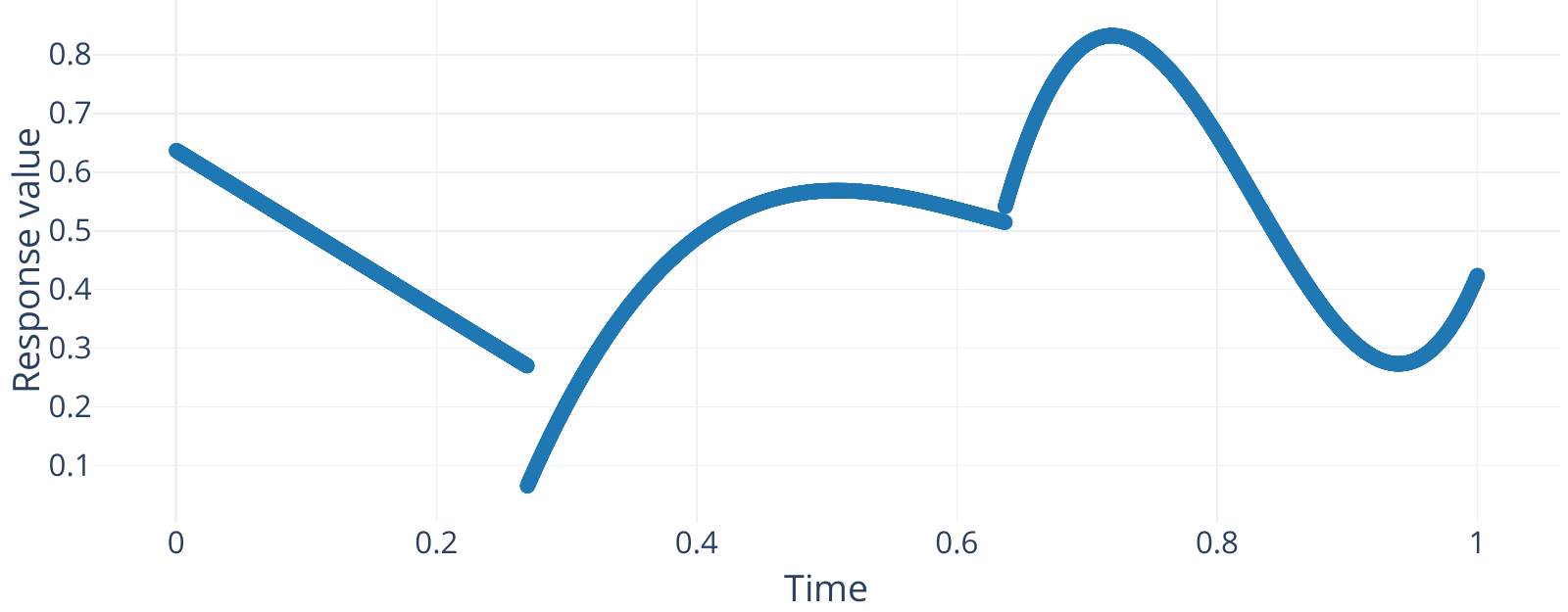}
    \hspace{0.1\textwidth}
    \includegraphics[width=0.4\textwidth]{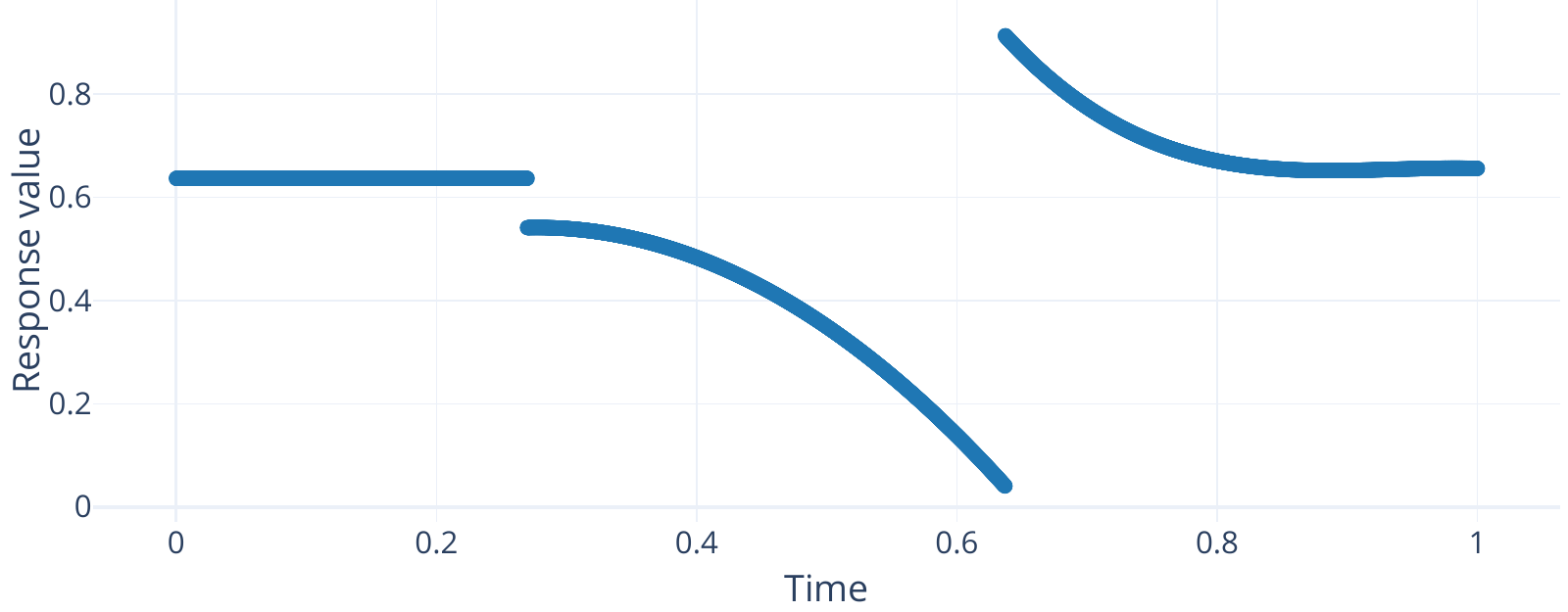}
    %\caption{Test signals (groundtruth).}
  \end{subfigure}
  \\
  \begin{subfigure}{\textwidth}
    \includegraphics[width=0.49\textwidth]{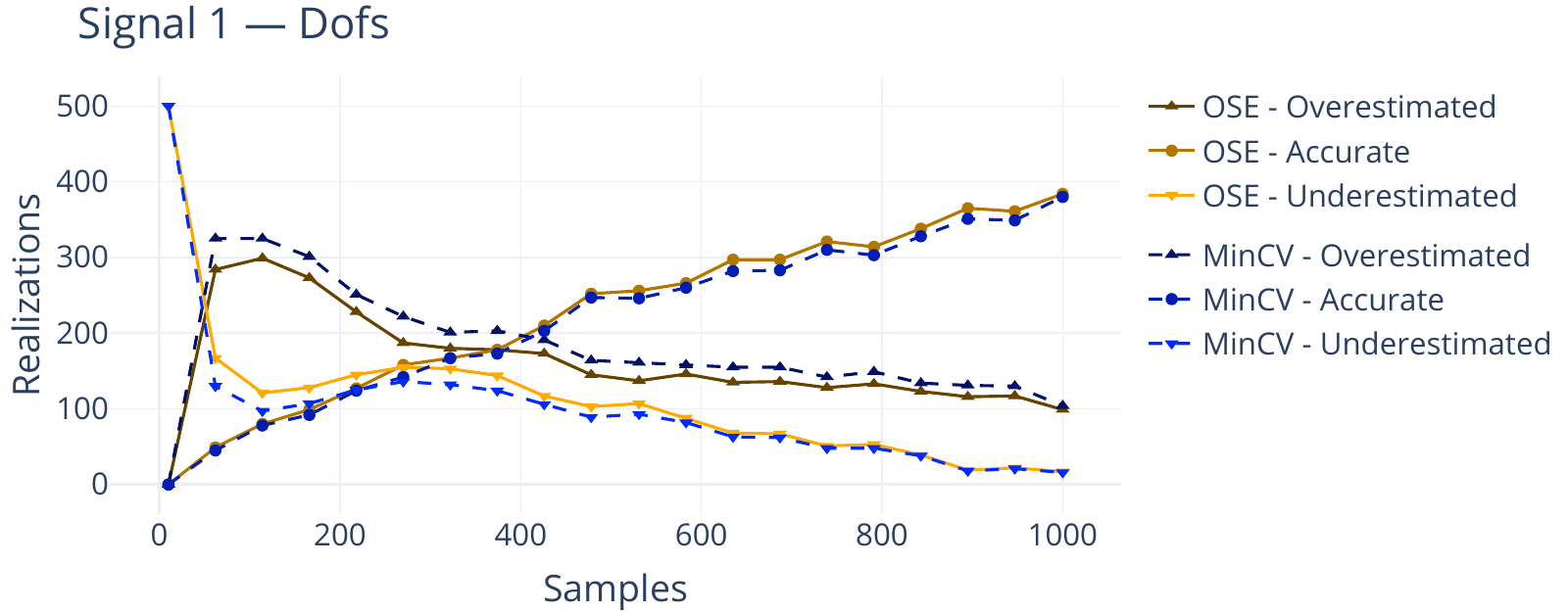}
    \includegraphics[width=0.49\textwidth]{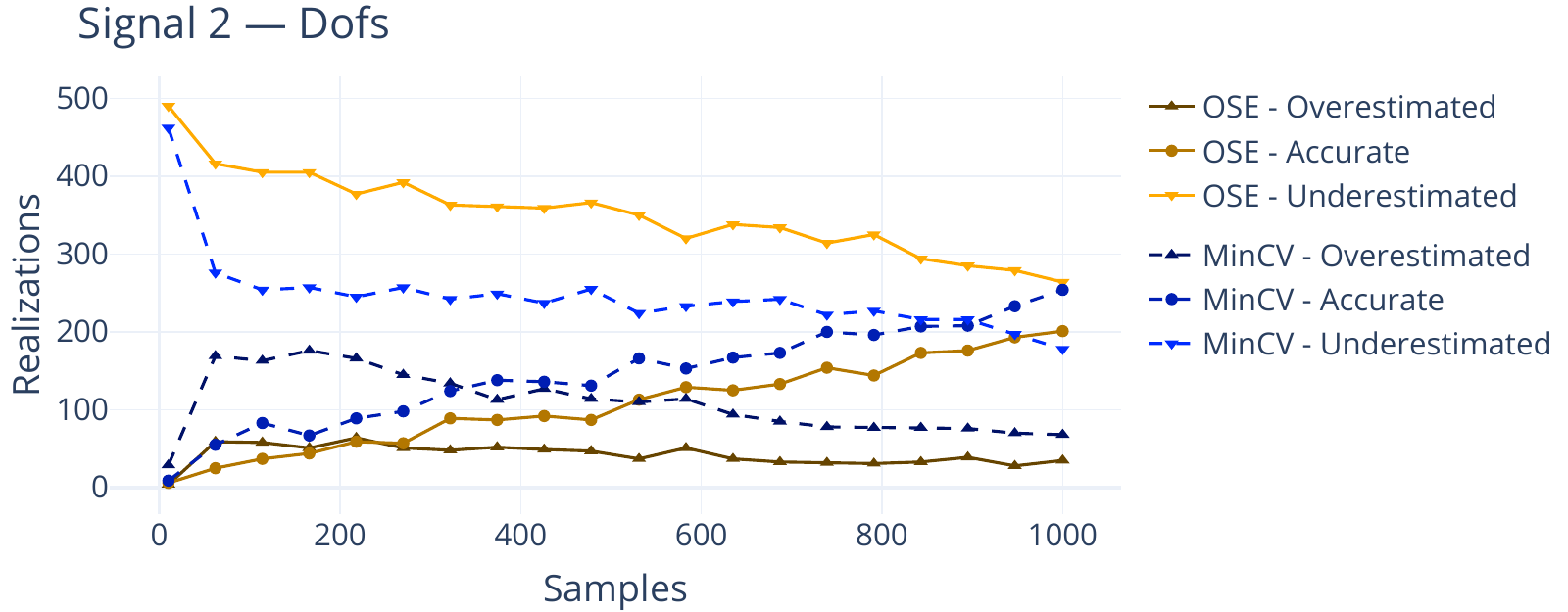}
    %\caption{Count of realizations where dofs are underestimated, accurately estimated, or overestimated vs number of samples.}
  \end{subfigure}
  \\
  \begin{subfigure}{\textwidth}
    \includegraphics[width=0.49\textwidth]{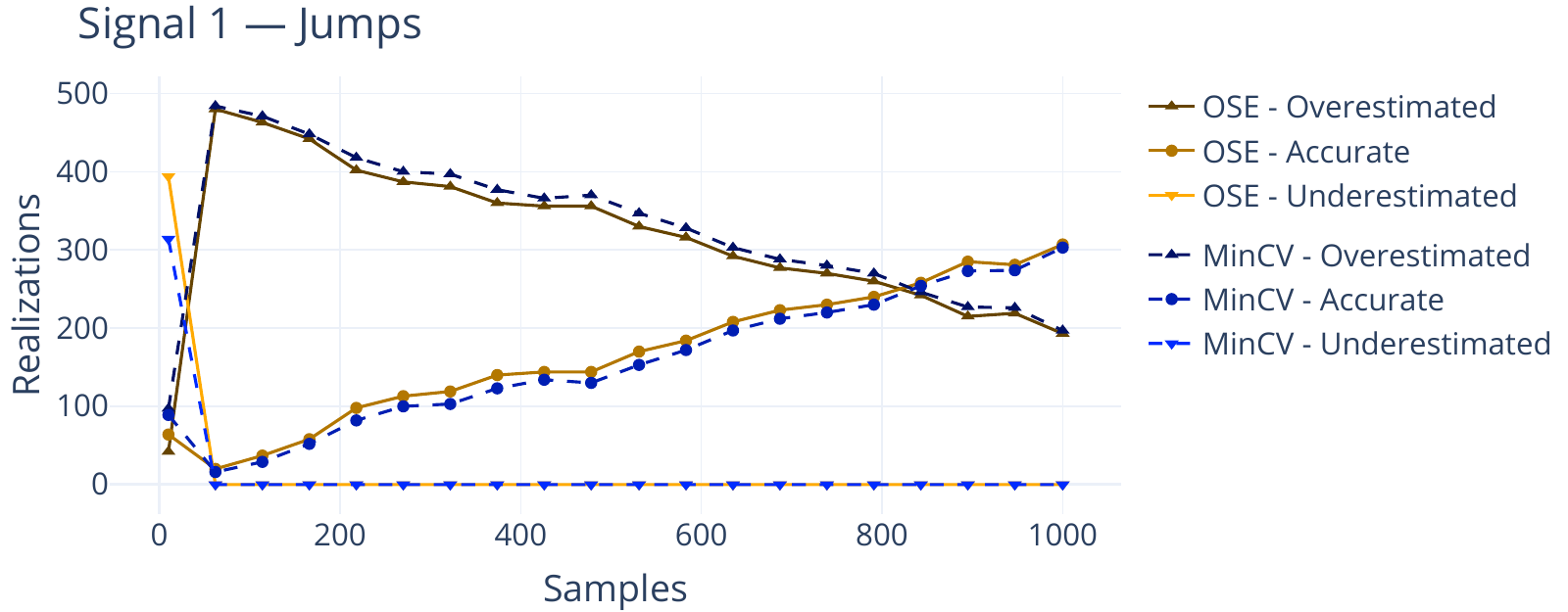}
    \includegraphics[width=0.49\textwidth]{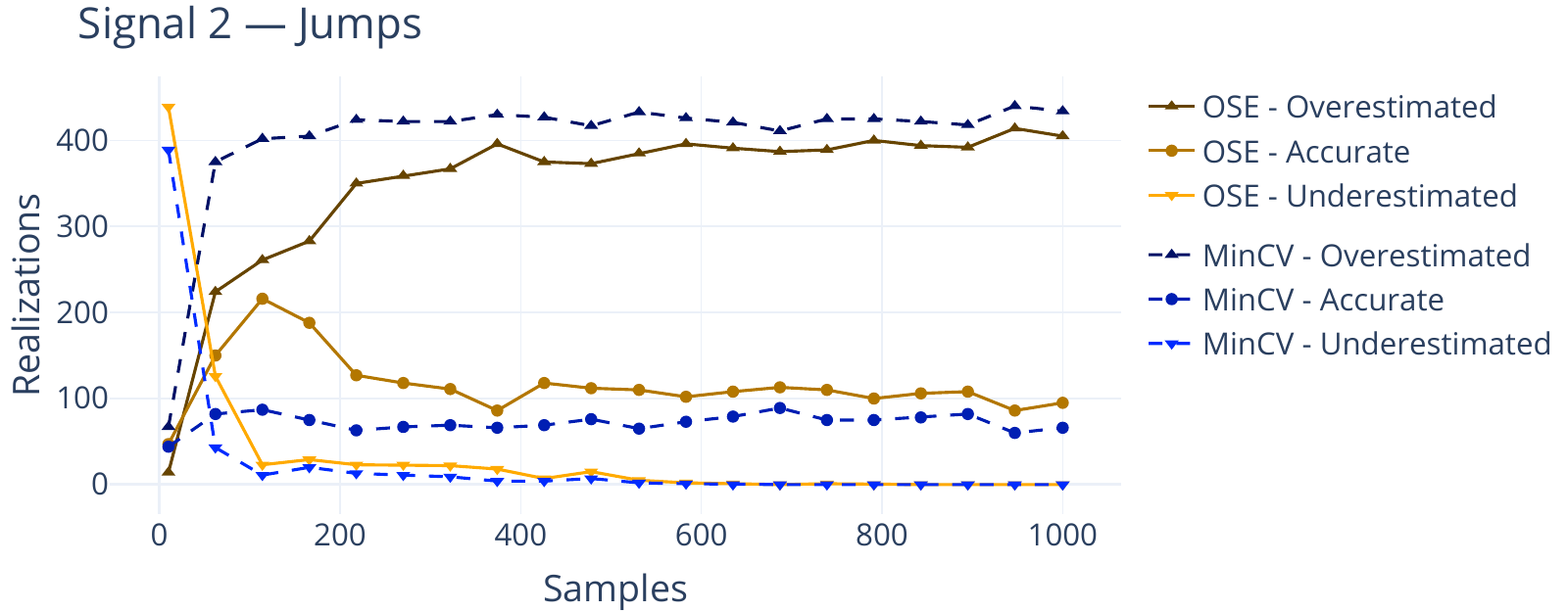}
    %\caption{Count of realizations where jumps are underestimated, accurately estimated, or overestimated vs number of samples.}
  \end{subfigure}
  \caption{CV score based parameter selection
    for two different signals.
    \emph{Top row:} True signals. \emph{Middle row:}
    Count of realizations where dofs are underestimated,
    accurately estimated, or overestimated versus the number of samples.
    \emph{Bottom row:} Count of realizations where jumps are underestimated, accurately estimated, or overestimated versus the number of samples.
  }
  \label{main:fig:accuracy}
\end{figure}

\cite{pein2021cross} have investigated a related CV based parameter selection strategy on a related problem, namely two-fold cross validation (of even/odd indexed points) for piecewise constant regression. They found that an $L^2$ based cross validation score may result in an underestimation in configurations where a dominant change overshadows a nearby smaller jump, and that this is improved by using an $L^1$ based CV score.
To investigate if an $L^1$ variant of the rolling CV score shows a similar behaviour in the piecewise polynomial case,
we look at the CV scores as a function of the penalty factor $\gamma$ and  the length of the signal part ($(y_r)_{r=1, \ldots, R}$ for $R = 1,\ldots,n-1.$
  Figure~\ref{main:fig:cv_vs_time} shows the results for
  the quadratic CV score \eqref{main:eq:cv_score_intro} and the $L^1$ variant
  where the square is replaced by the absolute value.
  (Note that the data term in \eqref{main:eq:proposed_model_intro_poly}
  remains quadratic in both variants.)
  In the present examples, the curves of the optimal $\gamma$-values have a relatively similar shape, and in particular the final models estimated on the full data exhibit only  minor differences (the $L^2$ solution shows one extra break in the second signal, the $L^1$  solution two extra breaks).
  We also report the results analogous to Figure~\ref{main:fig:accuracy} for $L^1$ CV data term in the supplementary material.
  It can be observed that the absolute value based rolling CV gives better results in the first signal,
  while the quadratic variant performs better for the second signal.
  The empirical results indicate that neither variant demonstrates a definitive advantage over the other in the current setup.
  \begin{figure}
    \centering
    \begin{subfigure}{\textwidth}
      \includegraphics[width=0.49\textwidth]{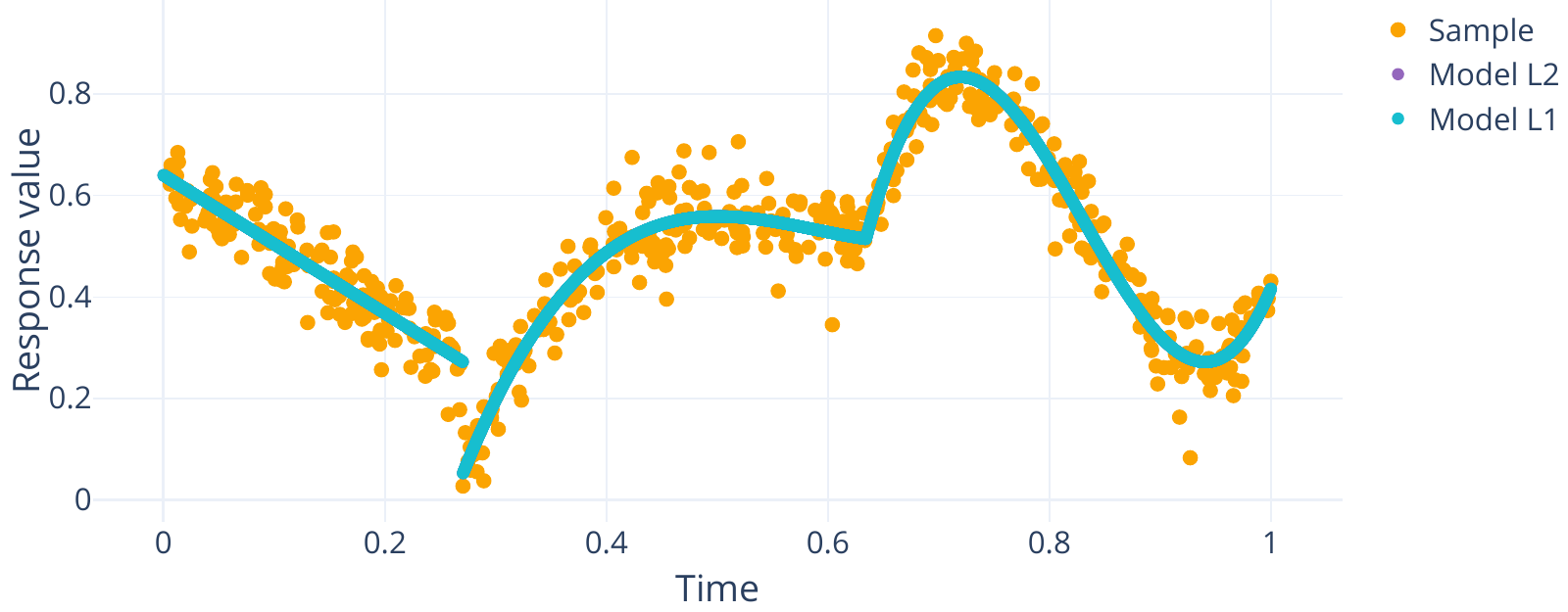} \hfill
      \includegraphics[width=0.49\textwidth]{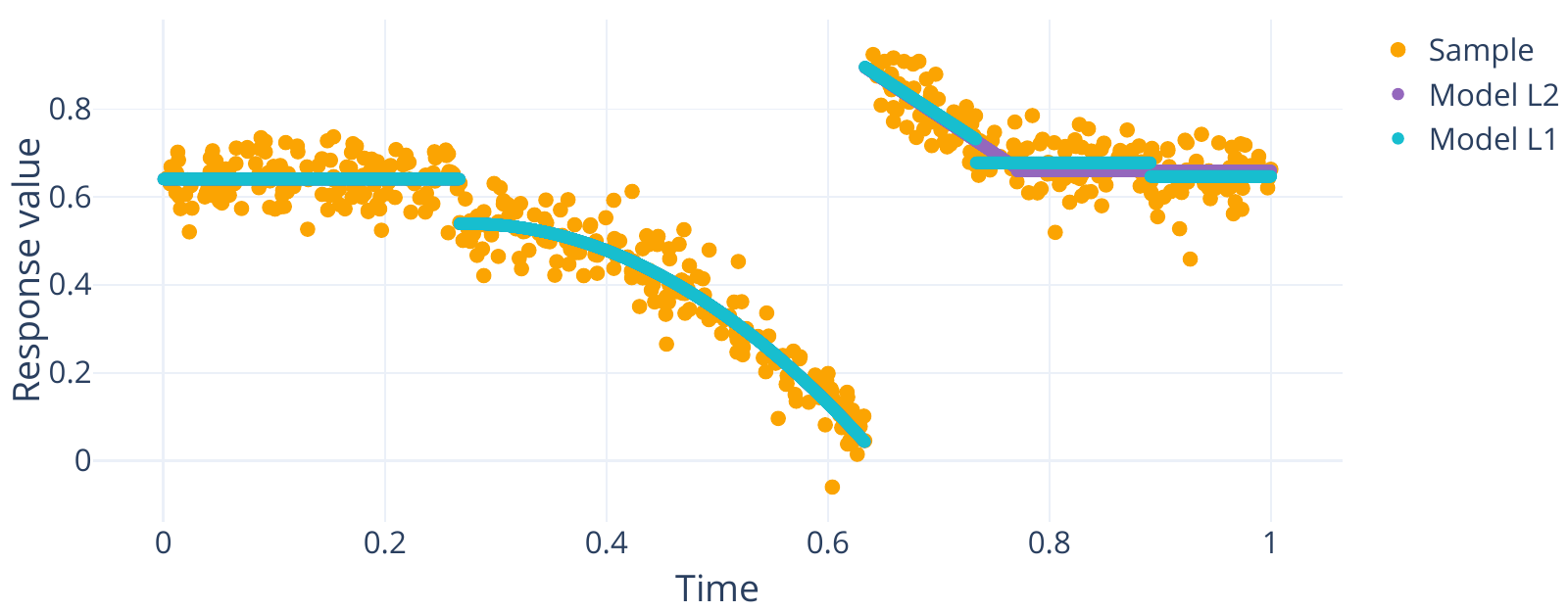}
      %\caption{}
    \end{subfigure}
    \\
    \begin{subfigure}{\textwidth}
      \includegraphics[width=0.46\textwidth]{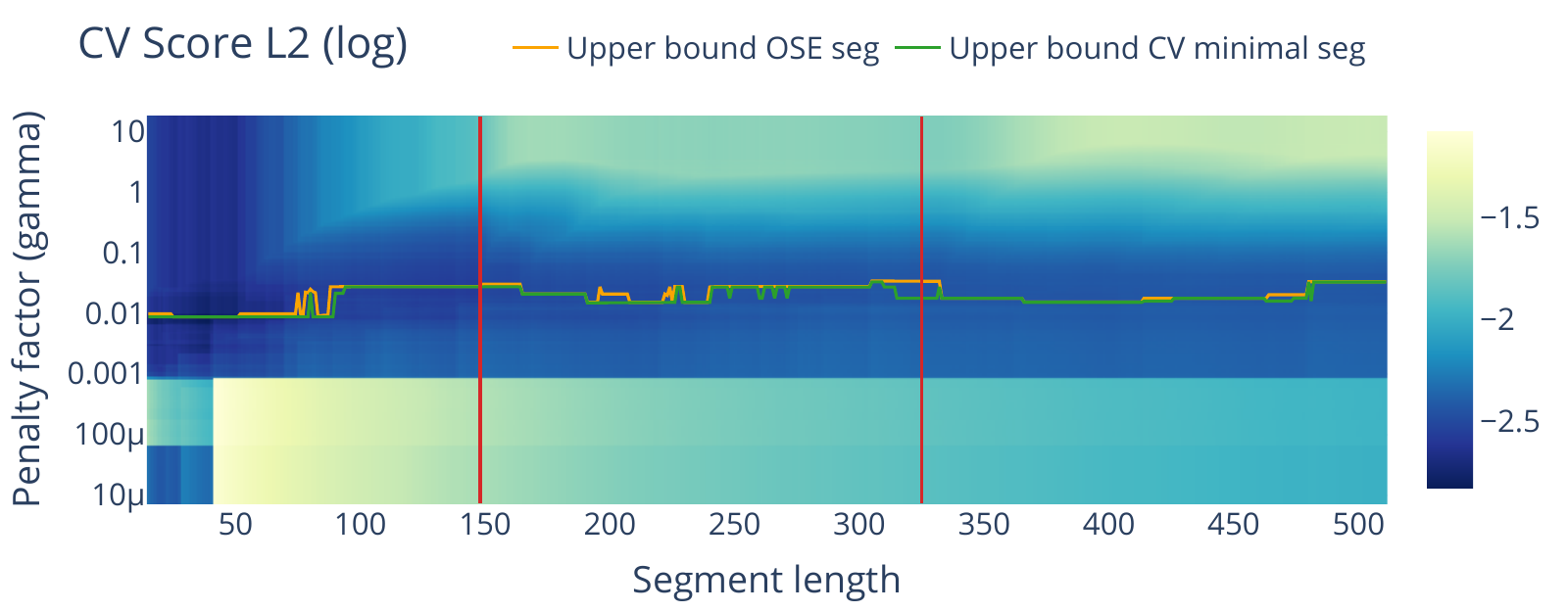}
      \hspace{0.04\textwidth}
      \includegraphics[width=0.46\textwidth]{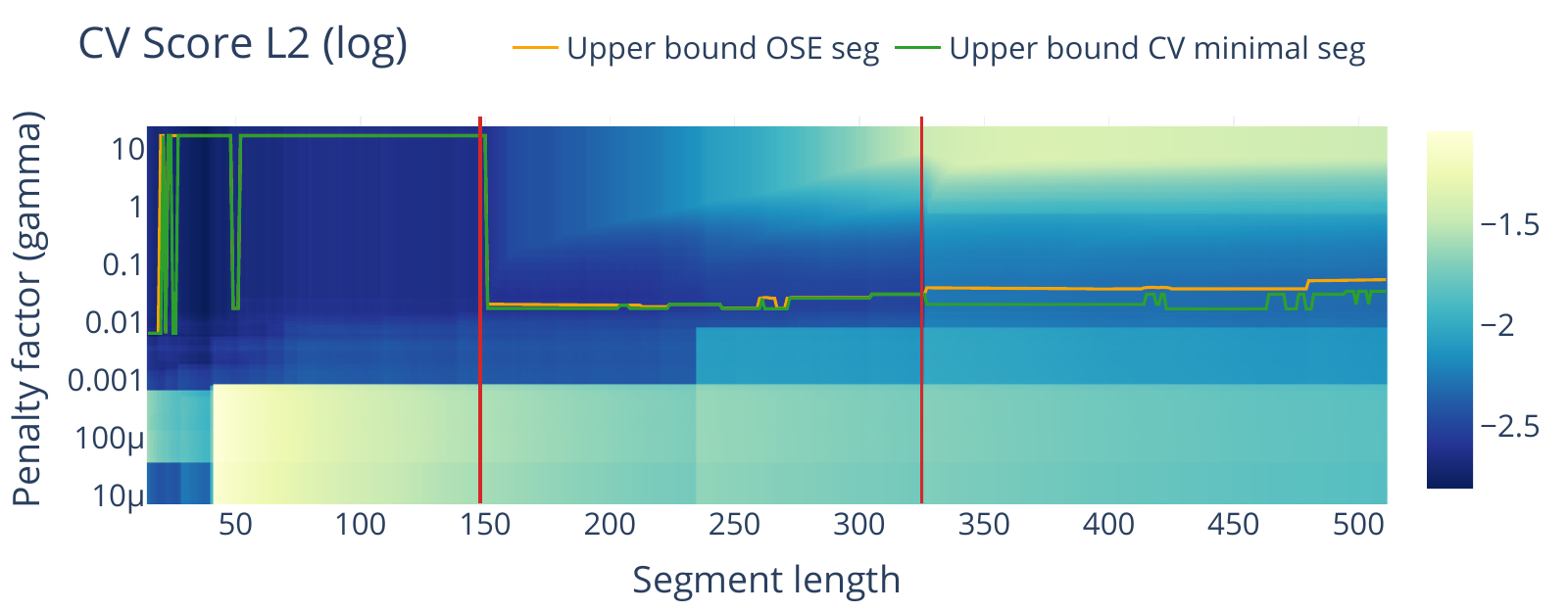}
      %\caption{}
    \end{subfigure}
    \\
    \begin{subfigure}{\textwidth}
      \includegraphics[width=0.46\textwidth]{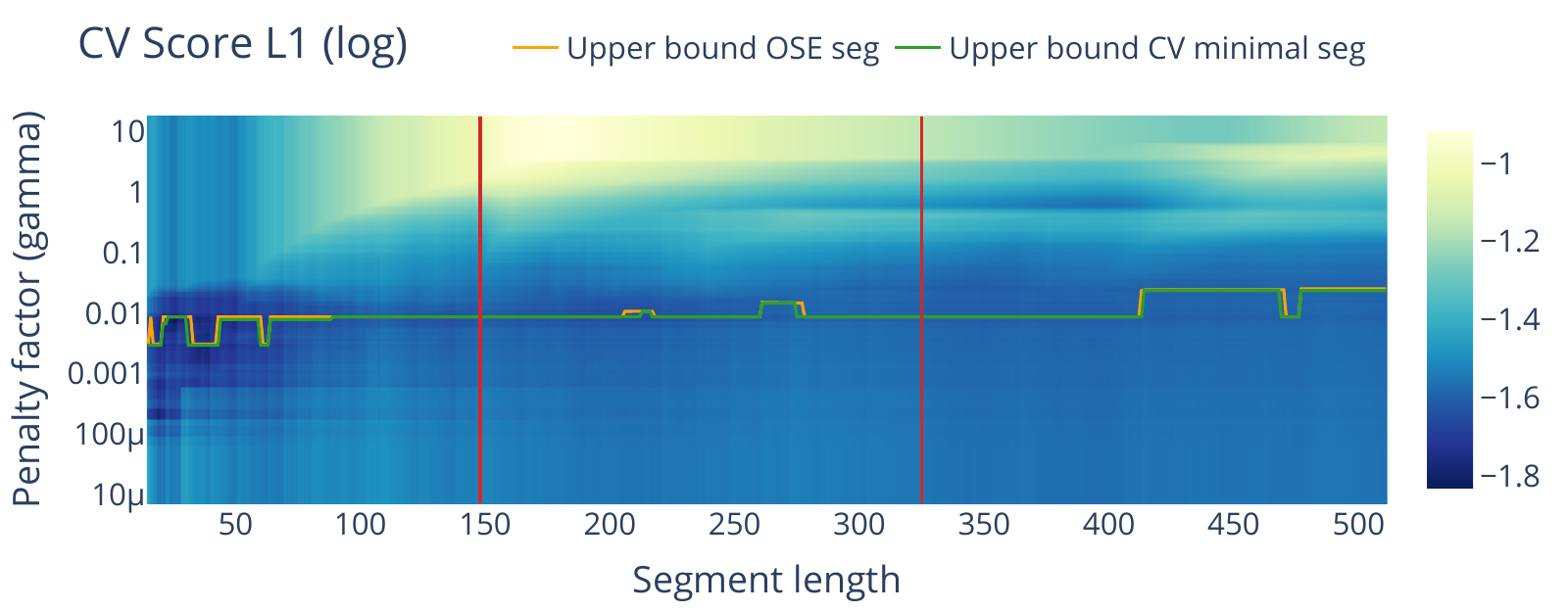}
      \hspace{0.04\textwidth}
      \includegraphics[width=0.46\textwidth]{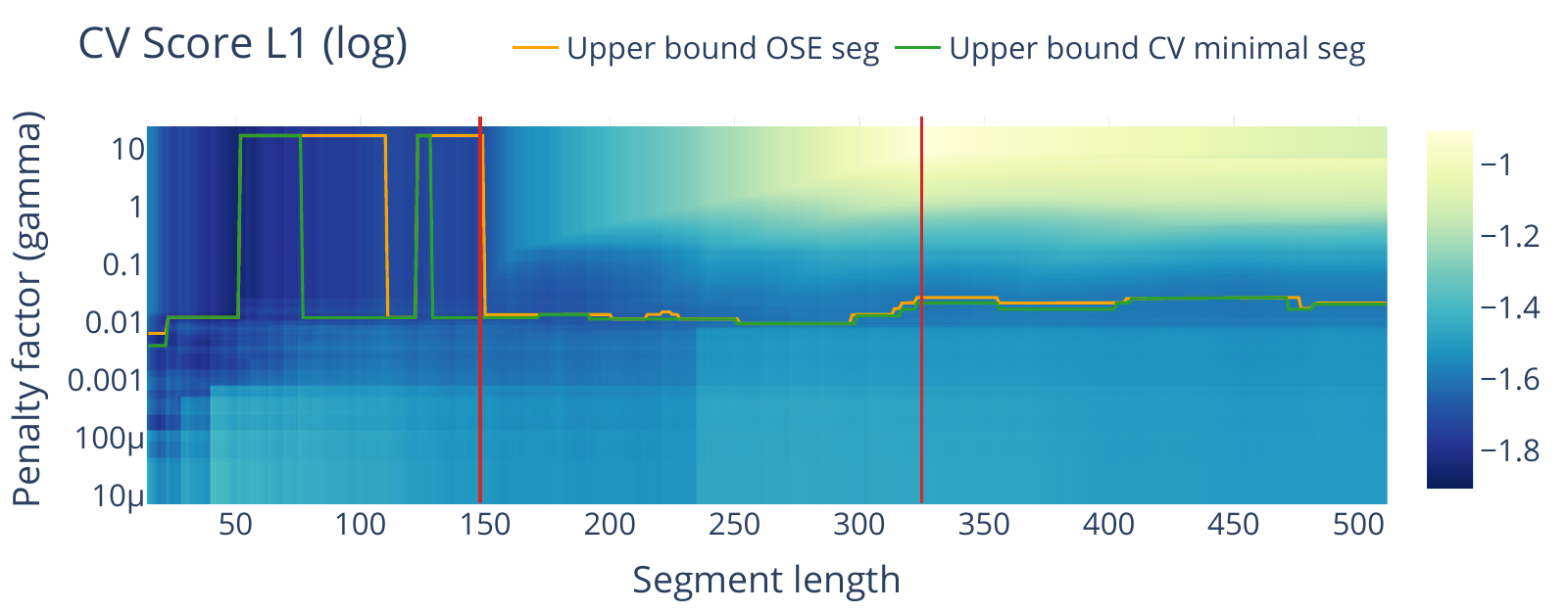}
      %\caption{}
    \end{subfigure}
    \caption{
      \emph{Top row:} Realizations with 500 samples (dotted)
      of the signals used in Figure~\ref{main:fig:accuracy},
      and OSE solutions using $L^1$ and $L^2$ CV scores (solid).
      \emph{Middle row:} CV score (color, on log scale) vs. penalty factor $\gamma$ and length of the signal part ($(y_r)_{r=1, \ldots, R}$ for $R = 1,\ldots,n-1$).
      The vertical red lines depict the true jumps.
      Curves from left to right depict the optimal parameters  $\gamma_{\mathrm{CV}}$ and $\gamma_{\mathrm{OSE}}.$
      In the second signal, these curve exhibits a clear drop after the first jump because the second segments needs more dofs.
      \emph{Bottom row:} The same experiment but using the absolute value in the CV score.
    }
    \label{main:fig:cv_vs_time}
  \end{figure}

  \subsection{Complete algorithm and  analysis of the computational complexity}\label{main:sec:total-time-complexity}

  Algorithm~\ref{alg:pseudocode-main} summarizes the main steps of the proposed algorithm that computes
  the regularization paths of \eqref{main:eq:proposed_model_intro}, the hyperparameters $\gamma_{\text{CV}}$ and $\gamma_{\text{OSE}},$
  and corresponding piecewise regression functions.
  \begin{algorithm}
    \begin{algorithmic}[1]
  \Input
  \State {Timeseries sample given by $(t_i)_{i=1,...,n}$, $(y_i)_{i=1,...,n}$}
  \State {Errors $d_S^\nu$ of local models for feasible segments $S$ and degrees of freedom $\nu$}
  \State {Optional user parameters: maximal total and local degrees of freedom $\nu_\text{total},$
  $\nu_\text{max}$}
  \EndInput
  \Output
  \State {Piecewise model $\omega^*_\gamma$ for timeseries data}
  \EndOutput

\item[]
  \For{feasible $r \in 1:n, \nu \in 1:r$} \Comment{Compute tabulation $B$ of Bellman values}
  \State{\parbox[t]{.9\linewidth}{Compute $B_r^\nu$ using the recurrence in \eqref{main:eq:bellman_recursion}.}}
  \EndFor
  % \State{\parbox[t]{.9\linewidth}{Note that the Bellman table generates a graph (an acyclic directed graph where every
  %         node has at most one outgoing edge -- \ie a polyforest with edges oriented towards the roots):
  %         for each pair $(r, \nu)$ from the previous loop there's a node in the graph.
  %         If the minimizing arguments of $\eqref{eq:bellman_recursion}$ indicate that
  %         there is only one segment in the partition (so the segment $0:r$ is best approximated using a single model
  %         with $\nu$ degrees of freedom) then the node $(r,\nu)$ in the graph has no outgoing edge.
  %         If it indicates that there is a prior segment with node $(r',\nu')$ then there's an edge from $(r,\nu)$ to
  %         $(r',\nu')$.
  %         We'll call this the cut graph.
  %     }}
  \For{$r \in 1:n$}
  \State{\parbox[t]{.9\linewidth}{Use the tabulation $B$ and~\eqref{main:eq:tabulation_gamma} to compute the
      critical values of the piecewise affine function $\gamma \mapsto \min_{\nu \in 1:r} B_r^\nu + \gamma \nu$
      using algorithm described in supplementary Section~\ref{sup:sec:pointwise}.
      This provides the correspondence of the parameters $\gamma$ of \eqref{main:eq:proposed_model_intro}
      and $\nu$ of \eqref{main:eq:degree_of_freedom_partition_problem} for each partial data with
  indices $1:r.$}}
  \State{\parbox[t]{.9\linewidth}{Determine the solution of \eqref{main:eq:degree_of_freedom_partition_problem}
  for each $\nu \in 1:r$ by backtracking, resolving ambiguities by RMGT as follows:}}
  \State{Start at node $(r, \nu)$ in the cut graph}
  \While{there is an outgoing edge $(r', \nu')$ at the current node}
  \State{\parbox[t]{.9\linewidth}{Add segment $r'+1:r$ with
  $\nu - \nu'$ degrees of freedom to the current partition and set the the current node to $(r',\nu')$.}}
  \EndWhile
  \State{Add the final segment to the partition.}
  \EndFor
  \For{$r \in 1:n$}
  \State{\parbox[t]{.9\linewidth}{Use the correspondence between $\gamma$ and $\nu$ and the solutions to
      problem~\eqref{main:eq:degree_of_freedom_partition_problem} to determine the piecewise constant mapping
      $\gamma \mapsto \omega^{*}_{\gamma, r},$ and compute the piecewise constant mapping
      $\CV_r: \gamma \mapsto (\omega^{*}_{\gamma, r}(t_{r+1}) - y_{r+1})^2.$
      (Note: For this step, it is sufficient to determine $\omega^{*}_{\gamma, r}$ on the rightmost segment of its
  corresponding partition.)}}
  \EndFor
  \State{Compute the piecewise constant mapping $\CV = \frac{1}{n-1}\sum_{r=1}^{n-1} CV_r,$ and determine $\gamma_{\CV}$ and $\gamma_{\text{OSE}}.$}
  \State{Determine the  piecewise regression functions $\omega^*_\gamma$
  for the hyperparameters $\gamma = \gamma_{\CV}$ or $\gamma = \gamma_{\text{OSE}}.$}
\end{algorithmic}

    \caption{The main DofPPR algorithm}
    \label{alg:pseudocode-main}
  \end{algorithm}
  For further details we refer to the supplementary material and  to the commented source code.

  We prove the following central result on the proposed algorithm:
  \begin{theorem}\label{main:thm:complexity}
    The time complexity for the proposed algorithm is in
    $\O(mn \max\{ m, n \Phi(m,n), n^2 \})$
    where  $n \in \N$ is the length of the input timeseries, $m \in \N$ is the maximal number of degrees of freedom to consider for each segment -- so the maximal local model complexity,
    $\Phi(m,n)$ is the cost of fitting a model with $m$ degrees of freedom to data of length $n$.

    In particular, the time complexity is $\O(n^3 m^2)$
    whenever $\Phi \in \O(mn).$
    For polynomials with least squares errors the time complexity is $\O(n^3 m)$.
  \end{theorem}
  The proof is given in the supplementary material.

  \begin{remark}
    Theorem~\ref{main:thm:complexity} contains the $\O(n^3)$ time complexity result for the piecewise constant case  \cite[Proposition 2]{friedrich2008complexity} as  the special case $m=1$.
  \end{remark}

  \section{Implementation, simulation study and applications}\label{main:sec:experiments}

  \subsection{Implementation and experimental setup}

  The implementation may be found at \url{https://github.com/SV-97/pcw-regrs}. The published packages are available at \url{https://crates.io/crates/pcw_regrs} and \url{https://pypi.org/project/pcw-regrs-py/}. The core part of the algorithm is implemented in Rust to provide a high-performance basis. On top of this we provide a native Python extension for ease of use.
  Multicore parallel processing is used to speed up the computations of the residual errors.
  The experiments were conducted on a Linux
  workstation with an AMD Ryzen 9 5900X CPU (4.6\GHz) and with 64GB RAM.

  If not stated differently,
  we use in the least squares instance of the
  DofPPR model \eqref{main:eq:proposed_model_intro_poly} with $\Omega^{\leq \nu}$ being the space of polynomials of maximum degree $\nu-1.$
  In addition to the search space reduction stated in
  \eqref{main:eq:numax_corr},
  we  exclude polynomials of order higher than $10$
  because polynomial regression with high degrees
  may lead to numerical instabilities.
  So in practice, we use $\nu_{\mathrm{max}}(I) = \min(\max(1, \ElemCount{I} - 1), 11).$
  Furthermore, we use by default uniform weights $w = (1, \ldots, 1),$ and the hyperparameter choice $\gamma_{\text{OSE}}$ (cf. \eqref{main:eq:CV_OSE}).
  The method accepts an optional parameter
  $\nu_{\mathrm{total}} \leq n$
  which imposes the additional upper bound on the
  total number of degrees of freedom $\sum_{I \in P} \lambda_I \leq \nu_{\mathrm{total}}$; it is useful when more parsimonious results are desired.
  As mentioned in the beginning of Section \ref{main:ssec:uniqueness},
  the (continuous) breakpoints of the piecewise polynomials can be placed anywhere between two data points belonging to two different segments. We here use the point that minimizes the  distance on the ordinate between the lefthand and the righthand polynomial, and take the midpoint if there is no unique solution.

  In the following, we frequently refer to the Turing change point detection (TCPD) benchmark data set which is specifically designed for changepoint detection studies \citep{paper:an-evaluation-of-change-point-detection-algorithms}. This dataset comprises 37 time series, each of which has been annotated by five human experts to establish the ground truth for changepoint locations.
  The paper of \cite{paper:an-evaluation-of-change-point-detection-algorithms}
  also provides a benchmark that compares 13 diverse changepoint detection algorithms on this dataset.
  The benchmark is designed for unsupervised detection so there is no split into training and test data set.
  More details will be given further below and in the paper of \cite{paper:an-evaluation-of-change-point-detection-algorithms}.

  \subsection{Results on simulated data}

  \begin{figure}[!t]
    \newcommand{\myincludegraphics}[2]{
      \begin{subfigure}[b]{0.45\textwidth}
        \centering
        \includegraphics[width=\textwidth, trim=0 0 150 0, clip]{#1}
        \caption{#2}
      \end{subfigure}
    }
    \myincludegraphics{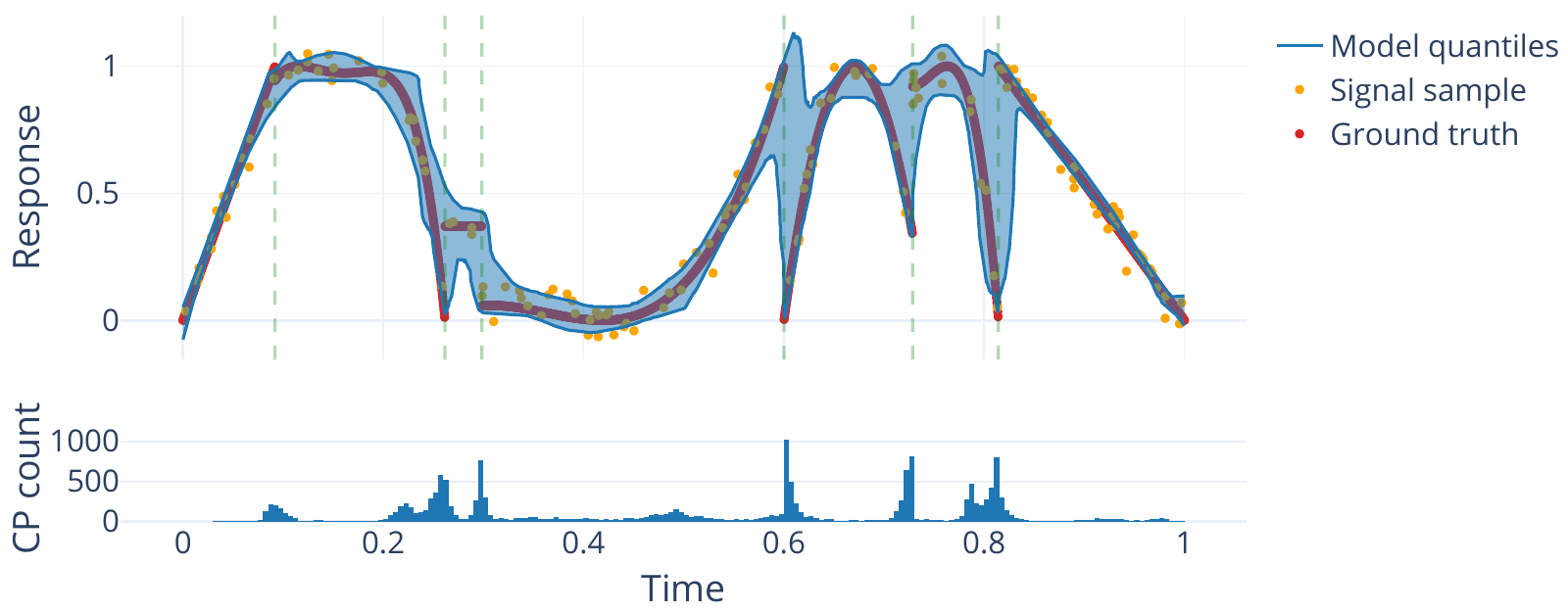}{$N = 150$}\hfill
    \myincludegraphics{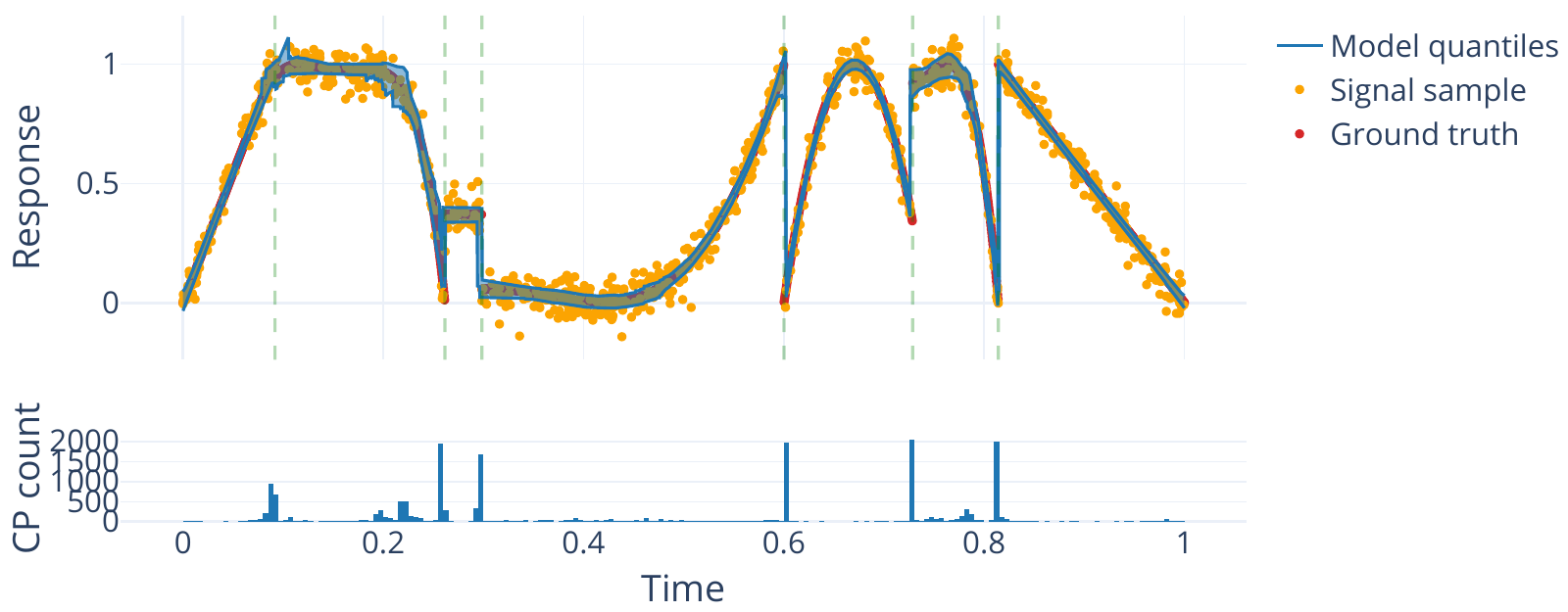}{$\sigma=0.05$}
    \myincludegraphics{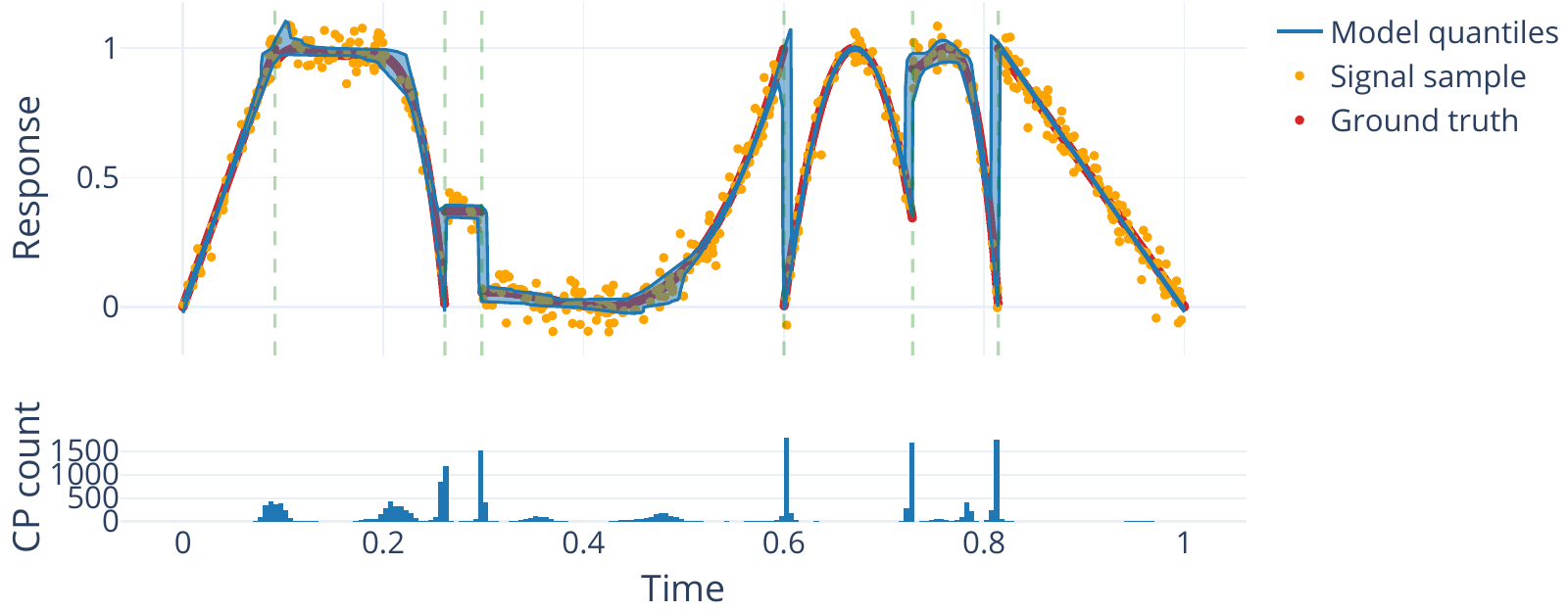}{$N = 500$}\hfill
    \myincludegraphics{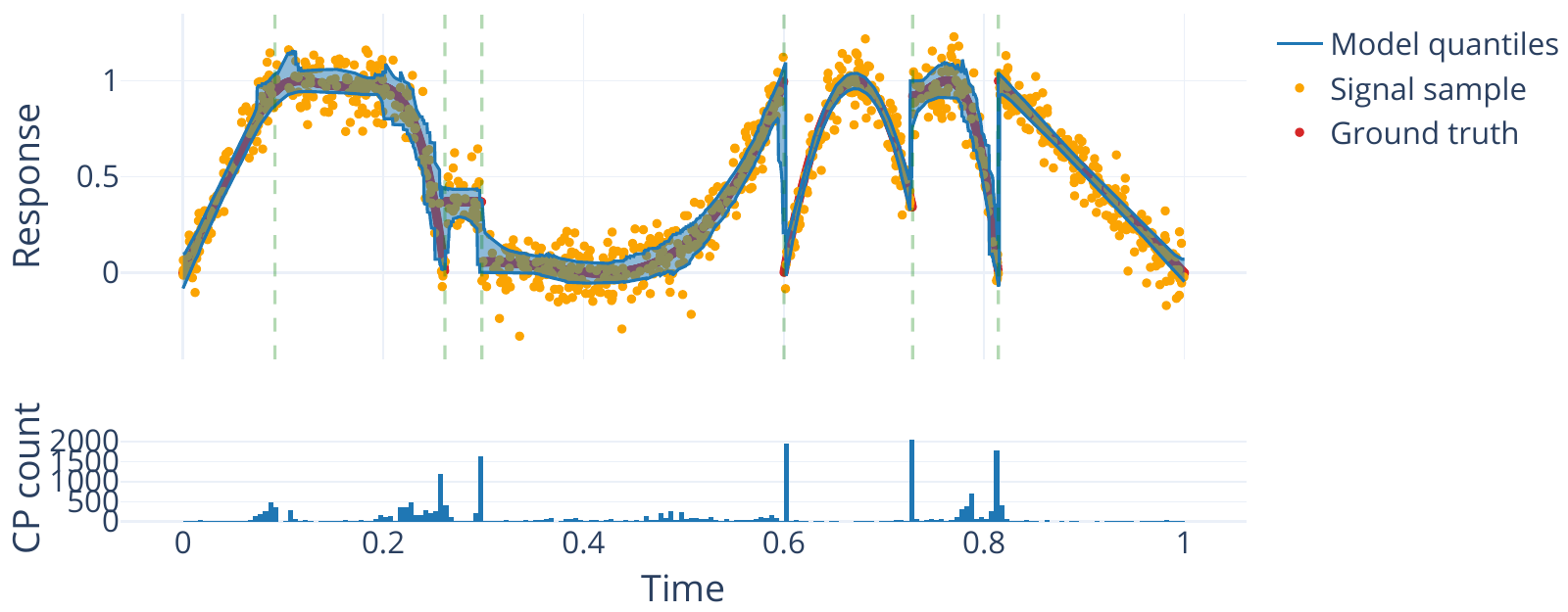}{$\sigma=0.1$}
    \myincludegraphics{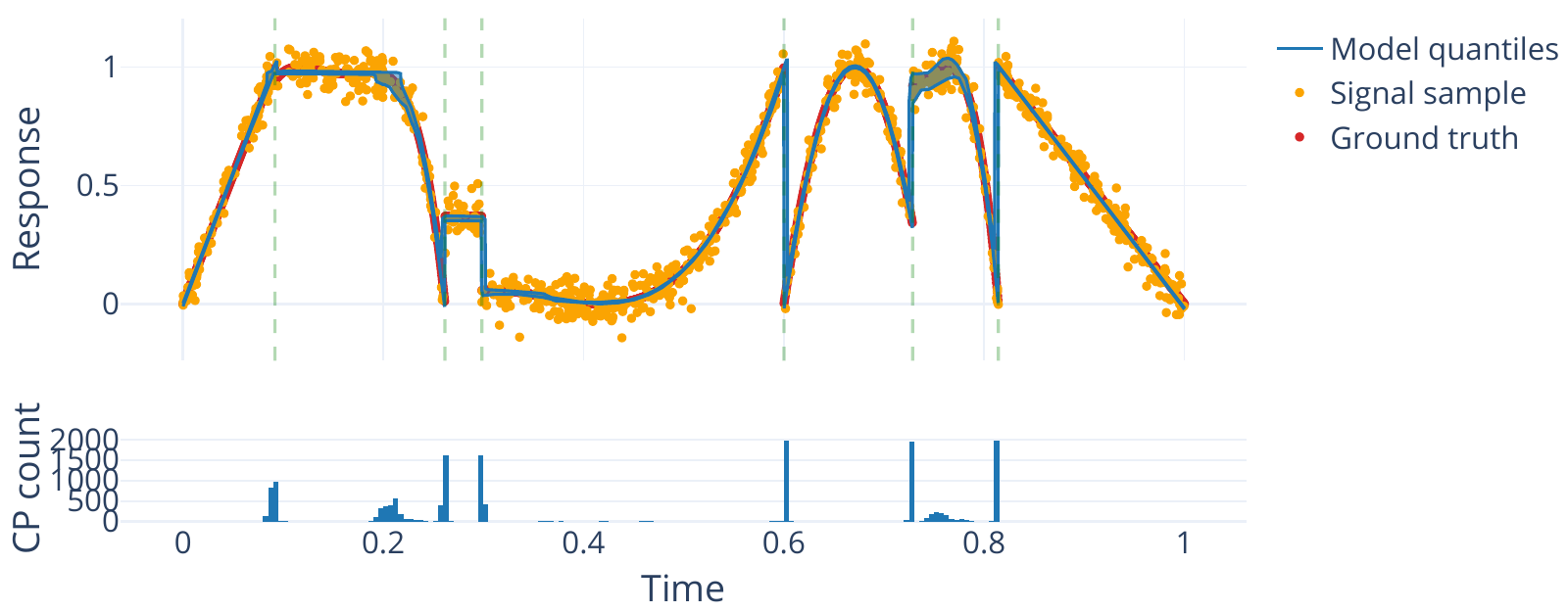}{$N=1000$}\hfill
    \myincludegraphics{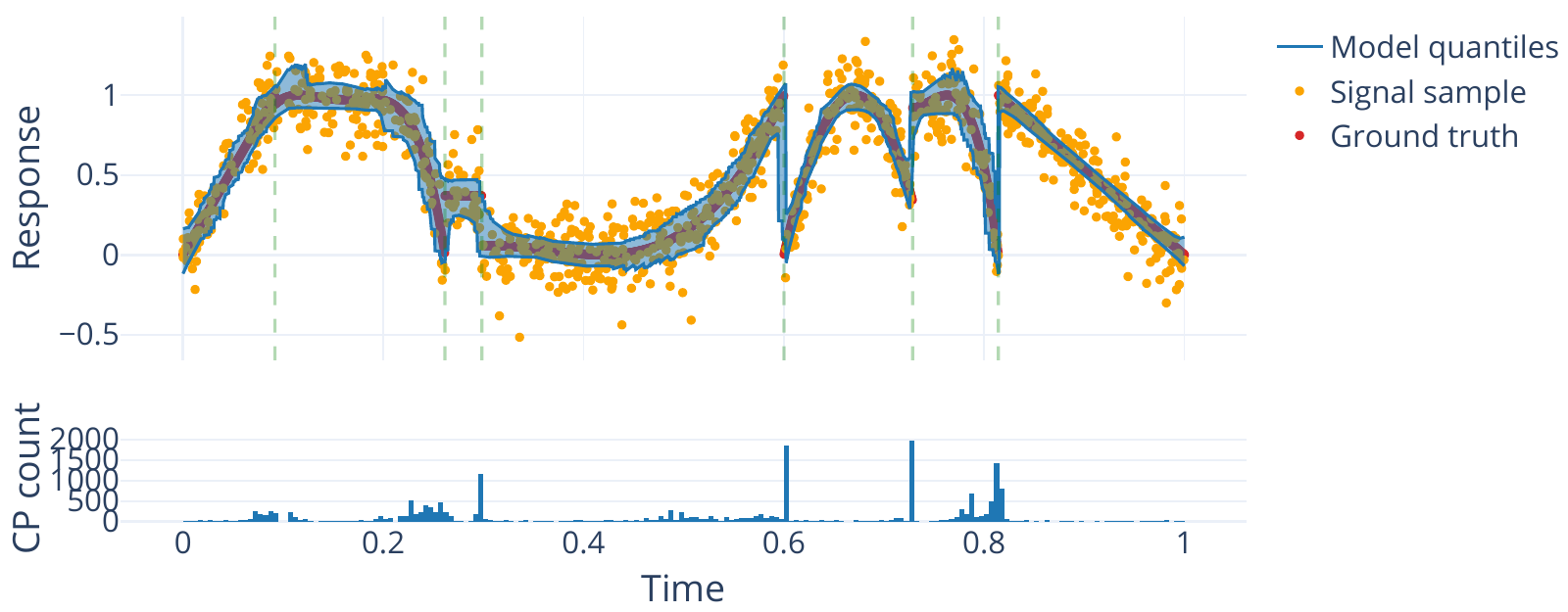}{$\sigma=0.15$}

    \caption{
      \emph{Left column:} We vary the number of samples in our timeseries samples from $150$ over $500$ to $1000$. The noise is kept fixed at $\sigma=0.05$ throughout. As we can see the pointwise accuracy clearly increases quite quickly. The histograms show that the spikes at the true changepoints clearly grow with the number of samples. One point that remains problematic for a lot of noise realizations is the supposed changepoint at about $0.2$ which results from the quite high degree of the ground truth admitting a quite simple decomposition into two low degree segments.
      \emph{Right column:} We vary the magnitude of the noise $\epsilon$ in our timeseries samples from $\sigma=0.05$ (top) to $\sigma=0.15$ (bottom) in increments of $0.05$. The number of samples is kept fixed at $500$ throughout. As we can see the pointwise accuracy clearly suffers as noise increases but all models remain sensible. The spike close to the first changepoint is likely a result of the models placing the first CP slightly too far to the right. We see that in particular the changepoints where the ground truth is close to being continuous aren't recognized as well at higher noise levels while the others are still being recognized quite well.
      \emph{Legend:} The orange dots depict a sample signal, the red line is the ground truth,
      the blue shadings the pointwise quantiles over 2000 realizations.
    }
    \label{main:fig:synth-poly-len-noise}
  \end{figure}

  The first timeseries we consider is a synthetic piecewise polynomial signal with Gaussian noise. We sample from $p(T) + \epsilon$ where
  $p$ is a synthetically generated piecewise polynomial function on $[0,1]$ with local degree no more than 10 with 6 jump points,
  $\epsilon$ is a normally distributed random variable with mean $0$ and standard deviation $\sigma$,
  and $T$ is a $[0,1]$-uniformly distributed random variable.
  (The full expression for the $p$ displayed can be found in supplementary section \ref{sup:sec:expression-piecewise-poly}.)
  The results are shown in Figure~\ref{main:fig:synth-poly-len-noise}.
  Each plot is based on sampling the signal distribution 2000 times. The top part of each subplot shows the ground truth signal $p$ in red, an exemplary sample in orange and the pointwise 2.5\% to 97.5\% quantiles of the resulting models as shaded blue regions. The bottom part of each subplot shows a histogram (with 200 bins) of the changepoints across all 2000 realizations.

  Figure~\ref{main:fig:quality_control_1} illustrates the application of our approach to a simulated signal from the TCPD dataset. This signal was  generated such that it contains a single changepoint at index 146. As described in the accompanying documentation, the noise characteristics differ before and after this changepoint, with Gaussian noise preceding it and uniform noise following it.
  In the process, five human annotators independently identified changepoints at indices 143, 144 (three times), and 146, respectively.
  The outcome of the DofPPR analysis reveals three segments in total: two constant segments and one linear segment. The predicted changepoints are located at indices 97.5 and 143.
  \begin{figure}
    \centering
    \includegraphics[width=.7\textwidth]{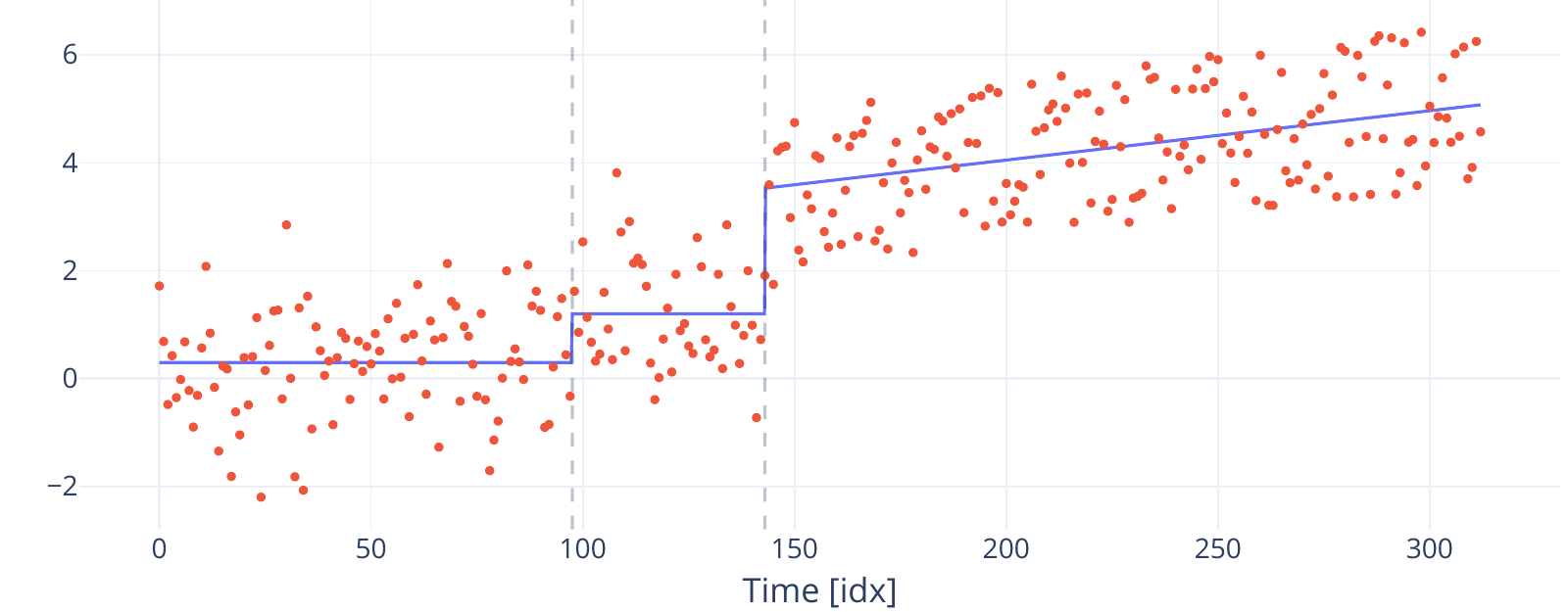}
    \caption{The  signal \texttt{quality\_control\_1} from the TCPD dataset was generated such that is has a single changepoint at index 146. According to the accompanying documentation the noise prior to this changepoint is Gaussian and after the changepoint it is uniform.
      The five human annotators selected the indices 143, 144 ($3 \times$), and 146 as changepoints, respectively.
    The DofPPR result (blue) consists of three segments: two constant and a linear one. The predicted changepoints are at 97.5 and 143.}
    \label{main:fig:quality_control_1}
  \end{figure}

  \begin{figure}[!t]
    \newcommand{\myincludegraphics}[2]{
      \begin{subfigure}[b]{0.45\textwidth}
        \centering
        \includegraphics[width=\textwidth, trim=0 0 150 0, clip]{#1}
        \caption{#2}
      \end{subfigure}
    }
    \myincludegraphics{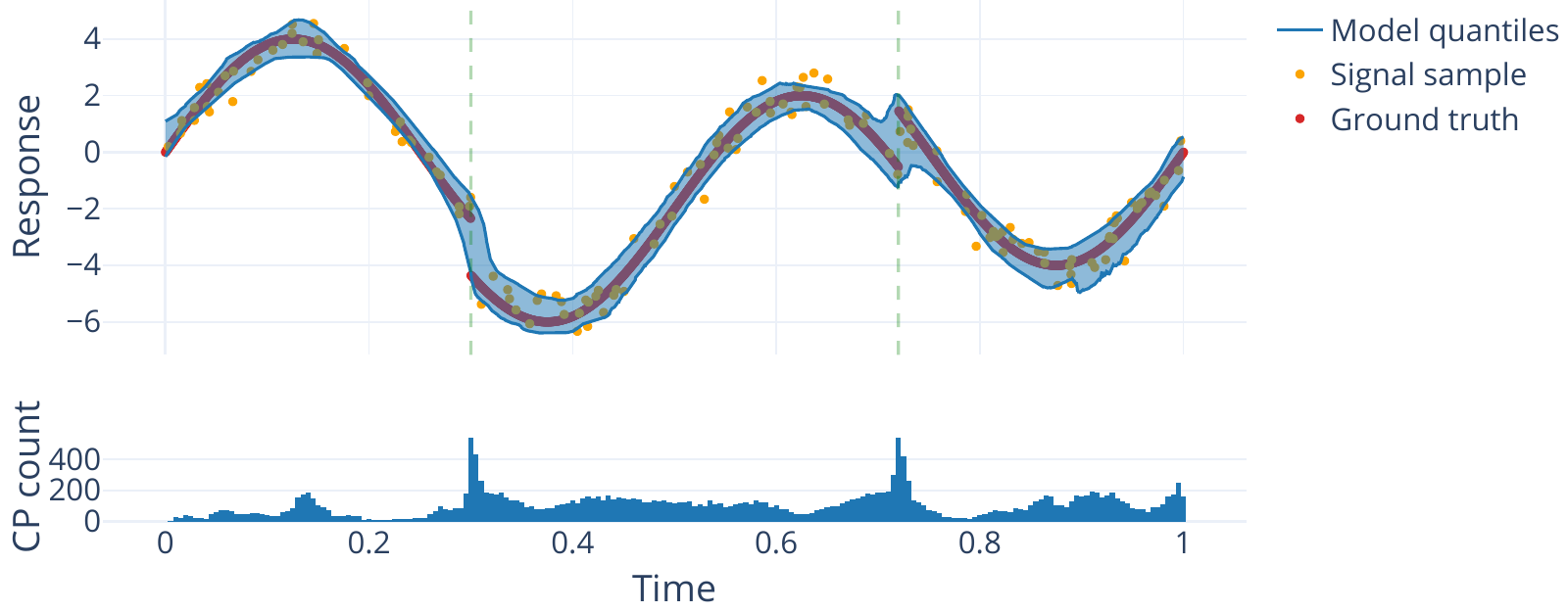}{$N = 150$}\hfill
    \myincludegraphics{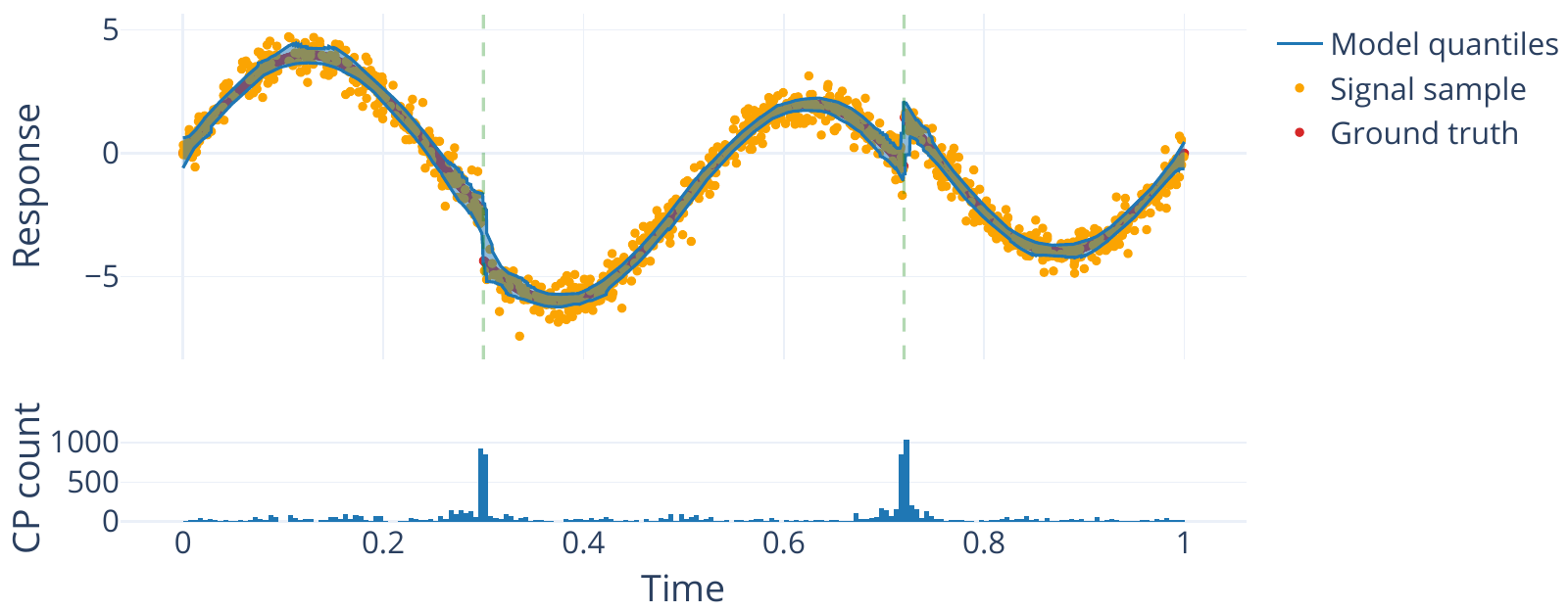}{$\sigma=0.5$}
    \myincludegraphics{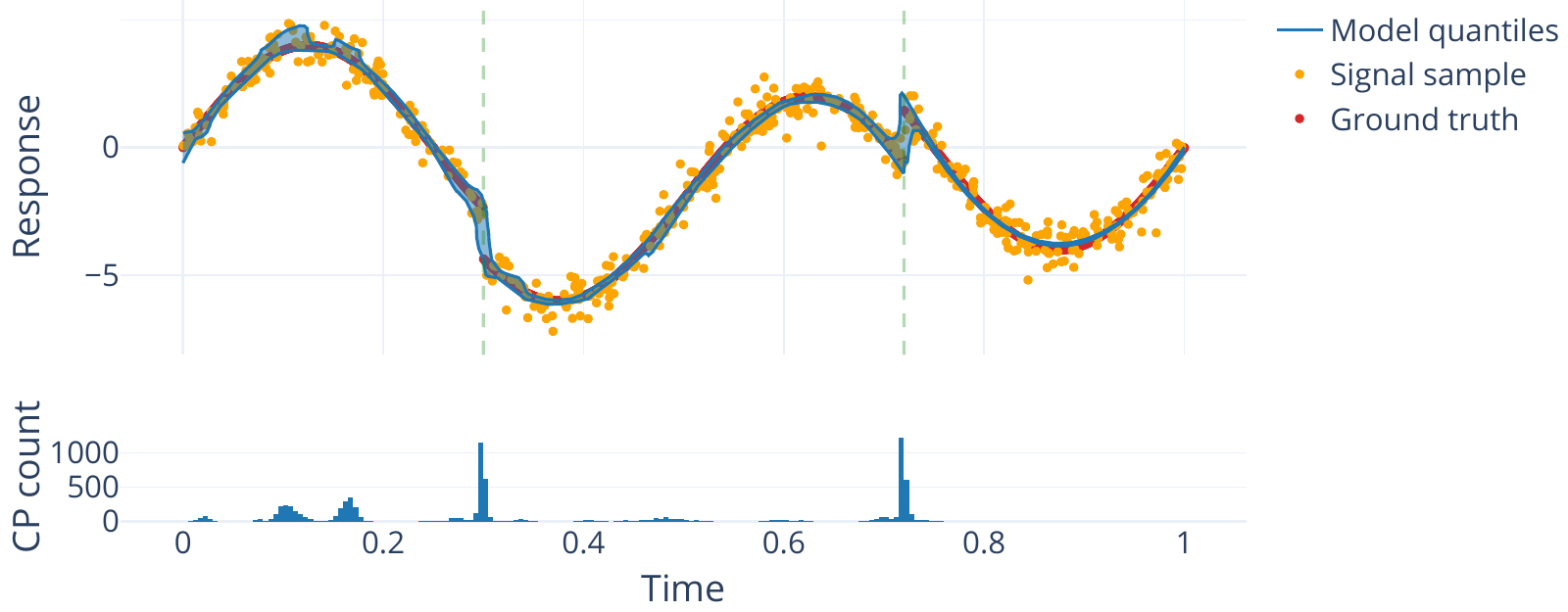}{$N = 500$}\hfill
    \myincludegraphics{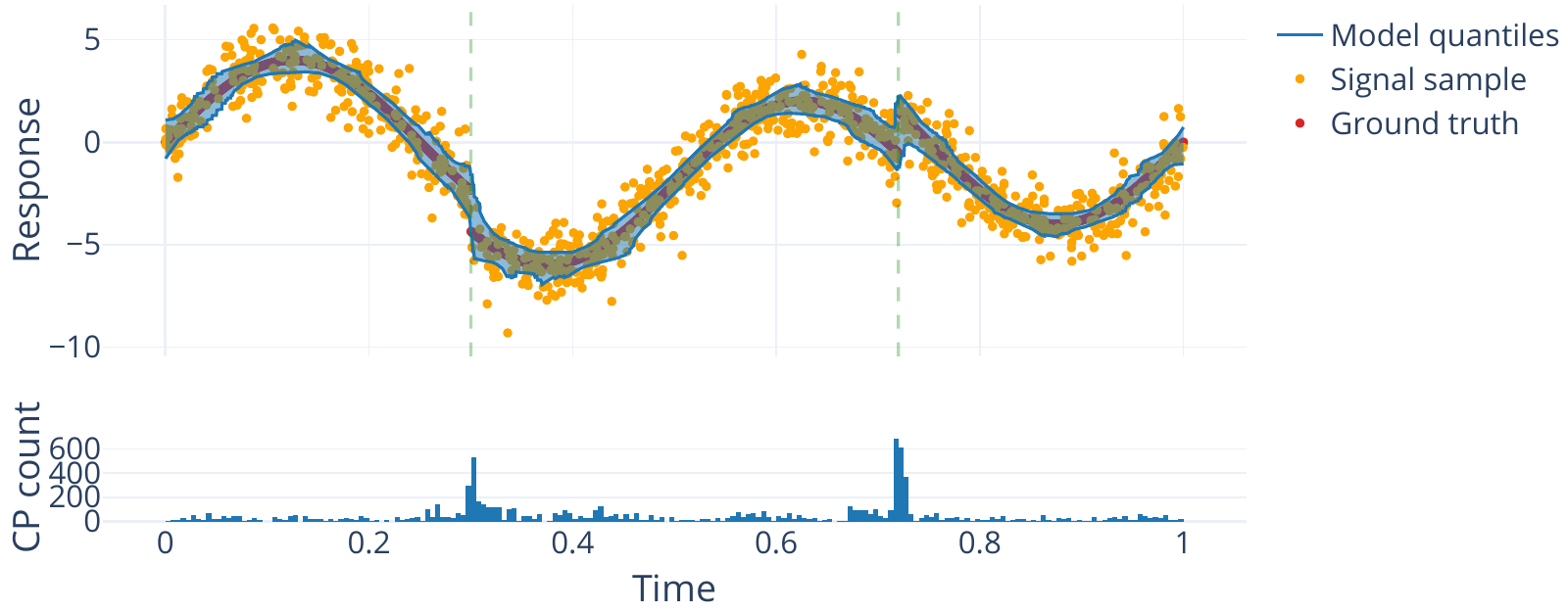}{$\sigma=1$}
    \myincludegraphics{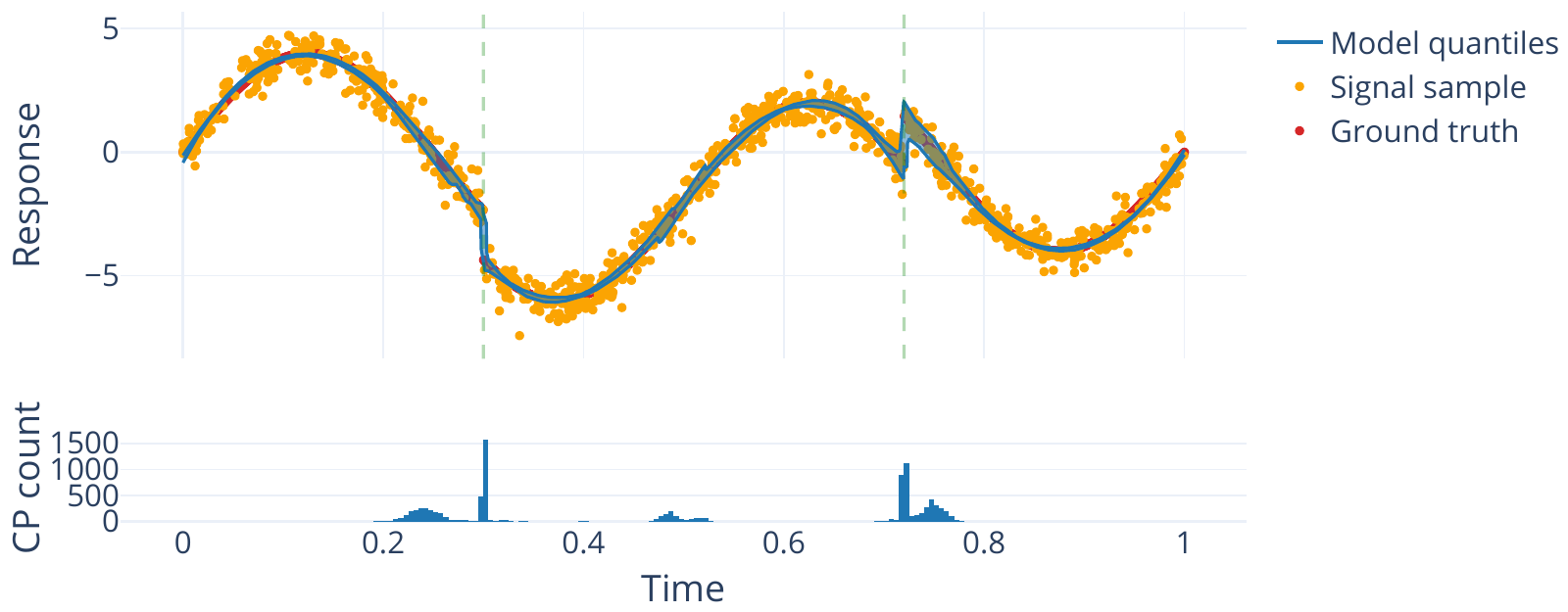}{$N=1000$}\hfill
    \myincludegraphics{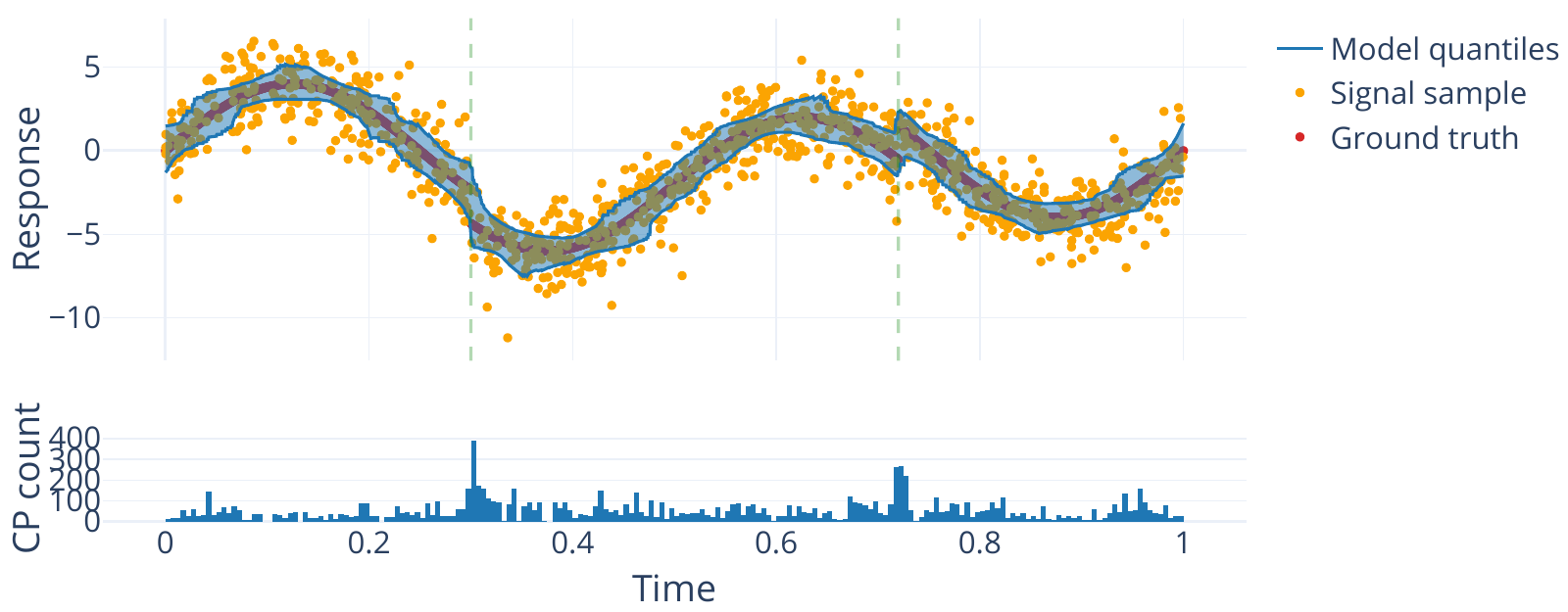}{$\sigma=1.5$}

    \caption{
      This is analogous to Figure~\ref{main:fig:synth-poly-len-noise} applied to the Heavisine signal. The notable difference is that the Heavisine function does not follow our basic model, \ie it is not piecewise polynomial and the complexity of the function does not change from segment to segment.
      \emph{Left column:} We vary the number of samples in our timeseries samples from $150$ over $500$ to $1000$. The noise is kept fixed at $\sigma=0.5$ throughout. As we can see the pointwise accuracy clearly increases quite quickly. The histograms show that the spikes at the true changepoints clearly grow with the number of samples. In the final subplot we can see see that there are a few models with changepoints around $0.23$ and $0.75$ that notably did not occur at lower numbers of samples.
      However when having a look at the evolution of the distributions of incorrect changepoints on intervals $[0, 0.3]$ and $[0.72,1]$ respectively,
      it appears to be the case that they are concentrating and moving towards their respectively closest true changepoints
      as the number of samples increase -- \ie their variance decreases while their mean approaches the true value.
      \emph{Right column:} We vary the magnitude of the noise $\epsilon$ in our timeseries samples from $\sigma=0.5$ (top) to $\sigma=1.5$ (bottom) in increments of $0.5$. The number of samples is kept fixed at $500$ throughout. As we can see the pointwise accuracy clearly suffers as noise increases but all models remain sensible and there is clear spikes around the true changepoints in all cases.
      \emph{Legend:} The orange dots depict a sample signal, the red line is the ground truth,
      the blue shadings the pointwise quantiles over 2000 realizations.
    }
    \label{main:fig:heavisine-len-noise}
  \end{figure}

  In Figure~\ref{main:fig:heavisine-len-noise} we consider the
  \emph{Heavisine} function $h \colon [0,1] \to \R, \quad h(x) = 4 \sin(4 \pi x) - \sgn(x - 0.3) - \sgn(0.72 - x)$
  of \cite{donoho1994ideal} where $\sgn$ is the sign function.
  It is a sine function with added jumps at $0.3$ and $0.72.$
  Similarly to the first example we sample from $h(T) + \epsilon$ with $T, \epsilon$ as above. To account for the Heavisine functions's
  larger range ($\min_x h(x) = -6, \max_x h(x) = 4$) the noise is scaled with a factor of $10$.
  The Heavisine function is not piecewise polynomial and the complexity of the function does not change from segment to segment,
  so the proposed model is not a perfect match for this signals.
  Although the results are not as good as for the piecewise polynomial ground truth signals,
  the model is still able to capture the jumps and the shape of the underlying signal on average.
  This indicates that the model exhibits a certain robustness to deviations from the signal class it was designed for.

  Figure~\ref{main:fig:runtime} shows the runtimes in dependence of the signal length $n.$
  When the number of dofs is not bounded, the runtime grows approximately as $\O(n^3)$ as predicted by Theorem~\ref{main:thm:complexity}.
  A signal of length 1000 can be processed in less than one minute on the above mentioned hardware.
  In practical applications, it can be reasonable to set a maximal total number of dofs to reduce the runtime.
  Then, the runtime grows only approximately quadratically in $n.$
  If we limit the maximal total number of dofs to $200,$ the runtime decreases to less than $10$ seconds for a signal of length 1000.
  The considered example exhibits an empirical runtime that scales approximately as $\O(n^{1.85})$ which is slightly better than the expected bound of $\O(n^2).$
  \begin{figure}[t]
    \newcommand{\myincludegraphics}[2]{
      \begin{subfigure}[b]{0.45\textwidth}
        \centering
        \includegraphics[width=\textwidth, clip]{#1}
        \caption{#2}
      \end{subfigure}
    }
    \myincludegraphics{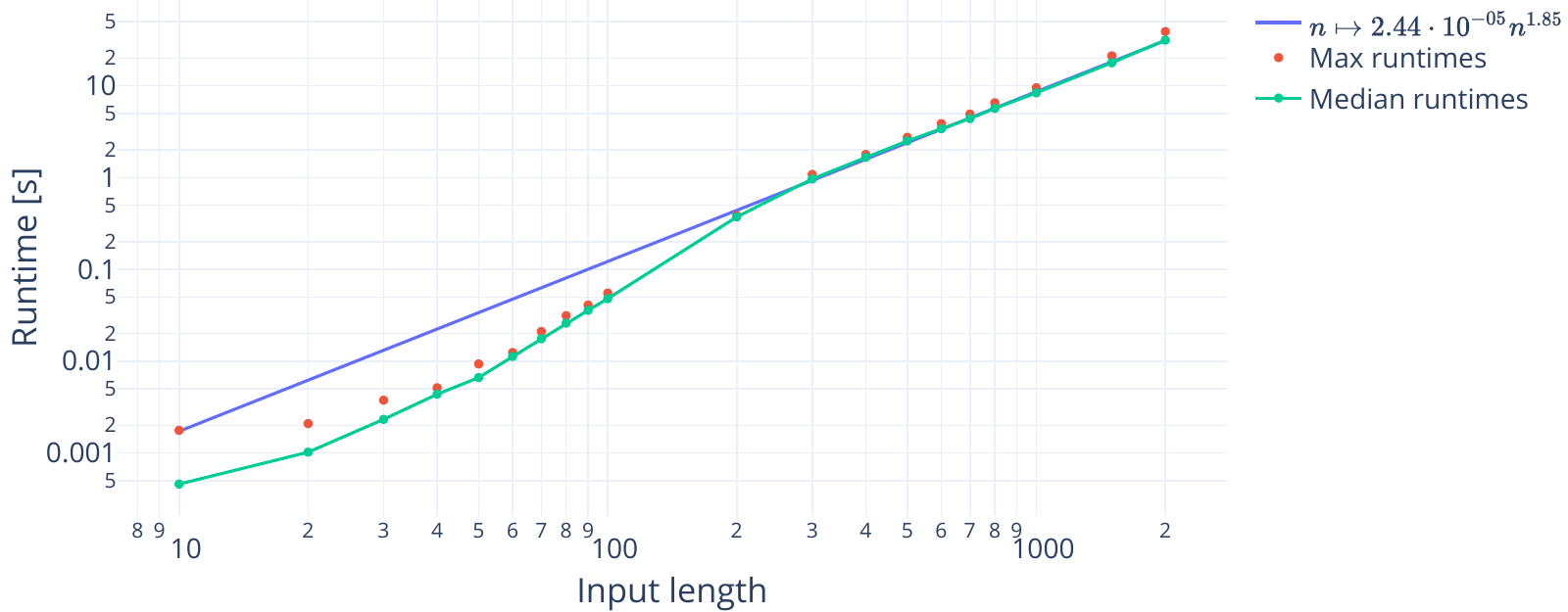}{Maximal total degrees of freedom set to $200$}
    \myincludegraphics{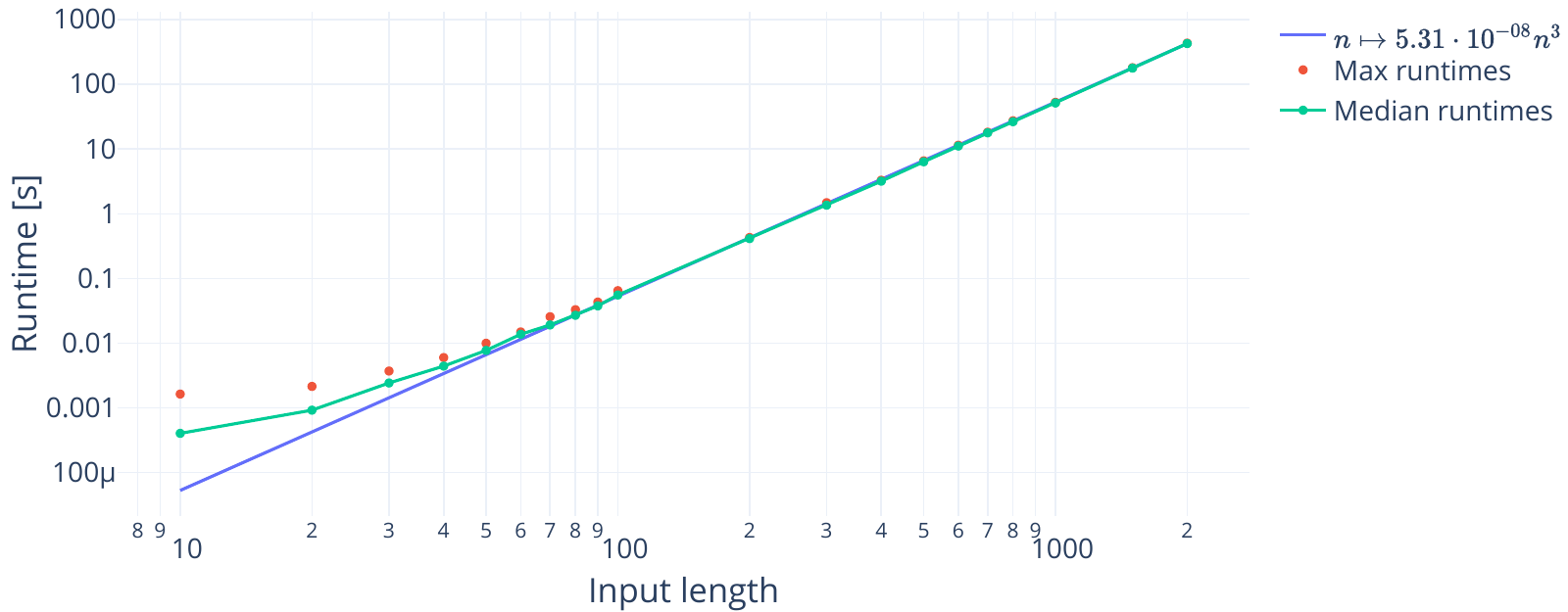}{No maximal total degrees of freedom}
    \caption{Runtime as a function of the length $n$ of the input timeseries. The algorithm was run on 5 different random piecewise polynomial functions sampled at varying numbers of points. The left figure depicts the algorithm with the default parameters of our implementation -- maximal local dofs set to $10$ and maximal total dofs set to $200$ -- while the right figure only uses the local limit of $10$ but no global limit. The plots show the median as well as maximum runtime across across 5 different input runs per number of samples. Finally the legends show the empirical sample complexity obtained via a least squares fit to the doubly-logarithmized sample medians.
      Our theoretical results predict a complexity of $\O(n^3)$ for the unbounded case and we see that this matches the experiment well. The bounded case appears to grow approximately as $\O(n^{1.85})$ which is slightly better than the expected bound of $\O(n^2)$.
    }
    \label{main:fig:runtime}
  \end{figure}

  \subsection{Comparison to partition penalized approaches}

  We compare the proposed method with the frequently used partition penalized approach \eqref{main:eq:partition_penalized_model}.
  As with the proposed method, we use least squares polynomial regression,
  but in the baseline the polynomials have fixed degrees ($0,1,2$ and $5$).
  A standard solver for the latter problem is the PELT algorithm of \cite{killick2012optimal}. The piecewise polynomial fit of fixed degree is implemented based on the Python package \emph{ruptures} from \cite{paper:ruptures}.This approach is referred to as the baseline in this subsection.

  A qualitative illustration using a relatively simple test signal and manually selected hyperparameters has been given in the introduction in Figure~\ref{main:fig:comparison_pcw_DofPPR}.

  We next provide a quantitative comparison using the
  following setup for the baseline method.
  To obtain a  hyperparameter selection strategy comparable to the method
  we apply a $k$-forward cross validation with $k=25$ based on the one standard error rule. So for each penalty, we fit models to data segments $1:25, 1:50, 1:75, ...$ and cross-validate using the one standard error rule based on the resulting prediction errors.
  Additionally, we report the results of an oracle  strategy (\enquote{Or}) where we select the most parsimonious model (\resp largest penalty) that is closest to the true number of changepoints.
  Both of these were based on a model population of $100$ equidistantly spaced penalties from the interval $[0, 50]$.
  Finally we translated the changepoint indices returned by ruptures to the middle of the corresponding intervals to match the convention used in this paper.

  For the synthetic data, we generated 20 random piecewise polynomial signals with 5 changepoints and locally no more than 6 degrees of freedom via the \texttt{rnd-pcw-poly} Python library available on PyPI.
  We sampled these at 1000 points uniformly selected from $[0,1]$ and added a Gaussian noise term with mean 0 and standard deviation of 0.025.
  We report the mean of the residual $L^2$ error and the Hausdorff distance of the indicated changepoints as well as the median of the difference between the indicated and true numbers of changepoints across all these signals.
  The reason for choosing the median in the latter case is that both our model as well as ruptures show extreme outliers in this metric for some signals.

  \begin{table}[!t]

    \centering
    \caption{Comparison of the proposed method with the baseline approach using the partition penalized model solved by the PELT method implemented in the ruptures library. The dataset consists of 20 synthetic piecewise polynomial data with additive Gaussian noise.}
    \begin{tabular}{lcccccc}
      \toprule
      \multicolumn{1}{c}{\multirow{2.5}*{Model}}             & \multicolumn{2}{c}{Residual $L^2$ $\downarrow$} & \multicolumn{2}{c}{Hausdorff CP $\downarrow$} & \multirow{2.5}*{SNR $\uparrow$} & \multirow{2.5}*{Diff. in \# CP $\downarrow$}                  \\ \cmidrule(lr){2-3} \cmidrule(lr){4-5}
      \multicolumn{1}{c}{}                                   & abs.                                            & rel.                                          & abs.                            & rel.                                                          \\
      \midrule
      Baseline Deg=0, Cv                                     & 0.03                                            & 1.87                                          & 0.23                            & 22.18                                        & 21.26 & 194.00 \\
      Baseline Deg=1, Cv                                     & 0.03                                            & 1.87                                          & 0.25                            & 16.38                                        & 22.02 & 194.00 \\
      Baseline Deg=2, Cv                                     & 0.05                                            & 2.83                                          & $\infty$                        & $\infty$                                     & 14.49 & 2.00   \\
      Baseline Deg=5, Cv                                     & 0.09                                            & 5.79                                          & $\infty$                        & $\infty$                                     & 6.59  & 5.00   \\
      Baseline Deg=0, Oracle                                 & 0.06                                            & 3.87                                          & 0.15                            & 18.55                                        & 9.95  & 0.50   \\
      Baseline Deg=1, Oracle                                 & 0.04                                            & 2.64                                          & 0.11                            & 3.71                                         & 14.77 & 0.00   \\
      Baseline Deg=2, Oracle                                 & 0.04                                            & 2.24                                          & 0.10                            & 14.79                                        & 17.05 & 1.00   \\
      Baseline Deg=5, Oracle                                 & 0.03                                            & 1.92                                          & 0.19                            & 25.59                                        & 19.24 & 2.00   \\
      \midrule
      \multicolumn{1}{c}{Ours  with $\gamma = \gamma_{OSE}$} & 0.02                                            & 1.00                                          & 0.08                            & 1.00                                         & 37.25 & 4.50   \\
      \bottomrule
    \end{tabular}
    \begin{tablenotes}
    \item \textit{Table note:} Models named Or are selected using an oracle to select for closest match in the number of changepoints and ones named Cv via cross-validation. Column \emph{residual $L^2$} contains the mean residual $L^2$ distance to the ground truth and mean of the normalized residual $L^2$, column \emph{Hausdorff CP} the mean Hausdorff distance to the true changepoints (and normalized version), column \emph{SNR} the mean signal to noise ratio and \emph{\# CP} the median error in the number of changepoints. A Hausdorff distance of $\infty$ indicates that the model didn't contain any changepoints.
      The signal-to-noise ratio is calculated as $\mathrm{SNR}(\omega) = \Norm{g}_{2}/\Norm{\omega - g}_{2}$ where $g$ is the ground truth and $\omega$ the model.
      The relative errors are measured relative to our score: so a relative score of $x$ means that the error is $x$-times as large as ours, \ie the score is normalized to our score as a reference.
      A column annotation of $\downarrow$ means that lower is better while $\uparrow$ means that higher is better.
    \end{tablenotes}
    \label{tab:comparison_PELT}
  \end{table}

  The results are reported in Table~\ref{tab:comparison_PELT}. The proposed method gives a lower residual  $L^2$ error as well as a lower Hausdorff distance of the changepoints than the baseline approach. This improvement may be attributed to the model's more efficient utilization of the "penalty budget" due to the heterogeneous local degrees: introducing a changepoint into the heterogeneous model has a non-constant cost -- in contrast to the constant jump cost in homogenous models. This allows the model to locally increase the degrees of freedom if they significantly improve the data fit, without influencing the penalty for the remainder of the model. The change point oracle models show that using the baseline it is theoretically possible to get closer to the true number of changepoints on average -- however it is not clear what selection strategy has to be used on the penalty for this in practice. Furthermore the cross-validated models of degrees 0 and 1 show that it's also possible to get very bad models (with respect to this criterion) from the baseline when using standard model selection methods.

  In settings of exploratory analysis both the baseline as well as our algorithm allow users to intervene when they notice an excessive number of changepoints. In the case of baseline approach using the  PELT solver this is achieved by limiting the number of segments; for our algorithm it amounts to limiting the total number of degrees of freedom of the full model.
  Finally it's worth mentioning that selecting both the degree as well as penalty for the baseline models comes with a relatively large runtime cost that is avoided by our algorithm's automatic selection strategy.

  \subsection{Results on real data}\label{main:sec:construction_dataset}

  Next we study two real data examples from the TCPD dataset \cite{paper:an-evaluation-of-change-point-detection-algorithms}.

  Figure \ref{main:fig:global_co2}  presents the application of the proposed method to the \texttt{global\_co2} data.
  Two human annotators annotated changepoints after indices 46, 90, and 47, 91, respectively,
  and three human annotators saw no changepoint.
  The  proposed method estimates two breaks at the indices 68 and 90.

  \begin{figure}[!t]
    \centering
    \includegraphics[width=0.7\textwidth]{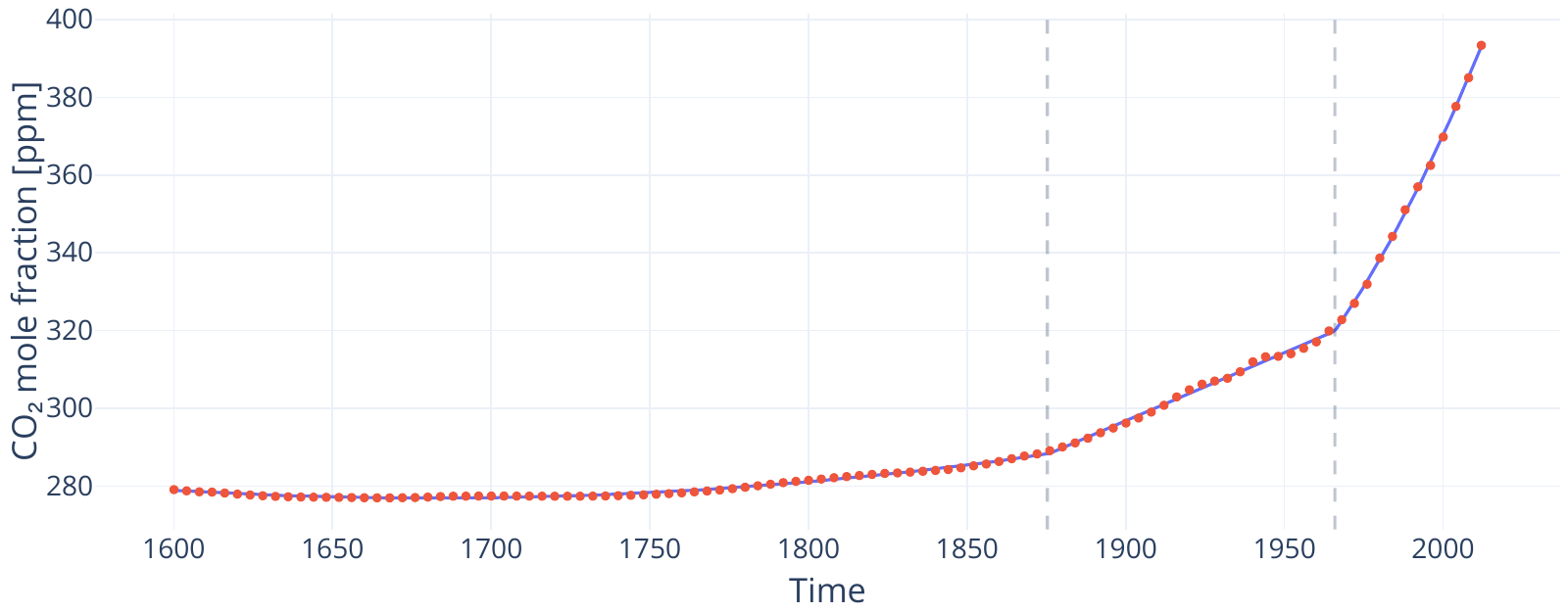} \hfill
    \caption{The TCPD dataset \texttt{global\_co2}. This dataset is based on the CMIP6 dataset of \cite{meinshausen2017historical} showing the monthly global hemispheric means of carbon dioxide in the atmosphere. The original timeseries contains over $24,000$ datapoints and has been cropped by \cite{paper:an-evaluation-of-change-point-detection-algorithms} for distribution as part of TCPD. Fitting the DofPPR model gives  a piecewise polynomial with three segments: a quadratic segment followed by a linear and another quadratic segment. The predicted changepoints are at 12.04.1875 and 19.11.1965 corresponding to indices 68 and 90. }
    \label{main:fig:global_co2}
  \end{figure}

  The next example is the \texttt{construction} dataset which contains the total
  private construction spending in the US over multiple years (Figure~\ref{main:fig:construction}). The accompanying documentation of this dataset suggests that potential change points occur at economical recessions.
  Here we have mapped the dates of the time interval into [0, 1] prior to analysis for simplicity. This dataset is interesting as it shows both seasonal waves as well as global trends.
  We extracted them using the proposed method in a two-stage exploratory data analysis.
  In a first step we fitted a model with the default settings to the data samples.
  We then manually restricted the total degrees of freedom to $\nu_{\mathrm{total}} \leq 81$
  which eliminated some spurious segments.
  Using the seasonal models we calculated the means for each segments by integrating the local polynomials across their corresponding intervals. In a second step, we applied the proposed method to those points.
  The breakpoints of the corresponding piecewise polynomial are (in \texttt{dd.mm.yyyy} format)
  07.07.1995,
  09.04.2001,
  03.03.2005,
  17.11.2009,
  12.01.2019.
  The breakpoints can be seen to line up with economically significant events to some extent. Interestingly the local extrema of the polynomial pieces also match up to such significant points.

  \begin{figure}[!t]
    \def\figwidth{0.7\textwidth}
    \centering
    \centering
    \includegraphics[width=\figwidth]{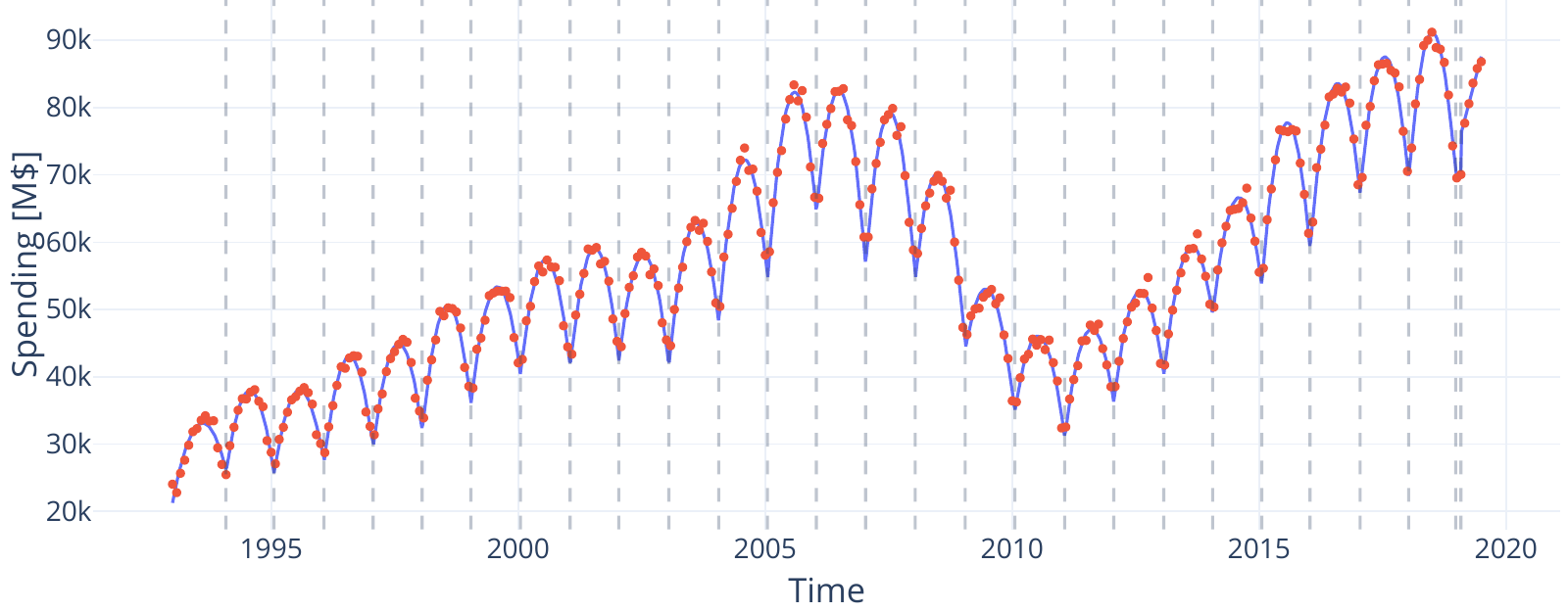}
    \centering
    \includegraphics[width=\figwidth]{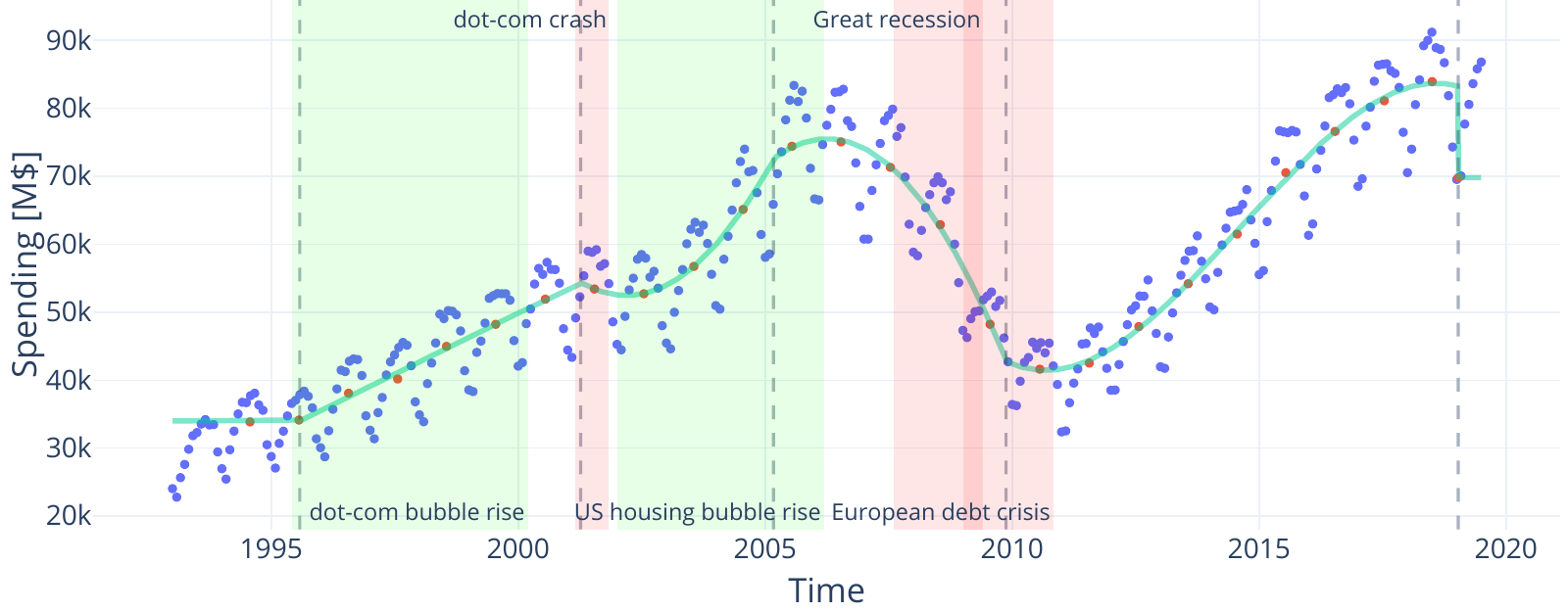}
    \caption{Exploratory data analysis on the construction data (top, red dots): In a first step, we applied the proposed model to data and manually restricted the total degrees of freedom to $\nu_{\mathrm{total}} \leq 81$ such that the resulting fit captures the seasonality. The dashed vertical lines show the changepoints indicated by this model. In a second step, we calculated the means for each segments by integrating the local polynomials across their corresponding domains (bottom, red dots), and applied the proposed method with default parameters to the resulting data points (green line). Economically significant developments are depichte as green and red shaded regions.}
    \label{main:fig:construction}
  \end{figure}

  \subsection{Results on the TCPD benchmark for  unsupervised changepoint detection}\label{main:sec:turing}

  We now evaluate the proposed method using the full TCPD benchmark \citep{paper:an-evaluation-of-change-point-detection-algorithms}. The benchmark consists of two separate evaluations:
  The Default setting aims to mimic the typical use-case where a data analyst, unfamiliar with the optimal parameter configurations, applies the algorithm to detect changepoints in a time series.
  The Oracle setting involves conducting a  grid search across the hyperparameters to identify the configuration that yields the highest performance for each algorithm.
  For the proposed DofPPR method,
  the Default setting refers to hyperparameter selection based on rolling cross validation with the OSE rule as described above.
  For the Oracle setting, the $\gamma$-parameter was varied
  with a grid of $101$ values evenly spaced on a log scale between $10^{-3}$ and $10^{3},$
  following the setup used for other penalized models in the benchmark.
  (We note that our Oracle setup  could be further refined
  by exploiting the access to the full regularization paths however this is not easily implementable due to how the benchmark is structured.)

  The results  are reported in the row \textsc{DofPPR-OSE} of Table~\ref{tab:turing}.
  We observe that the Oracle score is the highest among all competitors.
  By contrast, the Default score seems disappointing at first glance.
  To understand this  discrepancy, a deeper inspection of the human annotations,
  the parameter choice of the competitors and our parameter choice strategy is required.
  Human annotators were shown five control time series containing known changepoints to validate annotation quality.
  These control data had a maximum of two changepoints per series. Inspecting  some further annotations reveals that the human annotators preferred changepoint patterns on broader time scales to changepoint patters narrower-scale changes (e.g. the construction data set discussed in Section~\ref{main:sec:construction_dataset}).
  As outlined in Section~\ref{main:sec:construction_dataset}, when the hyperparameter
  $\gamma_{\text{OSE}}$ is used, it tends to adapt to narrower-scale  changepoint patterns such a seasonal patterns, often proposing a number
  of changepoints higher than a human annotator would.
  Another observation is that
  the top-performing method in the benchmark, \textsc{binseg}, uses a hard constraint on the number of change points,
  and so do most other of the best performing competitors (\textsc{segneigh}, \textsc{amoc}). Even the trivial method \textsc{zero}, which  always returns zero changepoints, is competitive.
  So assuming relatively few
  changepoints appears to be a strong prior information,
  which is used by the best competitors.
  To obtain a fair comparison with these methods, we propose a simple variant of our method enforcing a maximum of five changepoints as \textsc{binseg} does. This corresponds in the proposed model to a maximum of six degrees of freedom, $\nu_{\mathrm{total}} \leq 6.$
  (A piecewise constant model with five changepoints has six segments, corresponding to six degrees of freedom.)
  Using this constrained variant, referred to as \textsc{DofPPR-OSE-6} in the table,
  yields superior scores in both the cover metric and the F1 metric for the Default method.
  By the constraint, the oracle score is reduced but remains competitive.
  \begin{table}
    \centering
    \caption{Results on the TCPD benchmark \citep{paper:an-evaluation-of-change-point-detection-algorithms}}
    %% This table requires booktabs!
\begin{tabular}{lrrcrrcrrcrr}
\toprule
 & \multicolumn{5}{c}{Default} &  & \multicolumn{5}{c}{Oracle} \\
\cmidrule(lr){2-6} \cmidrule(lr){8-12}
 & \multicolumn{2}{c}{Univariate} &  & \multicolumn{2}{c}{Multivariate} &  & \multicolumn{2}{c}{Univariate} &  & \multicolumn{2}{c}{Multivariate} \\
\cmidrule(lr){2-3} \cmidrule(lr){5-6} \cmidrule(lr){8-9} \cmidrule(lr){11-12}
 & Cover & F1 & & Cover & F1 & & Cover & F1 & & Cover & F1\\
\midrule
\textsc{amoc}              & \textbf{0.668} & 0.653          &  &                &                &  & 0.717          & 0.773          &  &                &                \\
\textsc{binseg}            & \textbf{0.672} & \textbf{0.698} &  &                &                &  & \textbf{0.774} & \textbf{0.873} &  &                &                \\
\textsc{bocpd}             & \textbf{0.594} & \textbf{0.662} &  & \textbf{0.455} & \textbf{0.610} &  & \textbf{0.783} & \textbf{0.886} &  & \textbf{0.801} & \textbf{0.941} \\
\textsc{bocpdms}           & \textbf{0.590} & 0.495          &  & \textbf{0.512} & \textbf{0.426} &  & 0.753          & 0.659          &  & \textbf{0.689} & 0.654          \\
\textsc{cpnp}              & 0.488          & 0.586          &  &                &                &  & \textbf{0.759} & \textbf{0.845} &  &                &                \\
\textsc{DofPPR-OSE (ours)} & 0.275          & 0.385          &  &                &                &  & \textbf{0.792} & \textbf{0.905} &  &                &                \\
\textsc{DofPPR-OSE-6 (ours)}   & \textbf{0.676} & \textbf{0.753} &  &                &                &  & \textbf{0.783} & \textbf{0.866} &  &                &                \\
\textsc{ecp}               & 0.470          & 0.560          &  & \textbf{0.402} & \textbf{0.545} &  & 0.693          & 0.773          &  & \textbf{0.590} & \textbf{0.725} \\
\textsc{kcpa}              & 0.069          & 0.124          &  & 0.047          & 0.071          &  & 0.608          & 0.686          &  & \textbf{0.649} & \textbf{0.774} \\
\textsc{pelt}              & \textbf{0.652} & \textbf{0.674} &  &                &                &  & \textbf{0.772} & \textbf{0.864} &  &                &                \\
\textsc{prophet}           & 0.522          & 0.472          &  &                &                &  & 0.554          & 0.502          &  &                &                \\
\textsc{rbocpdms}          & 0.561          & 0.397          &  & \textbf{0.485} & \textbf{0.352} &  & 0.717          & 0.677          &  & \textbf{0.649} & 0.559          \\
\textsc{rfpop}             & 0.341          & 0.476          &  &                &                &  & \textbf{0.784} & \textbf{0.870} &  &                &                \\
\textsc{segneigh}          & \textbf{0.642} & \textbf{0.635} &  &                &                &  & \textbf{0.777} & \textbf{0.875} &  &                &                \\
\textsc{wbs}               & 0.264          & 0.365          &  &                &                &  & 0.366          & 0.482          &  &                &                \\
\textsc{zero}              & \textbf{0.566} & 0.645          &  & \textbf{0.464} & \textbf{0.577} &  & 0.566          & 0.645          &  & 0.464          & 0.577          \\
\bottomrule
\end{tabular}
    \begin{tablenotes}
    \item
    \end{tablenotes}
    \label{tab:turing}
  \end{table}

  \subsection{Further comparisons}

  For a further comparison with alternative piecewiese regression methods, we provide the results of
  two recent methods on the TCPD data shown in the paper; cf. Figures \ref{main:fig:quality_control_1}, \ref{main:fig:global_co2}, and \ref{main:fig:construction} above.
  The first method is cpop \citep{fearnhead2019detecting,fearnhead2024cpop}
  which estimates data by a continuous piecewise linear spline
  and is based on a change in slope model with a penalty on the number of slope changes.
  The second one is DeCAFS \citep{romano2022detecting}
  which combines elements of a Mumford-Shah/Blake-Zisserman type model with an autoregressive term.
  The results using the default parameters are reported in Figure~\ref{main:fig:comparison_other_methods}.
  We observe that the cpop method captures the shape of the first signal,
  but it requires two changes in slope for the jump point;
  in the other examples the number of changepoints appears to be overestimated.
  The DeCAFS model captures the jump of the first signal and the corresponding estimates exhibits some variation on the pieces.
  In the  construction data it adapts to the long term variation rather than the seasonal pattern.
  In the last signal, spurious discontinuities appear in regions with steep gradients. This phenomenon, observed in models that impose a penalty on the slope, was termed by \cite{blake1987visual} as the gradient limit effect.
  The results on the synthetic
  data  (Figures~\ref{main:fig:synth-poly-len-noise} and \ref{main:fig:heavisine-len-noise})
  are given in the supplementary material, noting that the DeCAFS method could not be applied to these data
  as it does not support non-equidistantly spaced data at the time of writing.

  \begin{figure}[!t]
    \centering
    \begin{subfigure}{\textwidth}
      \includegraphics[width=0.32\textwidth]{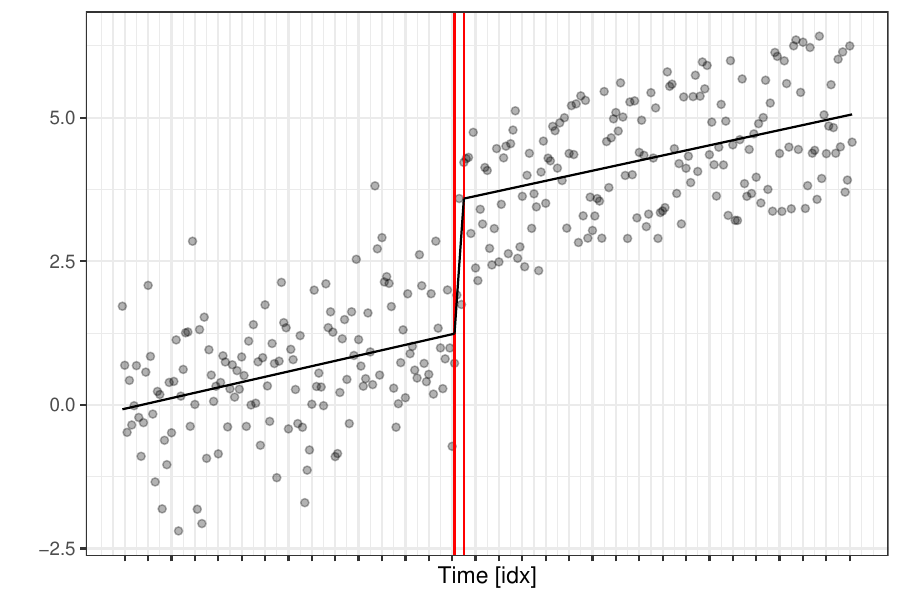} \hfill
      \includegraphics[width=0.32\textwidth]{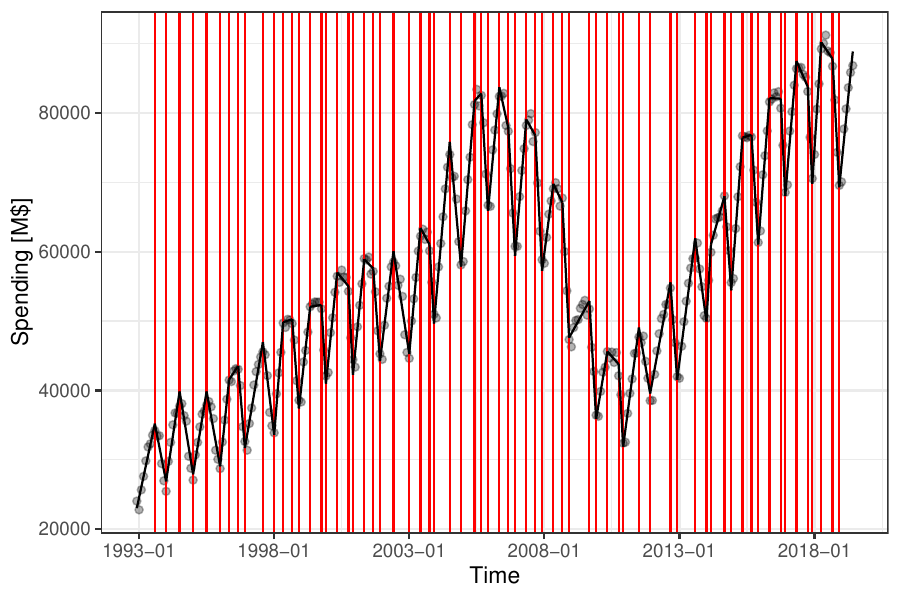} \hfill
      \includegraphics[width=0.32\textwidth]{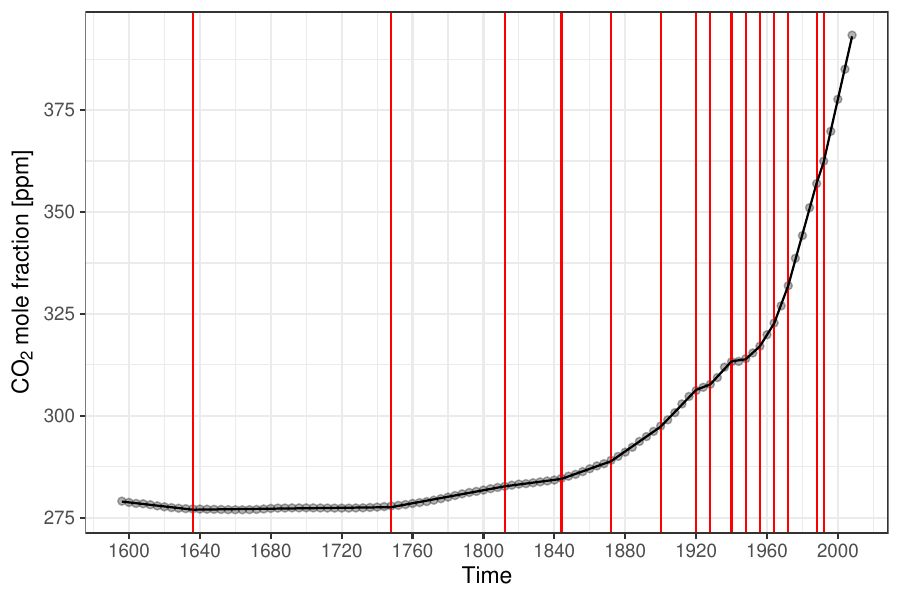}
      \caption{cpop \citep{fearnhead2019detecting,fearnhead2024cpop}}
    \end{subfigure}
    \begin{subfigure}{\textwidth}
      \includegraphics[width=0.32\textwidth]{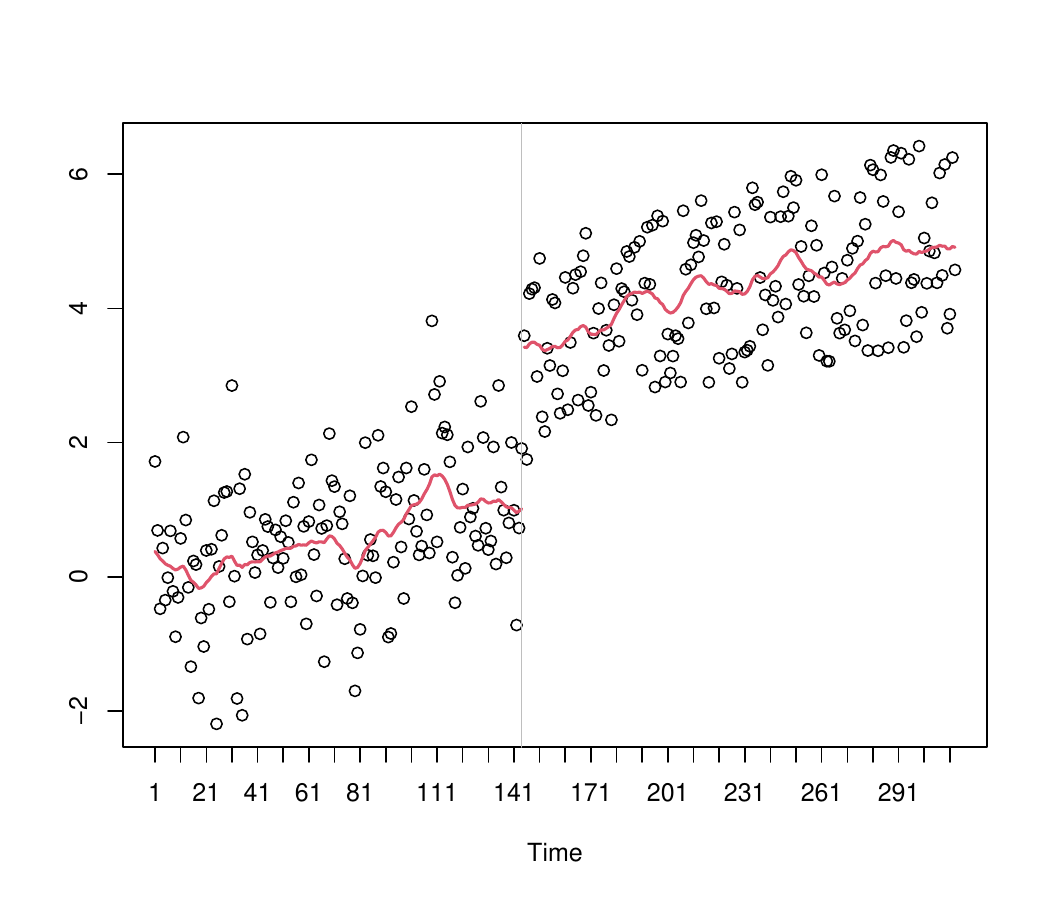} \hfill
      \includegraphics[width=0.32\textwidth]{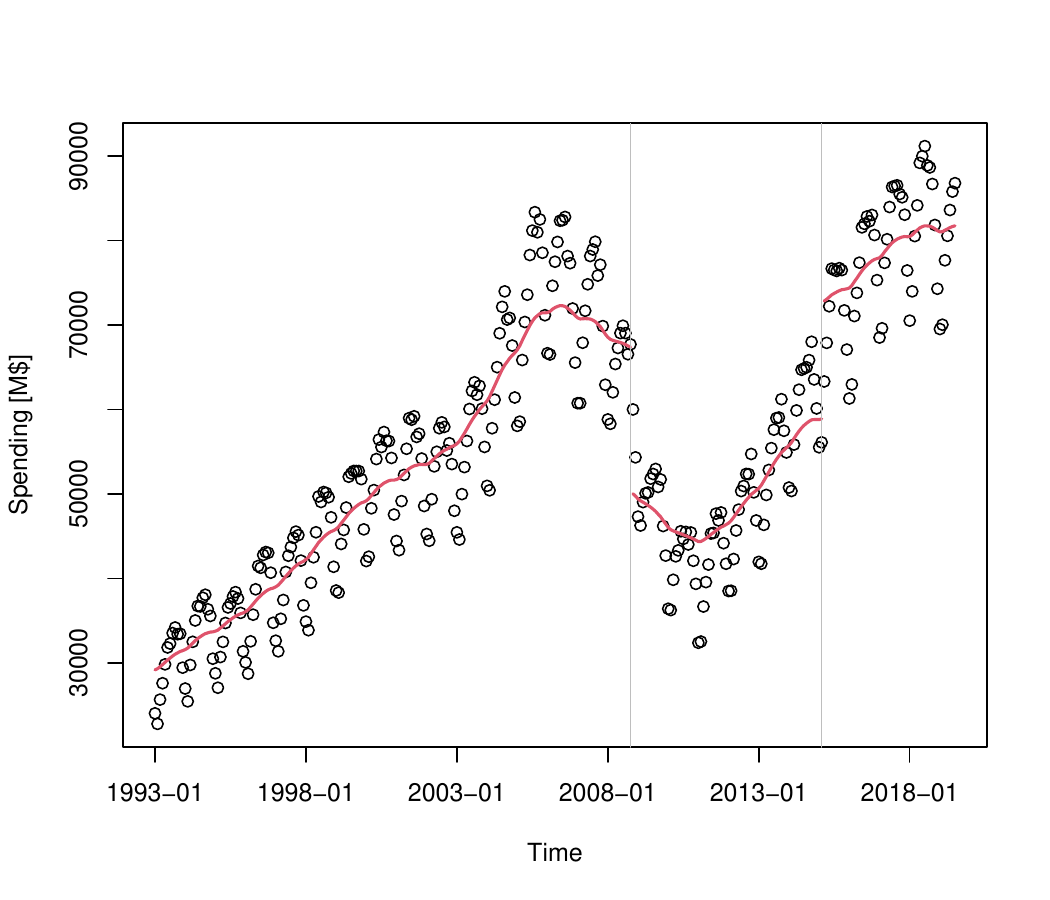} \hfill
      \includegraphics[width=0.32\textwidth]{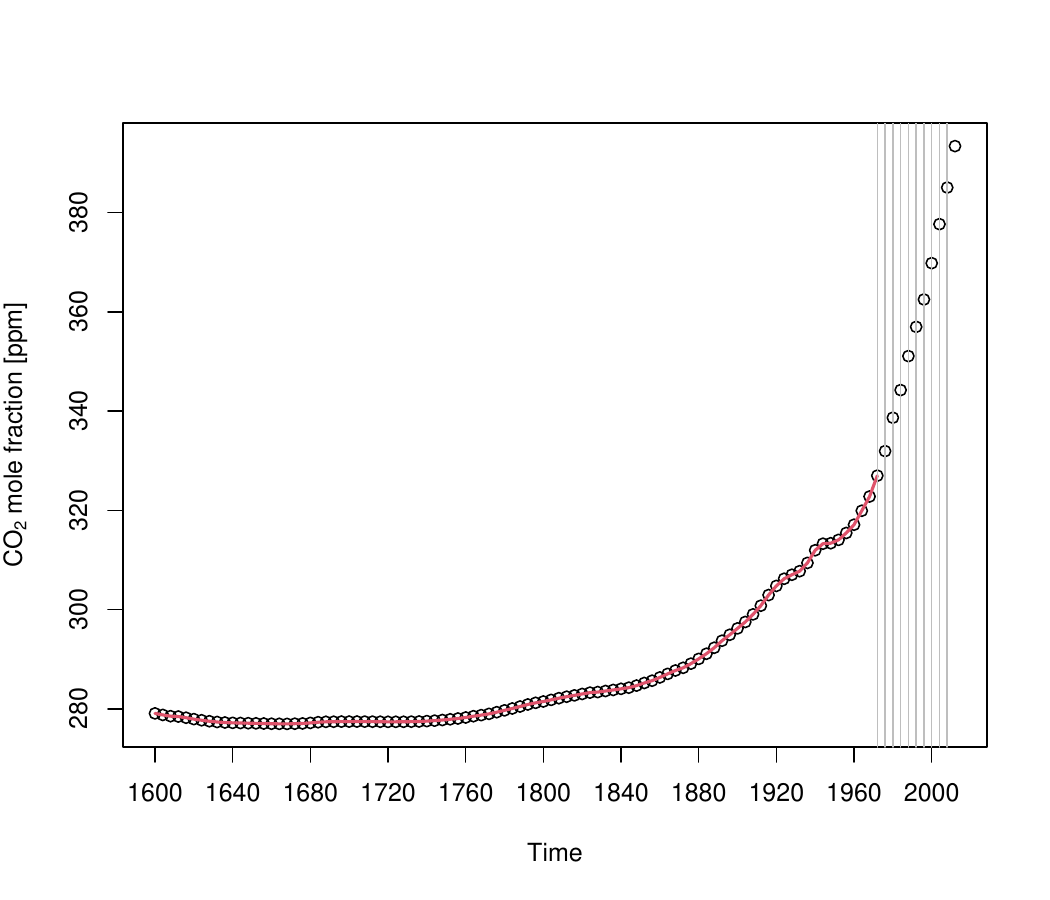}
      \caption{DeCAFS \citep{romano2022detecting}}
    \end{subfigure}
    \caption{Results on the TCPD time series \texttt{quality\_control\_1} (left), \texttt{construction} (middle), and \texttt{global\_co2} (right) using the cpop and DeCAFS methods for two recent piecewise regression models with default parameters of the corresponding R packages.
    }
    \label{main:fig:comparison_other_methods}
  \end{figure}

  \section{Discussion and Conclusion}\label{main:sec:conclusion}

  We have studied a piecewise regression model
  which is based on optimizing a tradeoff  between a goodness of fit term and a penalty on the total number of degrees of freedoms of the regression function, and the tradeoff is weighted by a model hyperparameter.
  On the analytical side, we have shown that the minimizer is unique (in an Lebesgue almost everywhere setting) for the important piecewise least squares setup
  when interpolatory functions are excluded.
  For other setups, we have proposed using distinguished solutions with
  parsimonious models and largest rightmost segments.

  We have developed a fast algorithm that computes the complete regularization paths.
  Its worst case time complexity is $\O(n^3 m^2)$
  whenever
  is the cost of fitting a model with $m$ degrees of freedom to data of length $n$ is in $\O(mn).$ For polynomials with least squares errors the time complexity is $\O(n^3 m)$.
  Furthermore, we have seen that performing model selection based on rolling cross validation (with or without one-standard-error rule) comes with little additional effort.
  For the least square instance with polynomials up to degree $10,$
  the total processing time (including model selection) is less than a minute
  for a signal of length $1,000$ on a standard desktop computer, and only around $10$ seconds
  when the total number of degrees of freedom is restricted to $200.$

  We have provided a full reference implementation of the piecewise polynomial case.
  The experimental results on synthetic data have illustrated the potential of the model for piecewise regression with mixed complexities.
  In particular we have seen that the model automatically gives satisfactory results for piecewise regression when the underlying signal is piecewise polynomial and the noise is i.i.d. Gaussian.

  We also demonstrated how the method suits exploratory data analysis.
  The method can be used to progressively reduce the degrees of freedom
  in a multiscale analysis fashion.
  By the automatic hyperparameter tuning,
  it is suitable for users which are not familiar with the proposed model
  and the hyperparameters. Yet, it is possible to request more parsimonious results by restricting the total number of degrees of freedom.

  Evaluation on the TCPD benchmark underpins the models strength in unsupervised changepoint detection outperforming the competitors in the oracle score.
  A variant that leverage an a priori restriction on the maximum number of changepoints -- as used by the best competing methods -- resulted in the highest score yet recorded for this benchmark.

  Though our results are promising, the model's hyperparameter selection could benefit from refinement. The current strategy based on rolling cross validation with one-standard-error rule tends to adapt to high-frequency changepoint patterns which may not always align with human interpretation of the correct scale. This discrepancy suggests the need to explore the human-perceivable scale for changepoints. Incorporating seasonal components into the model could be a potential solution. Additionally, the implementation of robust estimators, such as $\ell^1$
  data terms, could provide resilience against outliers.

  For solving related penalized models with a fixed $\gamma$-penalty,
  remarkable speedups in computatation time have been obtained using pruning strategies  \citep{killick2012optimal,fearnhead2024cpop,storath2014fast,runge2020change}.
  These could be applied as well for the proposed model when computing the solution for a fixed penalty.
  It is an interesting question of future research if  pruning can be extended for
  the proposed algorithm which computes the solution paths for the entire range of $\gamma$-penalties.

  A possible extension is an approach with two separate penalties for the number of segments and the degrees of freedom.

  \section*{Data Availability}
  Implementations of the algorithms developed in this paper are provided at \url{https://github.com/SV-97/pcw-regrs}.
  The simulated data were generated using the package \texttt{rnd-pcw-poly} available at \url{https://pypi.org/project/rnd-pcw-poly/} with the primary datum being given in supplementary Section \ref{sup:sec:expression-piecewise-poly}.
  The time series data are from the Turing Change Point Detection dataset  provided at \url{https://github.com/alan-turing-institute/TCPD} by the authors of the corresponding paper \citep{paper:an-evaluation-of-change-point-detection-algorithms},
  and the corresponding benchmark code is provided at \url{https://github.com/alan-turing-institute/TCPDBench}.

  \section*{Acknowledgement}
  Martin Storath  was supported by the project DIBCO funded by the research program \enquote{Informations- und Kommunikationstechnik} of the Bavarian State Ministry of Economic Affairs, Regional Development and Energy  (DIK-2105-0044 / DIK0264).
  Andreas Weinmann acknowledges support of Deutsche Forschungsgemeinschaft (DFG) under project number 514177753.

  %%%%%%%%%%%%%%%%%%%%%%%%%%%%%%%%%%%%%%%%%%%%%%%%%%%%%%%%%%%%%%%%%%%%%%%%%%%%%%%
  %%% end of main text
  %%%%%%%%%%%%%%%%%%%%%%%%%%%%%%%%%%%%%%%%%%%%%%%%%%%%%%%%%%%%%%%%%%%%%%%%%%%%%%%

  \printbibliography[heading=myheadingmain]
  \newrefsection

  %%%%%%%%%%%%%%%%%%%%%%%%%%%%%%%%%%%%%%%%%%%%%%%%%%%%%%%%%%%%%%%%%%%%%%%%%%%%%%%
  %%% start of supplementary
  %%%%%%%%%%%%%%%%%%%%%%%%%%%%%%%%%%%%%%%%%%%%%%%%%%%%%%%%%%%%%%%%%%%%%%%%%%%%%%%

  \newpage
  % \widetext

  \part{Supplementary Material}

  \renewcommand{\theequation}{S\arabic{equation}}
  \renewcommand{\thefigure}{S\arabic{figure}}
  \renewcommand{\thetable}{S\arabic{table}}
  \renewcommand{\thesection}{S\arabic{section}}

  % Note on the numbering: The  numbers below $100$ refer to equations, theorems etc. in the main document,
  % and numbers above $100$ to numbers in the supplementary file.

  \section{Proofs for Section~\ref{main:ssec:uniqueness}}

  \subsection{Proof of Lemma \ref{main:lem:projections}}

  \begin{proof}
    We start out to consider the corresponding representing matrices $\bar\pi_{P,\lambda},$ $\bar\pi_{Q,\mu}$ given by \eqref{main:eq:MatrixRep}.
    If $\bar\pi_{P,\lambda} \neq \bar\pi_{Q,\mu},$ then at least one of them does not equal the identity. W.l.o.g. assume $\pi_{P,\lambda}$ does not equal the identity.
    Then, by \eqref{main:eq:QuadWithProjection}, $G_{P, \lambda}(y)$ is a non-constant quadratic form w.r.t. $y.$
    In turn, $G_{P, \lambda}(y) - G_{Q, \mu}(y)$ is a non-constant quadratic form (with gradient $\pi_{P,\lambda} - \pi_{Q,\mu}$), and thus the set
    $\{ y \in \R^n : G_{P, \lambda}(y) = G_{Q, \mu}(y) \}$ has Lebesgue measure zero as a lower-dimensional manifold.
    If $\bar\pi_{P,\lambda} = \bar\pi_{Q,\mu},$ then, by \eqref{main:eq:QuadWithProjection},
    $G_{P, \lambda}(y) - G_{Q, \mu}(y) =  \gamma \sum_I  \left( \lambda_I -  \mu_I \right)$
    which is constant.
  \end{proof}

  \subsection{Proof of Theorem \ref{main:thm:uniquessFunction}}

  \begin{proof}
    For a given partitioning $P$ and given degreee vector $\lambda = (\lambda_I)_{I \in P},$ we consider those $(P,\lambda)$
    which yield the same projection/hat matrices $\bar\pi_{P,\lambda},$ i.e., we let
    \begin{equation}\label{sup:eq:DefMPnu}
      M_{(P,\lambda)} = \{ (Q,\mu):  \pi_{P,\lambda} = \pi_{Q,\mu}  \},
    \end{equation}
    and pick  $(Q,\mu) \in M_{(P,\lambda)}.$
    We observe that the corresponding  quadratic functionals $G_{P, \lambda}(y)$ and $G_{Q, \mu}(y)$ only differ by a constant by Lemma~\ref{main:lem:projections}.
    Hence, there is $(P^\ast,\lambda^\ast)$  such that
    \begin{equation}
      G_{P^\ast, \lambda^\ast}(y) = \min_{(Q,\mu) \in M_{(P,\lambda)}}  G_{Q, \mu}(y).
    \end{equation}
    From each $M_{(P,\lambda)}$ we pick such a   $G_{P^\ast, \lambda^\ast},$ and observe that there are at most countably many
    different $G_{P^\ast, \lambda^\ast}.$
    Further, the minimizer $\omega^\ast$ of Equation \eqref{main:eq:proposed_model_intro_poly}  equals
    $
    \min_{(P^\ast, \lambda^\ast)} G_{P^\ast, \lambda^\ast}(y)
    $
    for any data $y.$
    Taking two different sets  $M_{(P,\lambda)} \neq M_{(P',\lambda')},$
    their minimizing representatives $(P^\ast,\lambda^\ast)$ and $(P'^\ast, \lambda'^\ast)$ yield different projections
    $\bar\pi_{P^\ast, \lambda^\ast} \neq \bar\pi_{P'^\ast, \lambda'^\ast}.$ Hence, by Lemma~\ref{main:lem:projections}.,
    the set
    $\{ y \in \R^n : G_{P^\ast,\lambda^\ast}(y) = G_{P'^\ast, \lambda'^\ast}(y) \}$ has Lebesgue measure zero.
    Since there are at most countably many such zero sets where two functionals are equal, we may consider their union as an exclusion set $X$ of Lebesgue measure zero. On its complement $X^c$, the minimizing $P^\ast, \lambda^\ast$ is unique, its whole $ M_{(P,\lambda)}$ yields the same function and thus the minimizer is unique on the complement of the exclusion set.
  \end{proof}

  \subsection{Proof of Lemma~\ref{main:lem:PolynomialsFulfillAss}}
  \begin{proof}
    Part (i) of Assumption~\ref{main:assum:Basis} is a consequence of the invertibility of the corresponding Vandermonde matrices.
    Concerning part (ii) first note that the particular choice of polynomials is irrelevant, e.g. by plugging a change of basis matrices into the definition
    of $\pi_{I, \lambda_I}$ in \eqref{main:eq:ProjectionOnInterval}.
    Towards a contradiction, assume that $\pi_{I, \lambda_I}$ can be further decomposed into two or more block diagonal matrices.
    We denote these block matrices by $B_1,\ldots,B_k$ and their block sizes by $d_1,\ldots,d_k.$
    If the size of one of the block matrices exceeds the degree $\lambda_I$,
    say $d_r \geq \lambda_I,$ for $r \in \{1,\ldots,k\},$
    then we consider the corresponding subinterval $I_r \subset I$ together with data $y_r: I_r \to \mathbb R.$
    If $y_r$ equals $0$ (in all components), then the estimate $\hat y_r$ equals $0$ on $I_r$ by linearity.
    We now extent $y_r$ in two ways to functions $y_I$ and $y'_I$ on $I:$ we let $y_I$ be defined by $0$ on the complement of $I_r$ in $I,$
    and we let $y'_I$ be defined by $1$ on the complement of $I_r$ in $I.$
    By the fundamental theorem of algebra, both $\hat y_I$ and $\hat{y'}_I$ equal the zero vector on $I.$
    But since the polynomial hat operator reproduces constants,
    $\hat{y'}_I$ should equal $1$ on the complement of $I_r$ in $I.$
    This is a contradiction.
    We are hence left with the case
    that all block matrices  $B_1,\ldots,B_k$ have  block sizes $d_1,\ldots,d_k$ smaller than $\lambda_I.$
    Consider such a block $B_r,$  $r \in \{1,\ldots,k\},$ and arbitrary data $y_r: I_r \to \mathbb R.$
    Then $y_r$ is a sample of a polynomial $p$ of degree lower than $\lambda_I.$
    We consider this polynomial $p,$
    sample it on the whole interval $I$ to obtain data $y,$
    and apply the hat operator  $\pi_{I, \lambda_I}.$
    By polynomial reproduction, the hat operator reproduces the sample of $p,$
    and, in turn, $\hat y_r= y_r,$ it is the hat operator on $I_r$ equals the identity.
    Since $r$ was arbitrary, the hat operator on $I$ equals the identity.
    This contradicts the assumption $m<n$ representing the fact of more sample points than polynomial degree plus one.
    In consequence, $\pi_{I, \lambda_I}$ cannot be further decomposed which completes the proof.
  \end{proof}

  \subsection{Proof of Lemma~\ref{main:lem:standardDecompositionExUn}}

  \begin{proof}
    By construction, $\bar\pi_{P,\lambda}  = \bar\pi_{\tilde P,\tilde \lambda}$ which is formulated as statement (i).
    To see (ii) consider two standard partitions $\tilde P_1, \tilde P_2$ together with the degree vector $\tilde \lambda_1,\tilde \lambda_2$ corresponding to the partitioning $P$ with degree vector $\lambda = (\lambda_I)_{I \in P}.$
    Then, if $I$ in $\tilde P_1$ is of length $1$ it corresponds to an identity block in $P$ which results in a $1 \times 1$ block in
    $\bar\pi_{\tilde P_1,\tilde \lambda_1}$ which implies that  $I$ belongs to $\tilde P_2$ as well. Interchanging $\tilde P_1$ and $\tilde P_2$ shows the converse inclusion and in turn implies equality for identity blocks.
    Atomic blocks corresponding to non-interpolatory intervals of $P$ are not modified both in  $\tilde P_1$ and  $\tilde P_2,$ respectively.
    Together, $\tilde P_1 = \tilde P_2$ and $\tilde \lambda_1 = \tilde \lambda_2$ which implies the uniqueness of the standard decomposition.
  \end{proof}

  \subsection{Proof of Theorem~\ref{main:thm:UniquenessAEPart}}

  \begin{proof}
    The statement on the minimizing function was formulated as Theorem~\ref{main:thm:uniquessFunction}.
    So it remains to show the statement on the segments of a minimizing partitioning.
    For a given partitioning $P$ and given degree vector $\lambda = (\lambda_I)_{I \in P},$ we consider
    the set $M_{(P,\lambda)}$ defined by \eqref{sup:eq:DefMPnu}
    consisting of those partitions and degree vectors $(Q,\mu)$ which yield the same projection matrices
    $\bar\pi_{P,\lambda}.$
    We make the crucial observation that, in $M_{(P,\lambda)},$ all elements have the same standard block decomposition
    $\bar\pi_{P,\lambda}$ given by \eqref{main:eq:MatrixStandardDec} as a consequence of Lemma~\ref{main:lem:standardDecompositionExUn}.
    Hence, the non-interpolating intervals $I$ of all partitions are identical.
    Passing to the complement of the zero set $X$ identified in the proof of Theorem~\ref{main:thm:uniquessFunction},
    the segments $I^\ast$ of a minimizing partitioning corresponding to non-interpolating estimation are hence unique.
    This shows the second assertion of the theorem.
  \end{proof}

  \subsection{Proof of Corollary~\ref{main:cor:Uniqueness4Pol}}

  \begin{proof}
    By Lemma~\ref{main:lem:PolynomialsFulfillAss}  the system of polynomials fulfills Assumption~\ref{main:assum:Basis}.
    In consequence, the assertion of the corollary is a consequence of Theorem~\ref{main:thm:UniquenessAEPart}.
  \end{proof}

  \subsection{Proof of Corollary~\ref{main:cor:UniquenessNuMax}}

  \begin{proof}
    By the proof of Theorem~\ref{main:thm:UniquenessAEPart} we may minimize the problem in \eqref{main:eq:proposed_model_intro_poly}
    by minimizing w.r.t. all standard block decomposition instead of all partitionings.
    On a non-$1 \times 1$ block of the corresponding projection matrix of a block decomposition,
    the projection is non-interpolating. Thus the polynomial degree is strictly lower than the number of data points in the segment minus one
    which shows the assertion.
  \end{proof}

  \section{Proofs of Section \ref{main:sec:fast_algorithm}}

  \section{Efficient computation of residual errors in least squares polynomial fits}\label{sup:sec:efficient-poly}

  The general model allows us to fit optimal piecewise polynomial models to timeseries valued in $\R$. This necessitates the computation of the residual errors of all possible polynomial models for all segments of the timeseries:
  for all $S \in \Seg(1:n)$ and degrees of freedom $\nu \in 1:|S|$ we want to find the residual error corresponding to the polynomial $p$ of degree $\nu - 1$ minimizing
  \begin{align}
    \sum_{i \in S} (p(t_i) - y_i)^2
  \end{align}
  such that $\phi^\nu(t_S, y_S) = p$.

  We will now describe an algorithm that computes the residual errors of \emph{all} models with at most degree $d$ on data of length $n$ in just $\O(n^2 d^2)$; so it can in particular compute all possible models in $\O(n^4)$. The basic idea of the described algorithm is not original -- however the precise formulation resulting in the advantageous complexity may well be novel and appears to be currently unpublished. Note that a very similar algorithm to the one presented can be found in \cite{smoothing-for-signals-with-discontinuities}. This algorithm differs from ours in that it only deals with what we call the \emph{data recursion}, and only handles equidistantly spaced data.

  \subsection{Computing the residual error without the associated model}

  We will call a polynomial $p$ of degree $k$ a \emph{polynomial with $k+1$ degrees of freedom} and denote $k+1$ by $\dof(p)$.
  Let $\R[x]_{<\nu}$ be the linear space of all polynomials with no more than $\nu$ degrees of freedom (so degree less than $\nu$) with real coefficients and let $e_1, ..., e_\nu$ be a basis of $\R[x]_{<\nu}$ such that $\dof(e_j) = j$.
  Then there are $(\alpha_1, ..., \alpha_\nu) = \alpha \in \R^\nu$ such that the least squares polynomial $p$ we're after equals $\sum_{i=1}^\nu \alpha_i e_i$. We thus want to find $\alpha$ minimizing
  \begin{align}
    \sum_{j=1}^n \left( \sum_{i = 1}^\nu \alpha_i e_i(t_j) - y_j \right)^2,
  \end{align}
  which we easily show to be equivalent to
  \begin{align}
    \min_{\alpha \in \R^n} \Norm{A\alpha - y}^2_2
  \end{align}
  with $A =
  \begin{pmatrix}
    e_i(t_j)
  \end{pmatrix}_{\substack{j=1,...,n \\ i=1,...,\nu}} \in \R^{n, \nu}$ and $y = (y_1, ..., y_n)^T \in \R^n$.

  The central fact used for the efficient computation is the following lemma.

  \begin{lemma}
    There is an orthogonal matrix $Q \in \Ortho(n)$ such that the residual error of this minimization is given by $\Norm{(Q^T y)_{\nu+1:n}}_2^2$.
  \end{lemma}

  \begin{proof}
    Given a QR decomposition of $A$ into an orthogonal matrix $Q \in \Ortho(n)$ and a real matrix $\tilde{R} =
    \begin{pmatrix}
      R \\
      0_{n-\nu, \nu}
    \end{pmatrix}$ where $R \in \GL(\nu)$ we find that for all $\alpha$

    \begin{align}
      \Norm{A\alpha - y} = \Norm{Q^T(A\alpha - y)} = \Norm{\tilde{R}\alpha - Q^T y} = \Norm{
        \begin{pmatrix}
          R\alpha - (Q^T y)_{1:\nu} \\
          -(Q^T y)_{\nu+1:n}
      \end{pmatrix}}.
    \end{align}

    Squaring this expression we find
    \begin{align}
      \Norm{A\alpha - y}^2_2 = \Norm{R\alpha - (Q^T y)_{1:\nu}}_2^2 + \Norm{(Q^T y)_{\nu+1:n}}_2^2.
    \end{align}
    Since the rightmost term is independent of $\alpha$ it has to correspond to the residual error: minimizing our original target amounts to minimizing $\Norm{R\alpha - (Q^T y)_{1:\nu}}_2^2$, but since $R$ is invertible and the norm nonnegative this term trivially vanishes for the optimal solution $\alpha = R^{-1} (Q^T y)_{1:\nu}$. This shows that $\Norm{(Q^T y)_{\nu+1:n}}_2^2$ is indeed the residual error we're after.
  \end{proof}

  \subsubsection{Data recursion}

  We will now derive an recursion linking the residual error on data points $1,...,n$ with the one on data points $1,...,n+1$.

  Assume now that we already know the QR decomposition as stated above for data $1,...,n$. If we want to add another data point $t_{n+1}, y_{n+1}$ this amounts to adding another row to $A$ and $y$ to obtain $A' :=
  \begin{pmatrix}
    A \\
    e_{1:\nu}(t_{n+1})
  \end{pmatrix} \in \R^{n+1, \nu}$ and analogously a new right hand side $y' \in \R^{n+1}$. Setting $\tilde{Q} :=
  \begin{pmatrix}
    Q       & 0_{n,1} \\
    0_{1,n} & 1
  \end{pmatrix} \in \Ortho(n+1)$ we find that
  \begin{align}
    \tilde{Q}^TA' =
    \begin{pmatrix}
      R \\
      0 \\
      e_{1:\nu}(t_{n+1})
    \end{pmatrix}.
  \end{align}

  We can eliminate the last row using a sequence of Givens rotations $G_{1, n+1}$, $G_{2, n+1}, ..., G_{\nu, n+1} \in \Ortho(n+1)$ to obtain a QR decomposition for $A'$:
  \begin{align}
    A' = \underbrace{\prod_{i=1}^\nu G \tilde{Q}}_{\mathclap{=: Q' \in \Ortho(n+1)}}
    \begin{pmatrix}
      R' \\
      0
    \end{pmatrix}.
  \end{align}

  We find the corresponding residual to be $\Norm{(Q'y')_{\nu+1:n+1}}$. Since the Givens rotation $G_{j, n+1}$ only affects rows $j$ and $n+1$ of whatever matrix it acts on, we find the \emph{data recursion}
  \begin{align}
    \Norm{(Q'y')_{\nu+1:n+1}}_2^2 & = \Norm{(Qy)_{\nu+1:n}}_2^2 + \Norm{(Q'y')_{n+1}}_2^2 \\
    & = \Norm{(Qy)_{\nu+1:n}}_2^2 + ((Q'y')_{n+1})^2
  \end{align}
  relating the residual error of a $\nu$ degree of freedom model on data $1:n$ with the residual error of a $\nu$ degree of freedom model on data $1:n+1$. We want to emphasize the fact that $Q'$ is closely connected to $Q$ essentially via left-multiplication by a sequence of Givens rotations.

  \subsubsection{Degree recursion}

  We will now similarly derive a recursion linking the residuals of the $\nu$ and $\nu+1$ degree of freedom models on a given segment $1,...,n$.

  Assume now that we already know the QR decomposition as stated above for data $1,...,n$. We now extend our basis for $\R[x]_{<\nu}$ to a basis $e_1, ..., e_{\nu+1}$ of $\R[x]_{<\nu+1}$ and want to find the residual error of the least squares estimate from this larger space for the same data (of course under the assumption that $n \geq \nu+1$). Adding the new basis element amounts to adding another column to $A$ to obtain $A' :=
  \begin{pmatrix}
    A & e_{\nu+1}(t_{1:n})
  \end{pmatrix} \in \R^{n, \nu + 1}$.

  We factor $A'$ as
  \begin{align}
    A' =
    \begin{pmatrix}
      A & e_{\nu+1}(t_{1:n})
    \end{pmatrix} =
    \begin{pmatrix}
      Q
      \begin{pmatrix}
        R \\ 0
      \end{pmatrix} & e_{\nu+1}(t_{1:n})
    \end{pmatrix} = Q
    \begin{pmatrix}
      \begin{pmatrix}
        R \\ 0
      \end{pmatrix} & Q^T e_{\nu+1}(t_{1:n})
    \end{pmatrix}
  \end{align}
  and letting $G \in \Ortho(n)$ be the product of Givens rotations eliminating components $\nu+2$ through $n$ of $Q^T e_{\nu+1}(t_{1:n})$ from row $\nu+1$ we further factor
  \begin{align}
    A' = QG^T
    \begin{pmatrix}
      G
      \begin{pmatrix}
        R \\ 0
      \end{pmatrix} &
      \begin{pmatrix}
        w \\ 0
      \end{pmatrix}
    \end{pmatrix}
  \end{align}
  where $
  \begin{pmatrix}
    w \\ 0_{\nu+2, 1}
  \end{pmatrix} = GQ^T e_{\nu+1}(t_{1:n})$. Since $G$ is a product of Givens rotations that only operate on rows $\nu+1 : n$ we have $G
  \begin{pmatrix}
    R \\ 0_{n-\nu, 1}
  \end{pmatrix} =
  \begin{pmatrix}
    R \\
    0_{n-\nu,1}
  \end{pmatrix}$. Setting $R' =
  \begin{pmatrix}
    \tilde{R} & w
  \end{pmatrix}$, $Q' = QG^T$ we obtain the QR decomposition
  \begin{align}
    A' = Q'
    \begin{pmatrix}
      R' \\ 0_{n-(\nu+2)}
    \end{pmatrix}
  \end{align} and the residual error thus has to equal
  \begin{align}
    \Norm{(Q'^T y)_{\nu+2:n}}_2^2 = \Norm{(G Q^T y)_{\nu+2:n}}_2^2.
  \end{align} Once again using the property that $G$ really only ''modifies'' rows $\nu+1 : n$ we can easily obtain the new residuals from $(Q^T y)_{\nu + 1 : n}$ by applying the correct sequence of Givens rotations to it.

  \subsubsection{The full algorithm}

  We will now describe the full algorithm for obtaining the residuals for data $1:r$ for all $r \in 1:n$.

  To avoid conditioning problems the algorithm uses the Newton basis.
  Given a real sequence $x_1,...,x_d$ define the associated Newton basis for the space $\R[x]_{\leq d}$ of polynomials of degree no more than $d$ by
  \begin{align}
    N^0(x) & = 1,                                                    \\
    N^k(x) & = \prod_{i=1}^k (x - x_i) \qquad \text{for } k=1,...,d.
  \end{align}
  One central property of this basis is that $N^k(x_j) = 0$ for all $j \leq k$.

  The algorithm splits into two major parts: a rather involved core algorithm computing all residuals for segments starting at the first data point and a simple wrapper for finding the residuals on the the remaining segments by calling into the core for each possible starting point.

  The core algorithm iteratively constructs the matrix $\tilde{R}$ row by row and column by column by applying Givens rotations to the system matrix. This yields one new residual per eliminated matrix element. It's possible to do this without actually constructing the full system matrix and instead work with just a $(\nu + 1)\times (\nu + 1)$ matrix.

  We provide an extensively documented implementation of the algorithm in the form of a Rust crate at \url{https://crates.io/crates/polyfit-residuals}. The repository also contains a python implementation of the core algorithm, including a basic comparison with a numpy-based implementation of the naive algorithm.

  \begin{remark}
    The algorithm's outer loop is trivially parallelizeable. On a theoretical level (given an unbounded number of processes) this yields an algorithm of complexity $\O(nd^2)$; in practice the complexity will remain the same, however parallelization still allows for some great performance increases for all but the smallest data sizes and maximal degrees. Some benchmarks comparing the sequential and parallel implementations may be found on the crates' website mentioned above.
  \end{remark}

  \section{Linear algorithm for pointwise minimum of a set of affine functions}\label{sup:sec:pointwise}

  One commonly used algorithm for computing pointwise minima is a divide-and-conquer algorithm of quadratic complexity. A simpler linear time algorithm can be obtained by translating the computation of pointwise minima to that of finding convex-hulls on the dual of the homogenized functions. This correspondence is well-known. The Graham scan may then be used as a linear time algorithm.

  Alternatively the following algorithm may be used. It's robust to the potential numeric problems that can occur and integrates Occam's razor as desired by the main algorithm.
  (The algorithm is probably known, and we state it here for completeness as we did not find a suitable reference for it.)

  Where the usual recursive divide-and-conquer algorithm works by partitioning $\R$ directly and finding definitive minima at the points of intersection, the one described instead successively makes assumptions about the minimum and backtracks once it realizes those assumptions were incorrect. It does so by scanning through the input collection in order of the slopes calculating intersections of the unchecked functions with assumed minimal segments.
  A visualization of the main loop is shown in Figure \ref{fig:pointwise-min-iter-example}.

\begin{lstlisting}[language=Python]
from typing import List, Optional, Real, NamedTuple


class AffineFn(NamedTuple):
    slope: Real
    intercept: Real

    def __call__(self, x):
        return self.slope * x + self.intercept

    def __sub__(self, other):
        return AffineFn(self.slope - other.slope,
                        self.intercept - other.intercept)


def graph_intersection(fn1: AffineFn, fn2: AffineFn) -> Real:
    return (fn1.intercept - fn2.intercept) / (fn2.slope - fn1.slope)


class PcwAffFn(NamedTuple):
    borders: List[Real]
    fns: List[AffineFn]


def pointwise_min(affines: List[AffineFn]) -> PcwAffFn:
    """`affines` must be strictly sorted by slope in descending order"""
    match len(affines):
        case 0:
            return None  # an empty set has no minimum
        case 1:
            return PcwAffFn(borders=[], fns=affines)
        case n:
            minimals = [affines[0], affines[1]]
            jumps = [graph_intersection(affines[0], affines[1])]
            for right_fn in affines[2:]:
                left_fn = minimals[-1]
                while len(jumps) != 0 \
                        and (right_fn - left_fn)(jumps[-1]) <= 0.0:
                    minimals.pop()
                    jumps.pop()
                    left_fn = minimals[-1]
                new_jump = graph_intersection(left_fn, right_fn)
                # handles highly degenerate inputs
                while len(jumps) != 0 and new_jump <= jumps[-1]:
                    jumps.pop()
                    minimals.pop()
                jumps.append(new_jump)
                minimals.append(right_fn)
            return PcwAffFn(borders=jumps, fns=minimals)
\end{lstlisting}

  \begin{figure}
    \centering\includegraphics[width=\textwidth]{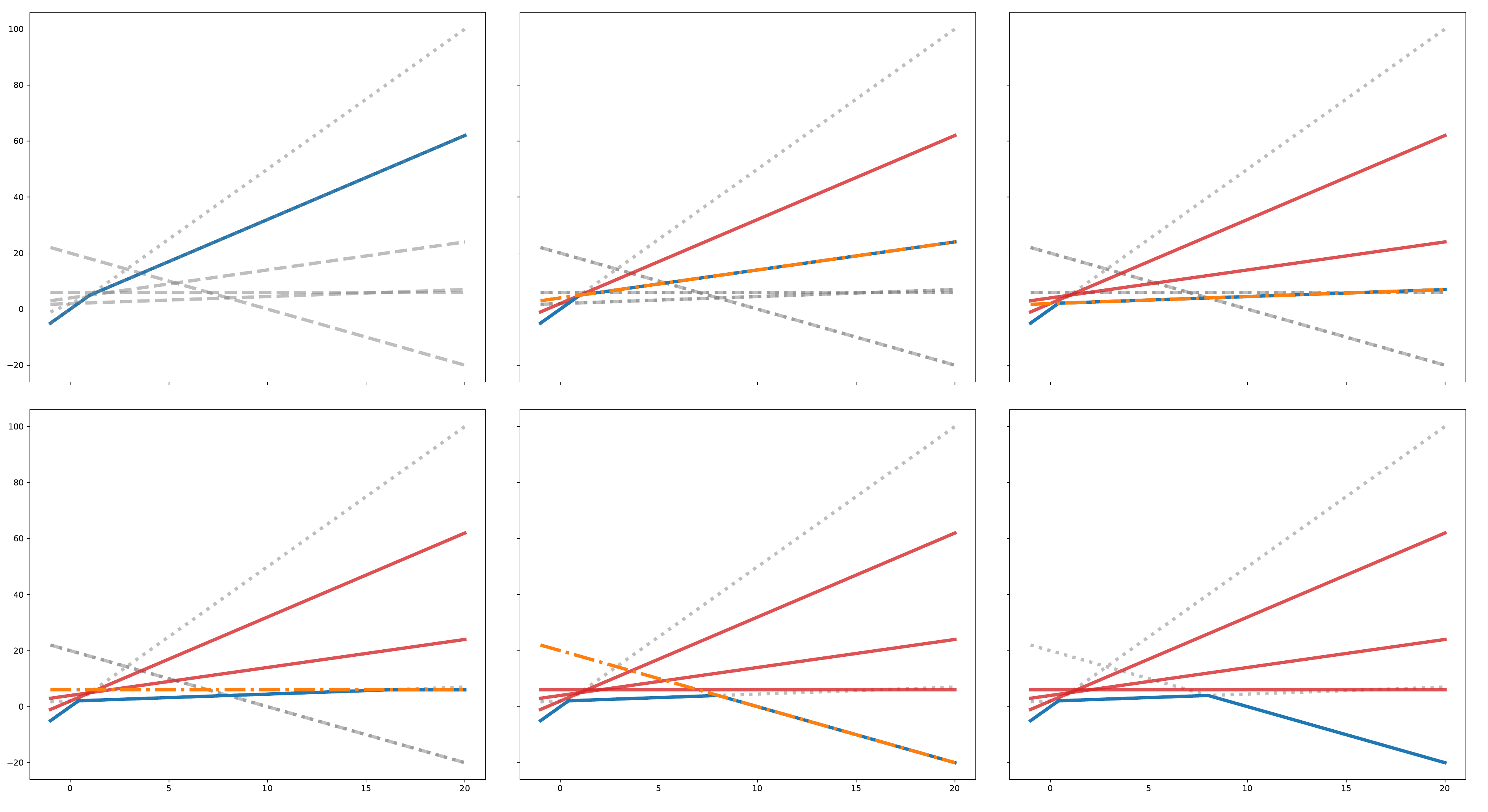}
    \caption{Visualization of the iterative algorithm. The top leftmost picture shows the initial state: we assume $f_1, f_2$ to make up the minimum with the border between the segments being at the point where the two graphs intersect. The other images in order left to right, top to bottom show successive iterations of the main loop finishing with the final state in the bottom right corner. The currently assumed pointwise minimum $F$ is drawn in blue, lines that we know to be nonminimal in red and lines we're yet to consider in dashed grey. During the main loop the currently considered \emph{right function} is drawn in orange.}
    \label{fig:pointwise-min-iter-example}
  \end{figure}

  \subsection{Proof of correctness}

  We will now prove that this algorithm actually works.

  \begin{proof}
    Denote by $f_1, ..., f_n$ the input functions, by $a_1, ..., a_n$ their slopes and by $F$ the pointwise minimum.
    We will use an induction argument on the size $n$ of the input assuming that $\mathcal{F} = \{f_1, ..., f_n\}$ such that $a_1 > ... > a_n$.

    For $n=2$ the main loop doesn't do anything and the returned initial state correctly represents the solution.

    We thus assume that the algorithm works for $n \geq 2$ and show that it also works for inputs of size $n+1 = |\mathcal{F}|$. By the induction assumption in the last iteration of the loop, the two (nonempty) stacks form the pointwise minimum $\widetilde{F}$ of $\{f_1, ..., f_n\}$. We thus only have to show that the last iteration will find $x \mapsto \min \{\widetilde{F}(x), f_{n+1}(x)\}$.
    It's not hard to show that for large enough $x$ the function $f_{n+1}$ of lowest slope will be minimal. Thus $f_{n+1}$ will be part of the correct solution $F$ and has to intersect $\widetilde{F}$ in some point $\xi$. Lets assume that $\widetilde{F}$ is given by $f_{i_1}$ on the open interval $I_1$, $f_{i_2}$ on $I_2$ and so on up to $f_{i_k}$ on $I_k$ \st $a_{i_1} > a_{i_2} > ... > a_{i_k}$.
    There are two cases to consider:
    \begin{enumerate}
      \item $\xi \in I_\nu$ for some $\nu$ ($f_{n+1}$ intersects $\widetilde{F}$ in a linear segment)
      \item $\xi \not\in I_\nu$ for all $\nu$ ($f_{n+1}$ intersects $\widetilde{F}$ in a corner)
    \end{enumerate}

    We start by analyzing the first case:
    Assume that $\nu = k$, $f_{n+1}(\xi) = f_{i_\nu}(\xi)$. Since $a_{i_\nu} > a_{n+1}$ and $\xi_k < \xi$ this implies that $f_{n+1}(\xi_k) > f_{i_\nu}(\xi_k)$ where $\xi_k$ denotes the top of stack border and the backtracking stops immediately without removing anything resulting in the correct minimum being returned.

    If on the other hand $\nu < k$ we know that $f_{n+1}(\xi) = f_{i_\nu}(\xi) < f_{i_k}(\xi)$ for $\xi < \xi_k$ and thus $f_{n+1}(\xi_k) < f_{i_k}(\xi_k)$ which means $f_{i_k}$ and $\xi_k$ are removed from their stacks in the first backtracking step. At this point the rest of the algorithm will behave exactly the same as if it was started on $\mathcal{F} \setminus \{f_{i_k}\}$ for which we know that the correct result will be returned by the induction hypothesis.

    A similar argument works for the second case completing the proof.
  \end{proof}

  \section{Expression for synthetic piecewise polynomial}\label{sup:sec:expression-piecewise-poly}

  The piecewise polynomial function used in the numerical experiments is given (to 3 decimal places) by
  \begin{align}
    \begin{cases}
      0.0 + 10.88 x                                                   & x \in [0.0, 0.092)   \\
      -3.688 + 128.812 x - 1299.917 x^2 + 5707.356 x^3 - 9229.693 x^4 & x \in [0.092, 0.262) \\
      0.37                                                            & x \in [0.262, 0.298) \\
      -3.881 + 33.877 x - 95.406 x^2 + 87.499 x^3                     & x \in [0.298, 0.6)   \\
      -88.748 + 267.536 x - 199.38 x^2                                & x \in [0.6, 0.729)   \\
      1523.272 - 6132.631 x + 8230.53 x^2 - 3679.993 x^3              & x \in [0.729, 0.814) \\
      5.383 - 5.383 x                                                 & x \in [0.814, 1.0]
    \end{cases}.
  \end{align}

  \section{Graphical scheme of algorithm}
  Figure \ref{fig:algo_graphic} depicts the basic architecture of the full algorithm.
  \begin{figure}
    \begin{center}
      \includegraphics[width=\textwidth, height=\textheight, keepaspectratio]{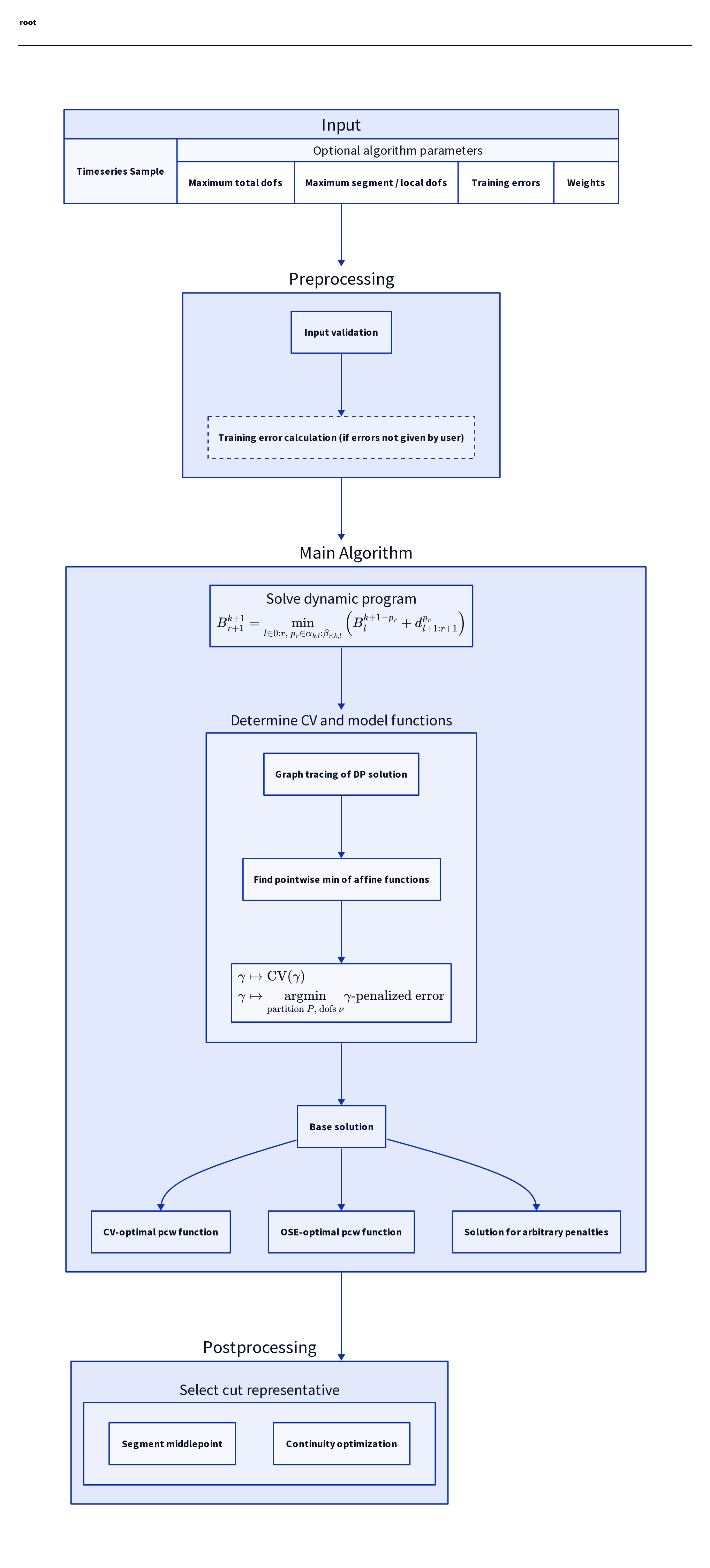}
    \end{center}
    \caption{Basic architecture of the full algorithm.}
    \label{fig:algo_graphic}
  \end{figure}

  \section{Comparison to $L^1$ based CV score}
  Figure~\ref{supp:fig:accuracy}
  shows the results of the experiment shown in Figure~\ref{main:fig:accuracy} of the main document
  for the $L^1$ based CV score.
  \begin{figure}
    \centering
    \begin{subfigure}{\textwidth}
      \includegraphics[width=0.49\textwidth]{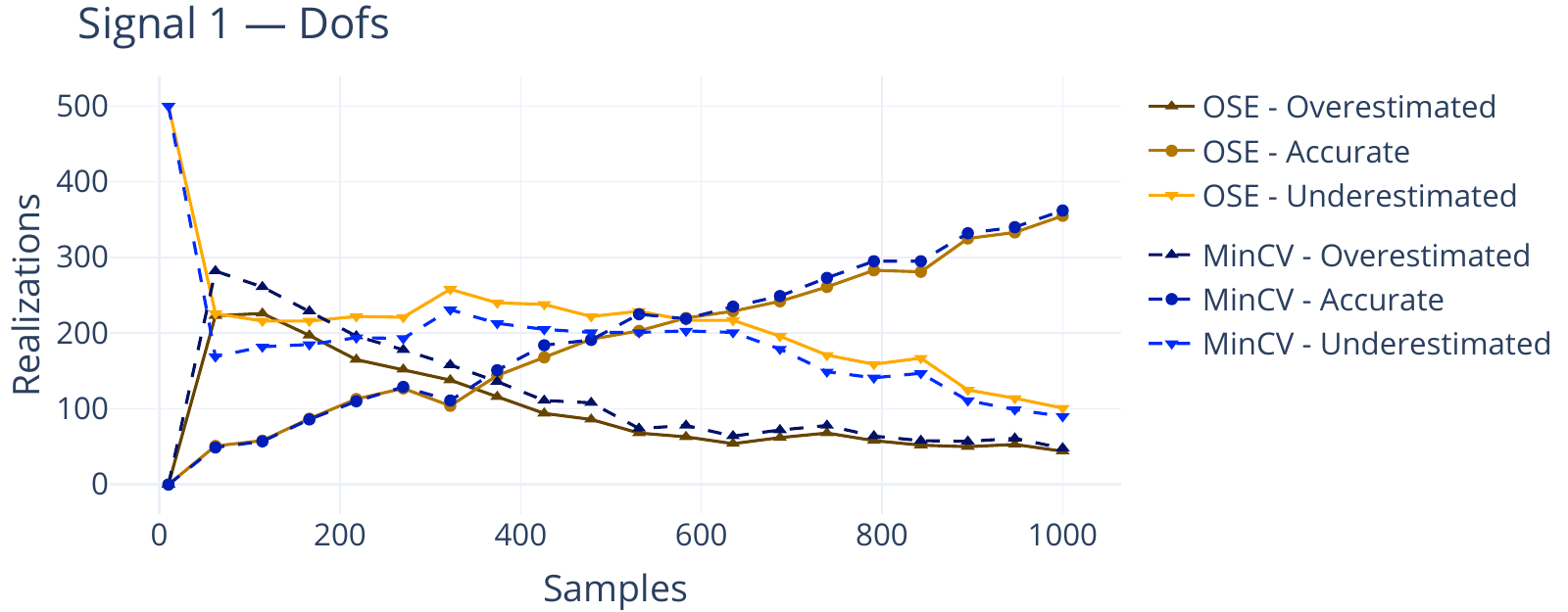}
      \includegraphics[width=0.49\textwidth]{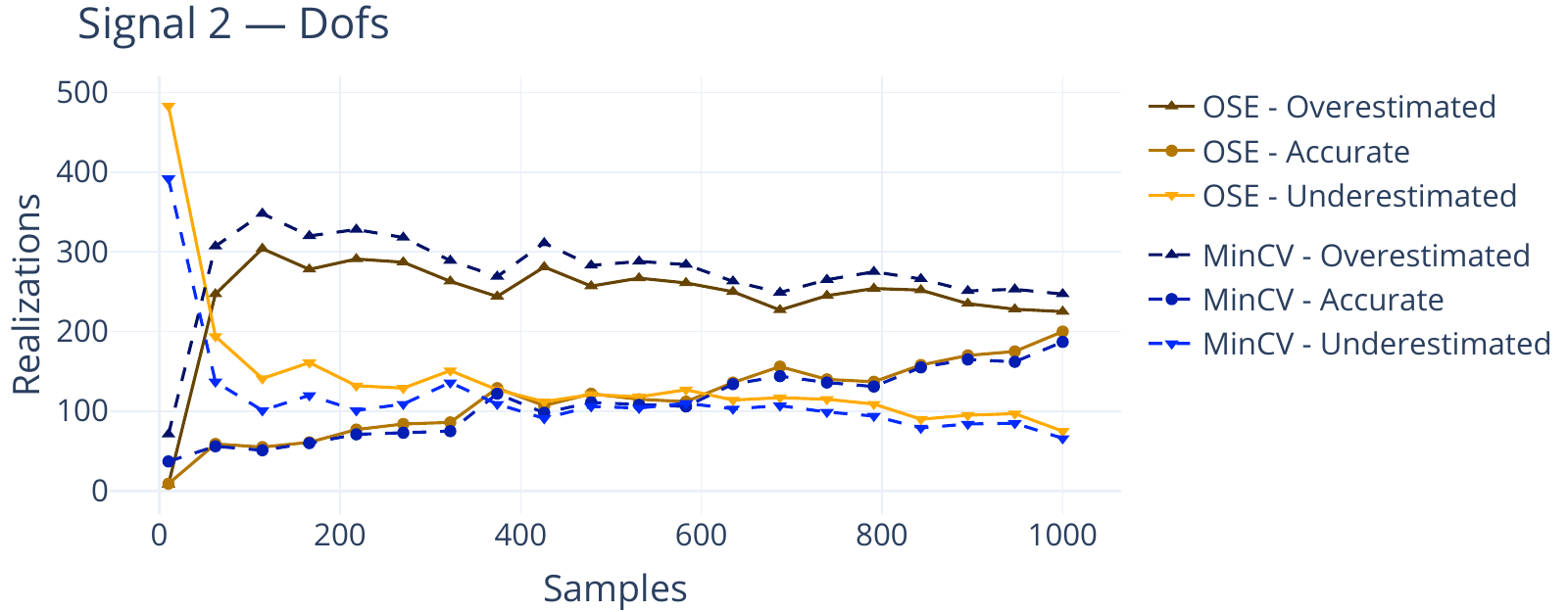}
      %\caption{Count of realizations where dofs are underestimated, accurately estimated, or overestimated vs number of samples.}
    \end{subfigure}
    \\
    \begin{subfigure}{\textwidth}
      \includegraphics[width=0.49\textwidth]{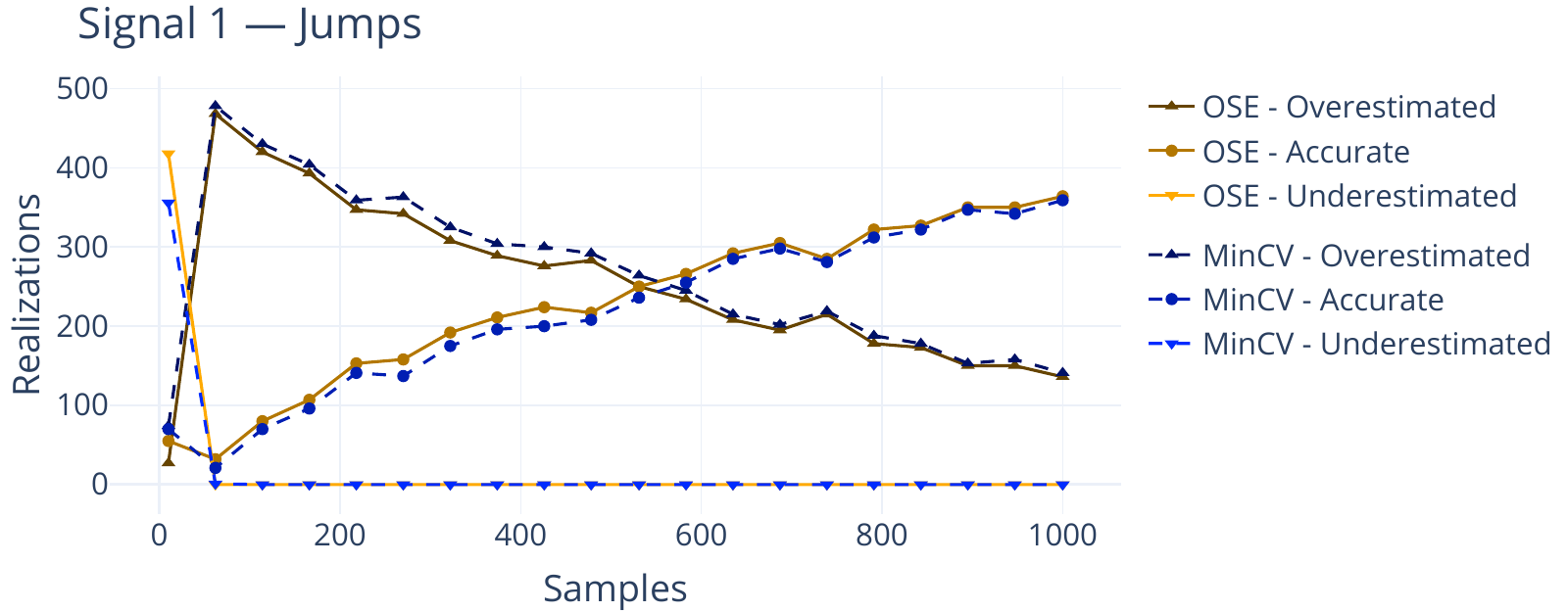}
      \includegraphics[width=0.49\textwidth]{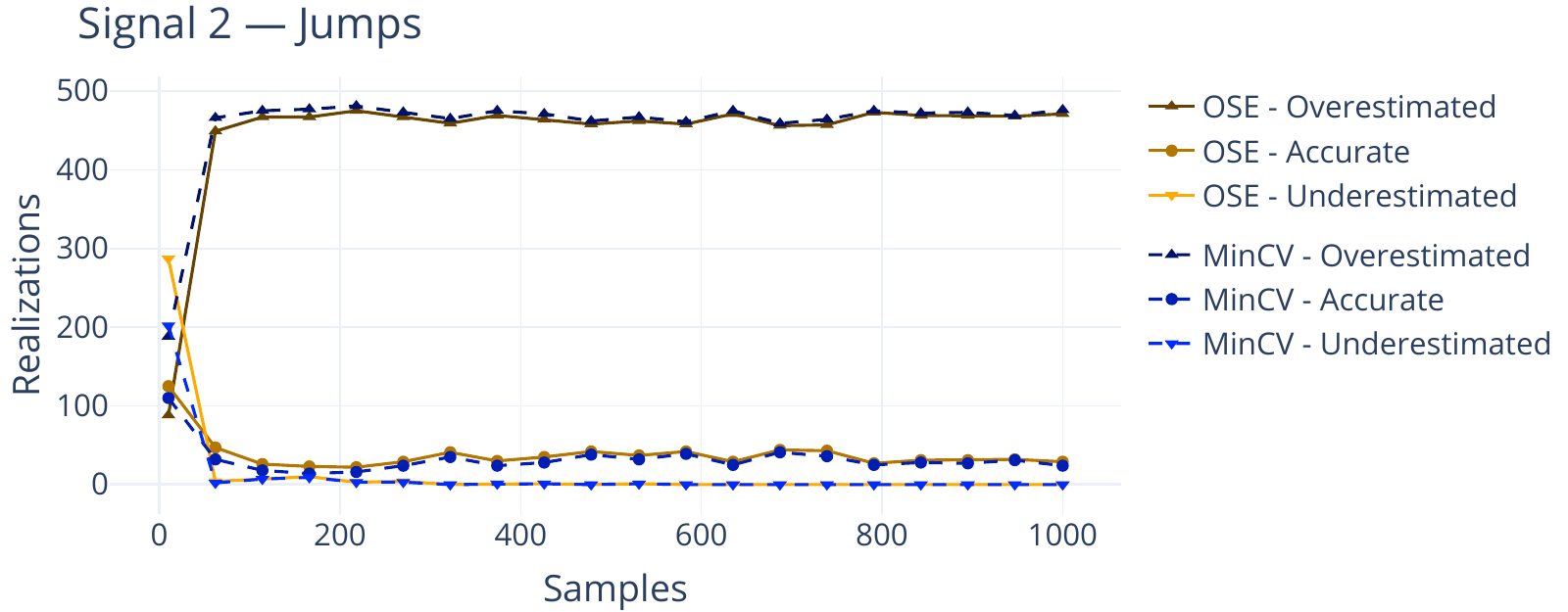}
      %\caption{Count of realizations where jumps are underestimated, accurately estimated, or overestimated vs number of samples.}
    \end{subfigure}
    \caption{Results of the experiment in Figure~\ref{main:fig:accuracy} of the main document
      for the $L^1$ based CV score.
    }
    \label{supp:fig:accuracy}
  \end{figure}

  \section{Further numerical comparison for the synthetic data}\label{sup:sec:more-experiments}

  Results of the Figures \ref{main:fig:synth-poly-len-noise} and \ref{main:fig:heavisine-len-noise}
  of the main document for the cpop algorithm \citep{fearnhead2024cpop}
  are shown in Figures \ref{fig:synth-poly-len-noise-cpop} and \ref{fig:heavisine-len-noise-cpop} respectively.

  \begin{figure}[!t]
    \newcommand{\myincludegraphics}[2]{
      \begin{subfigure}[b]{0.45\textwidth}
        \centering
        \includegraphics[width=\textwidth, trim=0 0 150 0, clip]{#1}
        \caption{#2}
      \end{subfigure}
    }
    \myincludegraphics{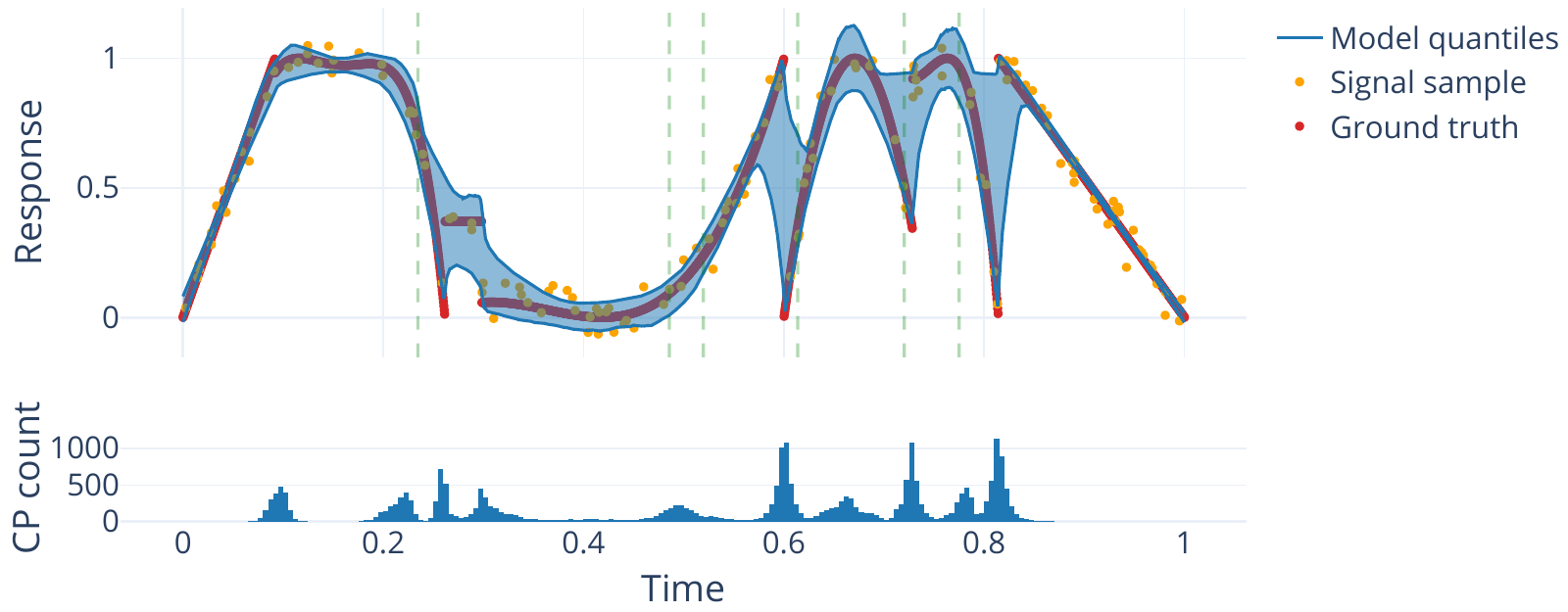}{$N = 150$}\hfill
    \myincludegraphics{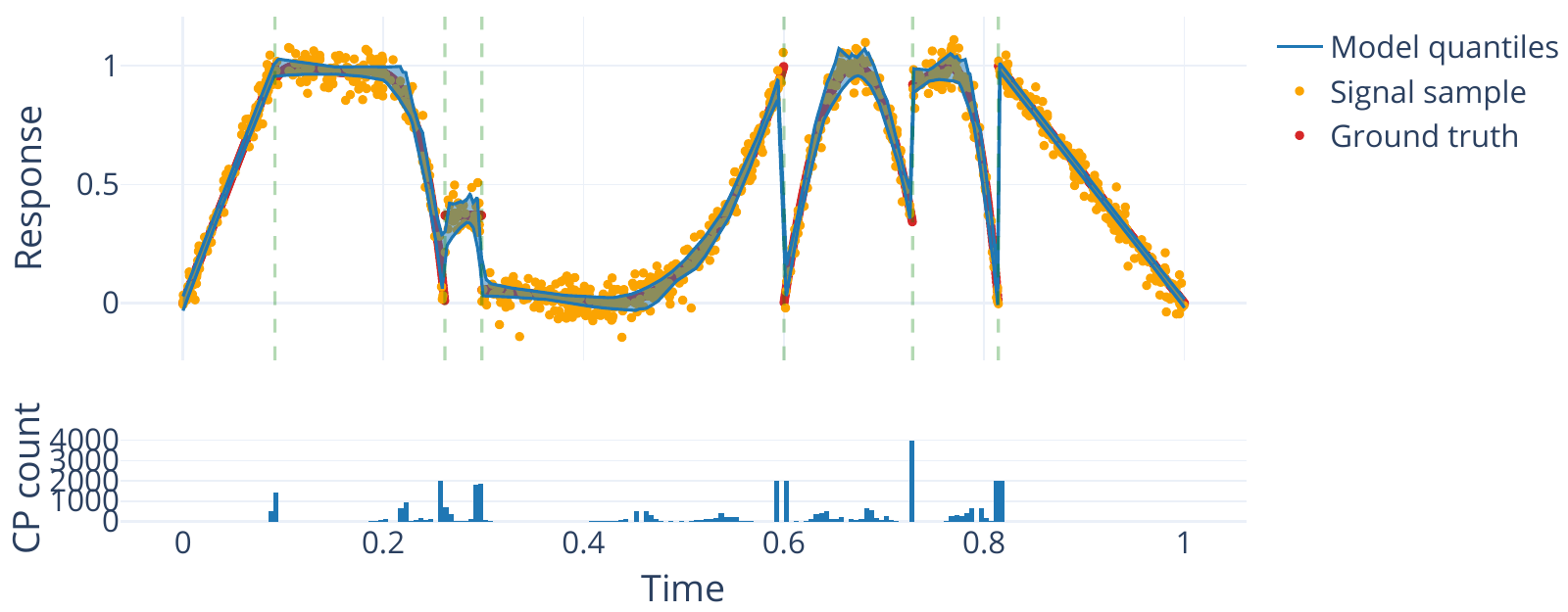}{$\sigma=0.5$}
    \myincludegraphics{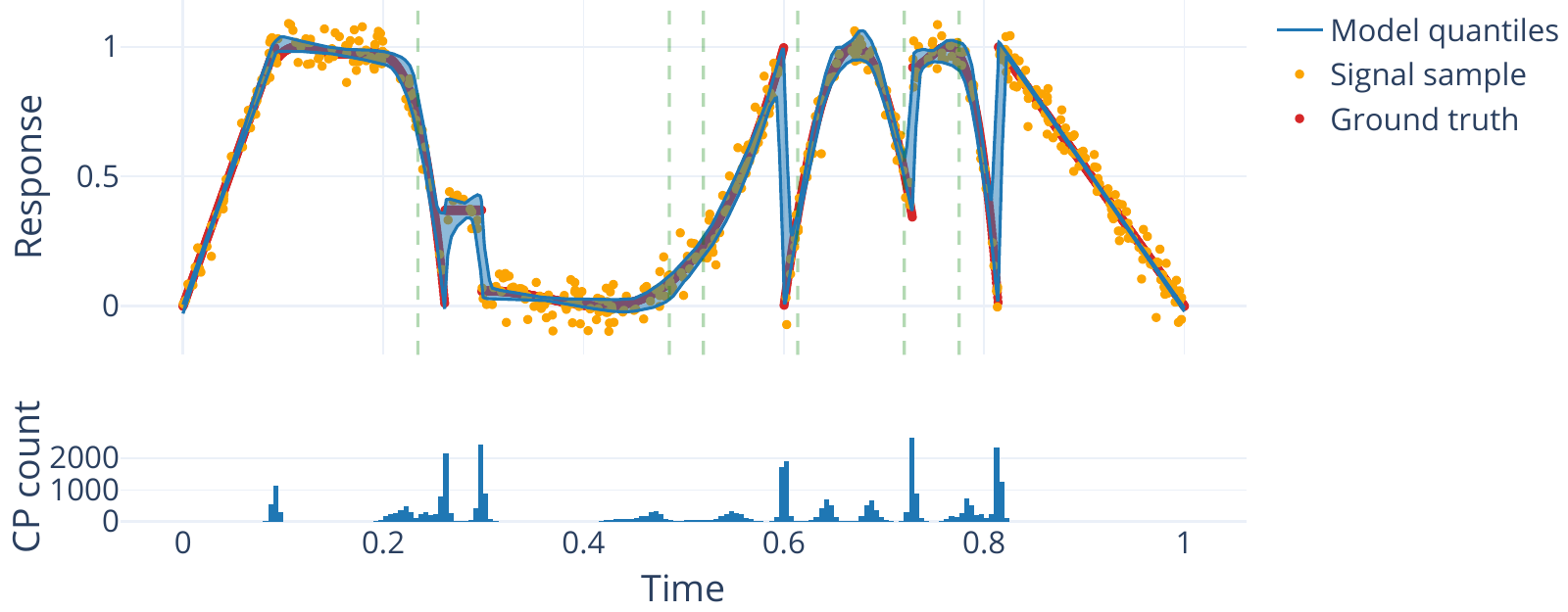}{$N = 500$}\hfill
    \myincludegraphics{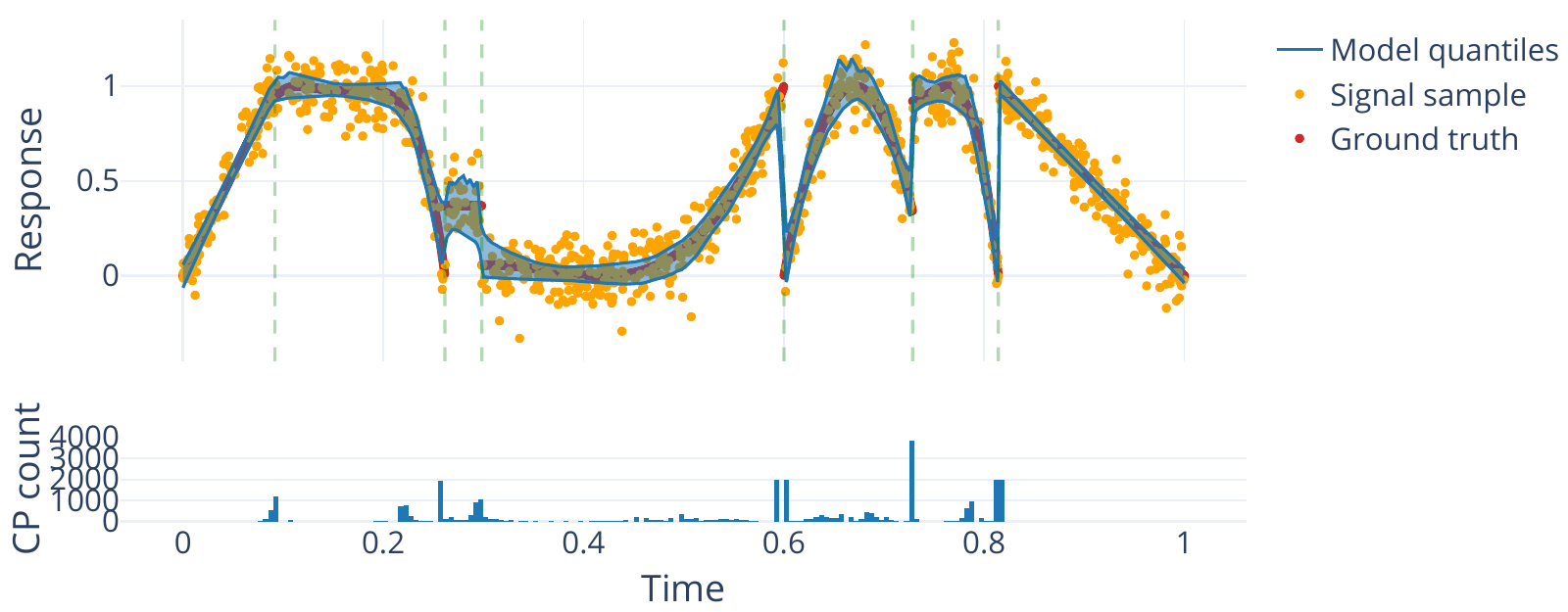}{$\sigma=1$}
    \myincludegraphics{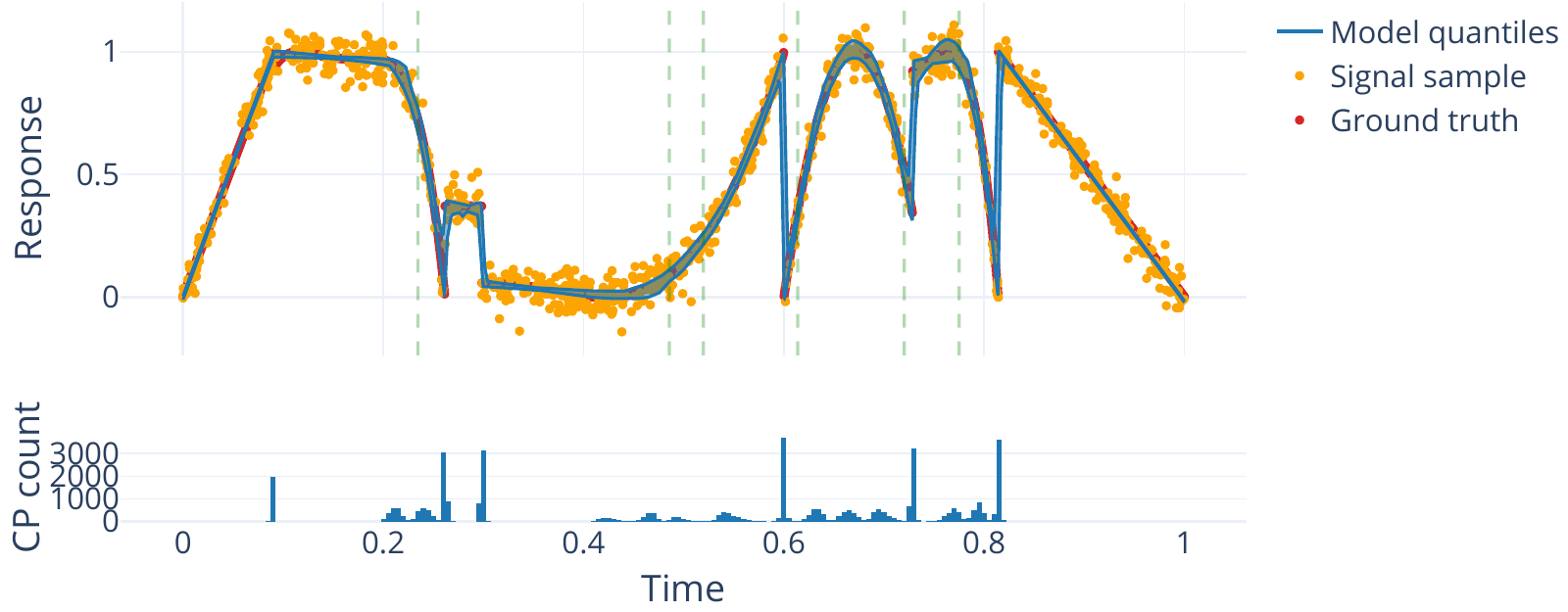}{$N=1000$}\hfill
    \myincludegraphics{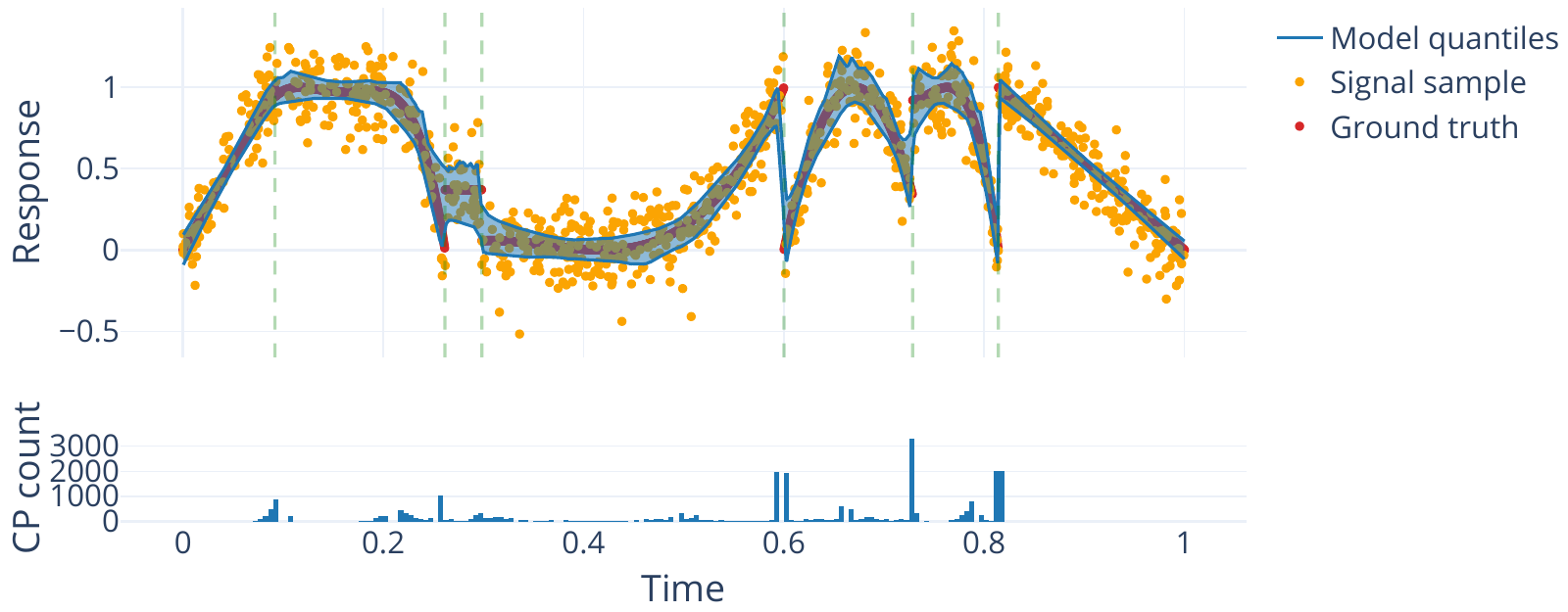}{$\sigma=1.5$}
    \caption{
      Results analogous to Figure~\ref{main:fig:synth-poly-len-noise} of the main document for
      the cpop algorithm \cite{fearnhead2024cpop}.
    }
    \label{fig:synth-poly-len-noise-cpop}
  \end{figure}

  \begin{figure}[!t]
    \newcommand{\myincludegraphics}[2]{
      \begin{subfigure}[b]{0.45\textwidth}
        \centering
        \includegraphics[width=\textwidth, trim=0 0 150 0, clip]{#1}
        \caption{#2}
      \end{subfigure}
    }
    \myincludegraphics{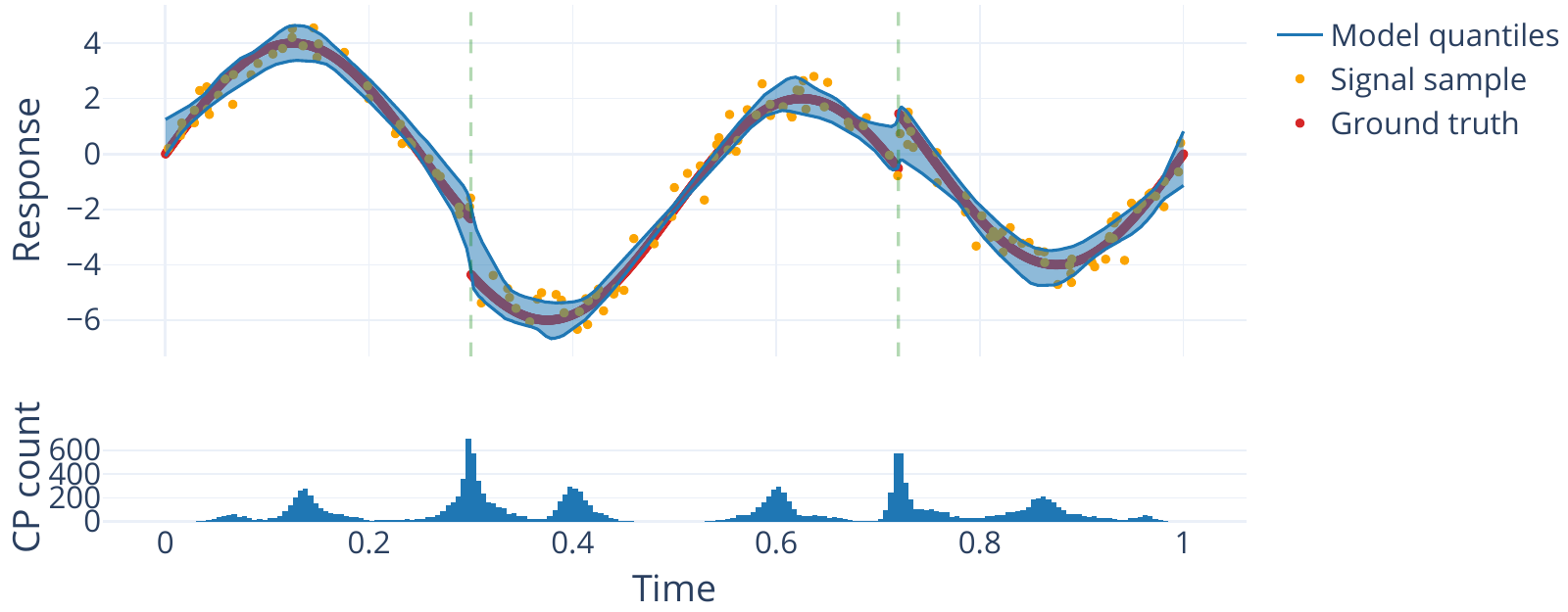}{$N = 150$}\hfill
    \myincludegraphics{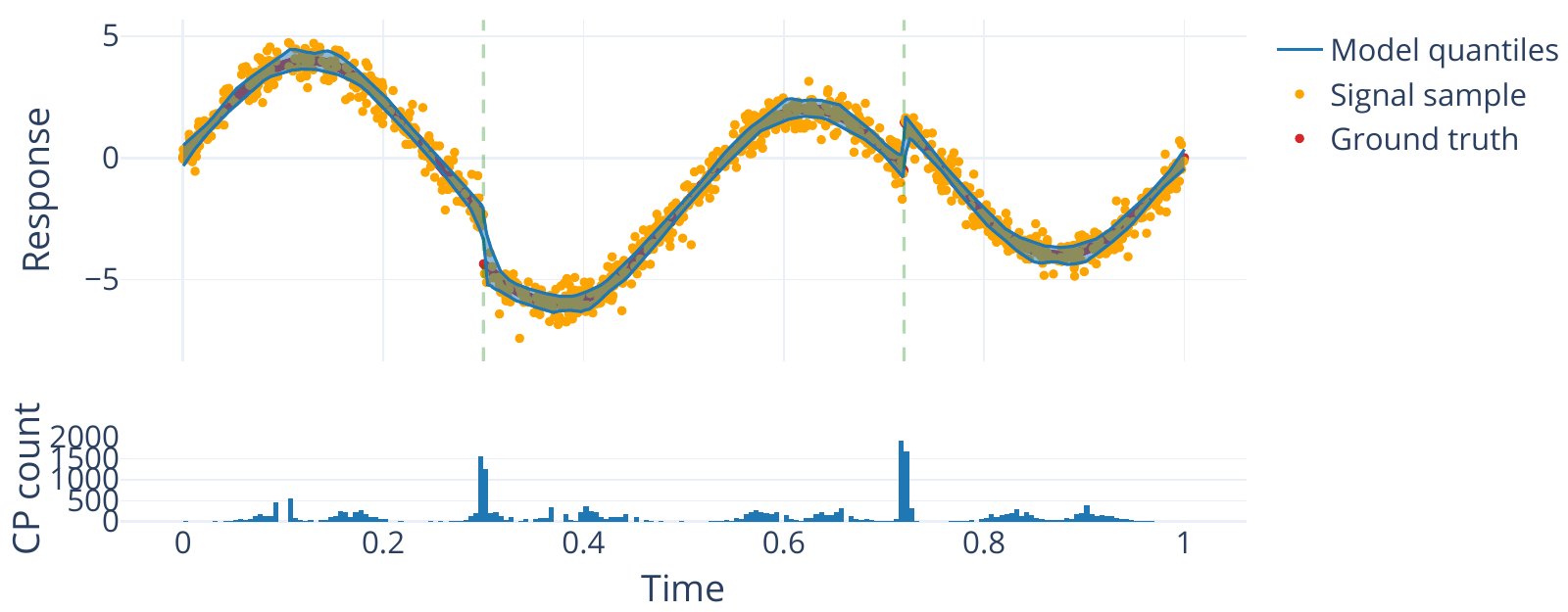}{$\sigma=0.5$}
    \myincludegraphics{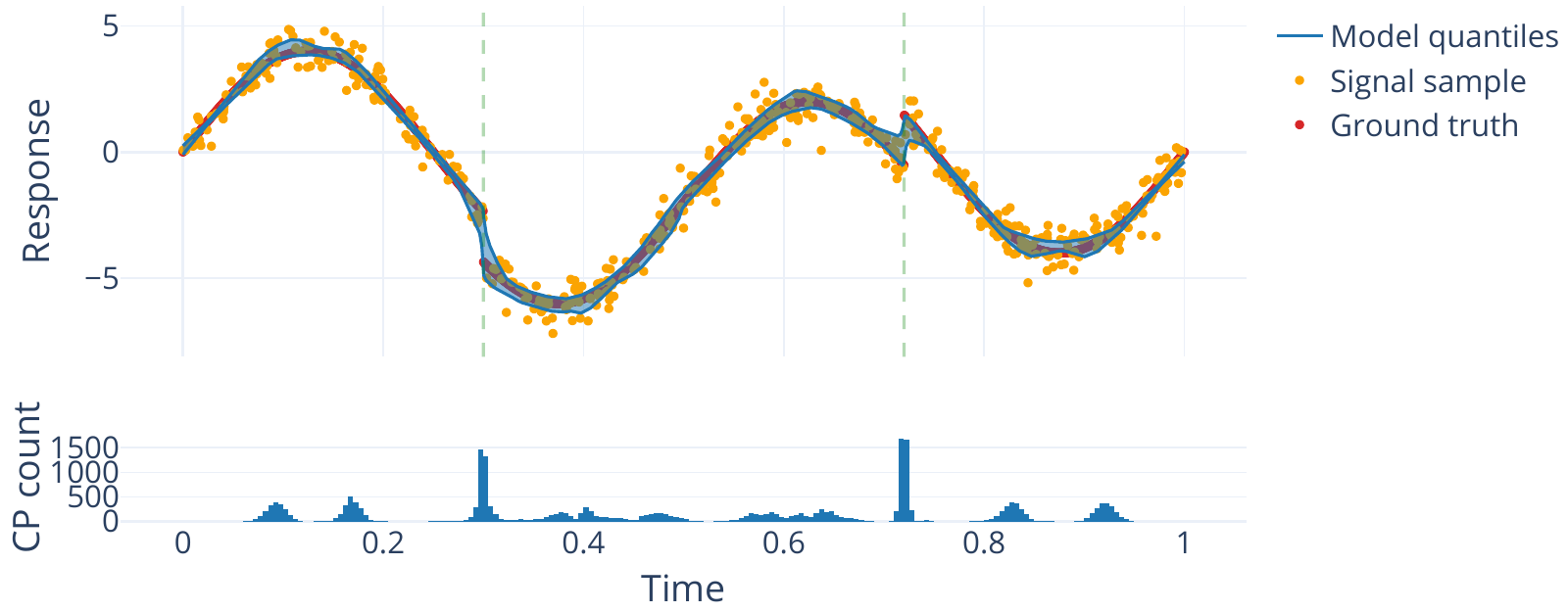}{$N = 500$}\hfill
    \myincludegraphics{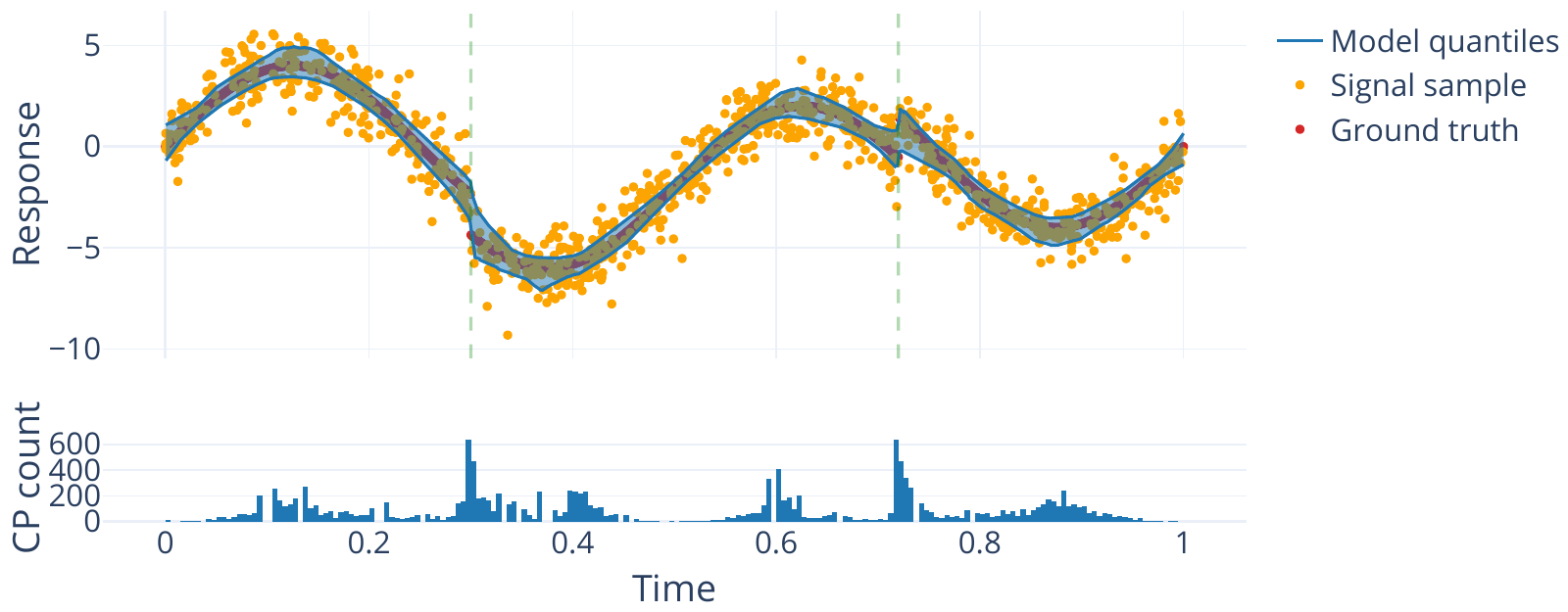}{$\sigma=1$}
    \myincludegraphics{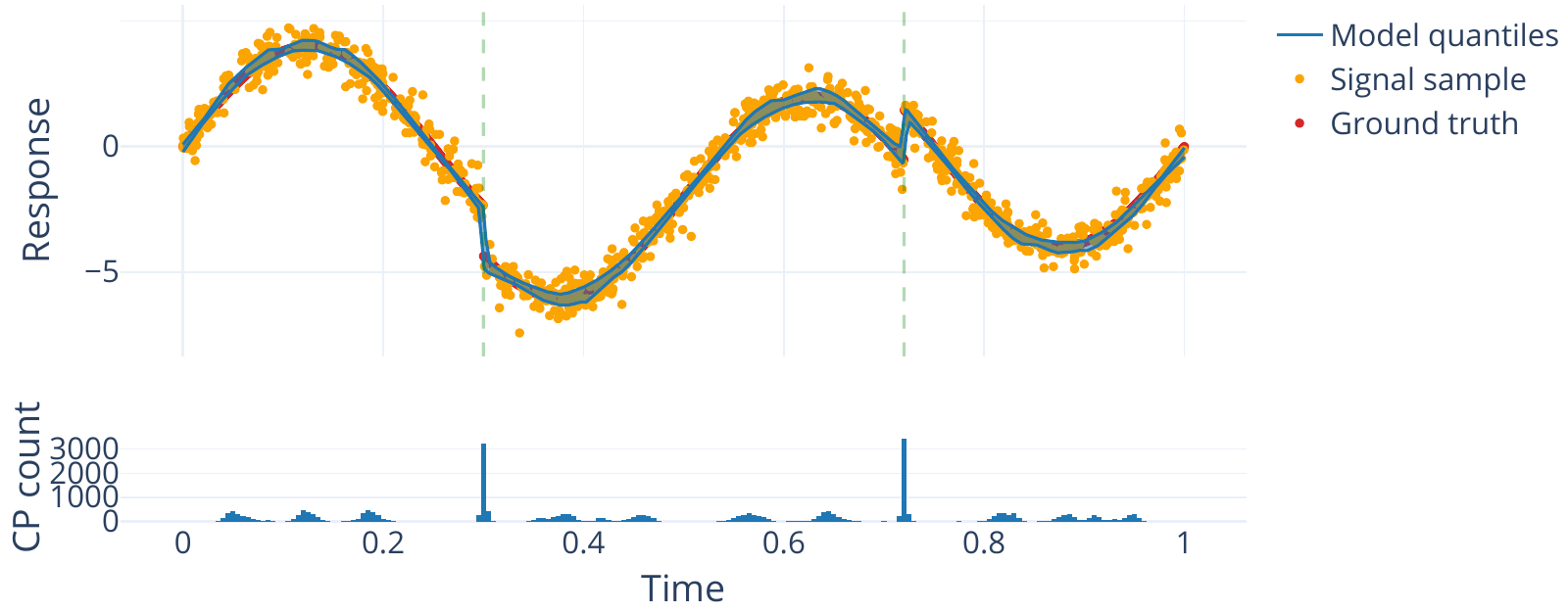}{$N=1000$}\hfill
    \myincludegraphics{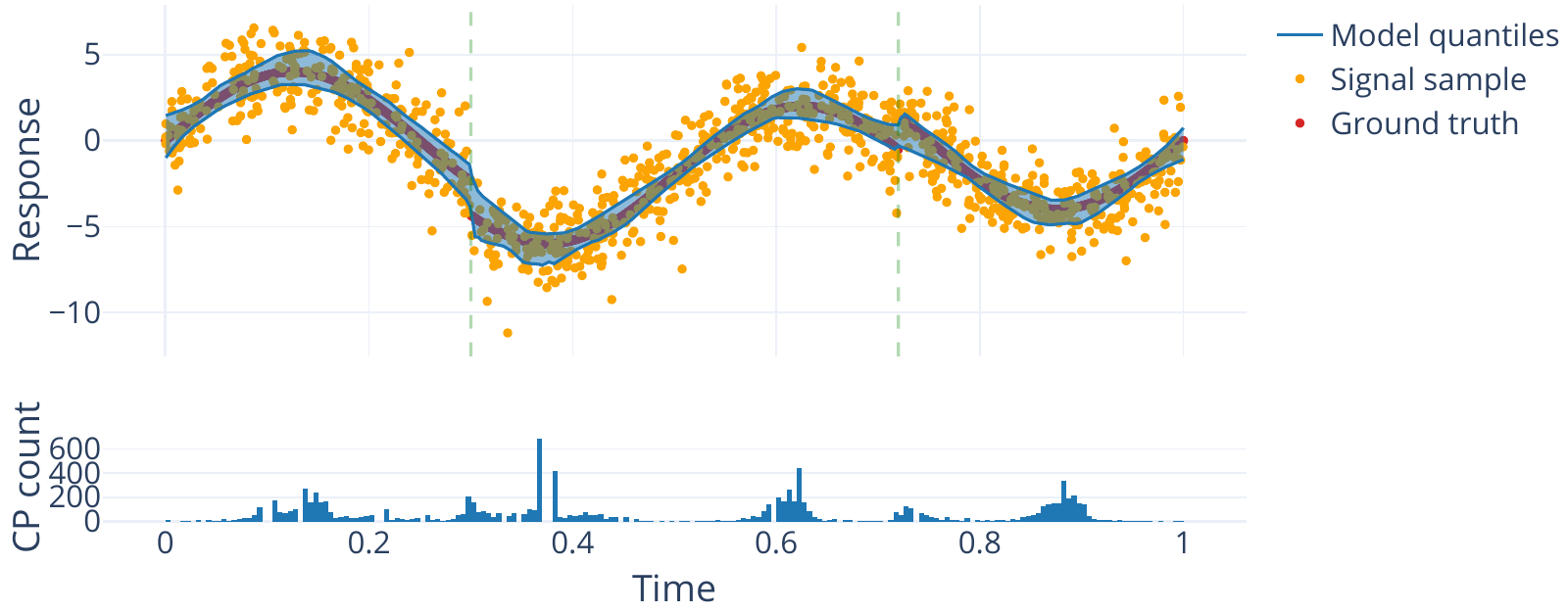}{$\sigma=1.5$}
    \caption{
      Results analogous to Figure~\ref{main:fig:heavisine-len-noise} of the main document for
      the cpop algorithm \cite{fearnhead2024cpop}.
    }
    \label{fig:heavisine-len-noise-cpop}
  \end{figure}

  \printbibliography[heading=myheadingsupplementary]

  \end{document}